\documentclass{amsart}
\pdfoutput=1

\usepackage{xargs}                      
\usepackage[pdftex,dvipsnames]{xcolor}  
\usepackage[]{todonotes}
\newcommandx{\TODO}[2][1=]{\todo[linecolor=red,backgroundcolor=red!25,bordercolor=red,#1]{#2}}

\usepackage{mhsymb}
\usepackage{mhequ}
\usepackage{mhfig}
\usepackage{geometry}

\usepackage{amsfonts}
\usepackage{amsmath}
\usepackage{amssymb}
\usepackage{mathtools}
\usepackage{color}
\usepackage{amsthm}
\usepackage{ifpdf}
\usepackage{todonotes}
\usepackage{mathrsfs}
\usepackage{enumerate}
\usepackage{esint}
\usepackage[pdftex,
  	bookmarks=true,
	unicode=true,
	pdfborder={0 0 0},
	pdfstartview={FitV},
	pdfpagelayout={TwoPageRight},
	pdftitle={},
	pdfauthor={Christian Bayer, Peter Friz, Paul Gassiat, J\"org Martin, Benjamin Stemper},
	colorlinks=false,
	linktoc=all,
	plainpages=false,
	pdfpagelabels,
	pdffitwindow=true]{hyperref}

\usepackage{microtype}
\usepackage{booktabs}
\usepackage{tikz}
\usepackage{pgfplots}
\usepackage{color}

\usepackage[ruled, linesnumbered]{algorithm2e}
\usepackage{caption}

\usetikzlibrary{shapes.misc}
\usetikzlibrary{shapes.symbols}
\usetikzlibrary{snakes}
\usetikzlibrary{decorations}

\def\symbol#1{\textcolor{symbols}{#1}}

\definecolor{symbols}{rgb}{0.1,0.35,1}

\ifpdf
  
\else
  
\fi
\theoremstyle{plain}
\newtheorem{theorem}{Theorem}[section]
\newtheorem{corollary}[theorem]{Corollary}
\newtheorem{lemma}[theorem]{Lemma}

\theoremstyle{definition}

\newtheorem{definition}[theorem]{Definition}
\newtheorem{example}[theorem]{Example}
\theoremstyle{remark}
\newtheorem{remark}[theorem]{Remark}

\newcommand{\RR}{\ensuremath{\mathbb{R}}}
\newcommand{\dd}{\ensuremath{\mathrm{d}}}

\newcommand{\vertiii}[1]{{\vert\kern-0.25ex\vert\kern-0.25ex\vert #1 
    \vert\kern-0.25ex\vert\kern-0.25ex\vert}}

\newcommand{\pa}{\partial}

\newcommand{\norm}[1]{\left\lVert#1\right\rVert}

\newcommand{\floor}[1]{\left\lfloor#1\right\rfloor}

\renewcommand{\bar}{\overline}
\renewcommand{\tilde}{\widetilde}
\renewcommand{\hat}{\widehat}

\numberwithin{equation}{section}

\def\1{{\mhpastefig{root}}}
\def\2{{\mhpastefig[2/3]{tree11}}}
\def\3{{\!\mhpastefig[1/2]{tree12}}}
\def\9{{\mhpastefig[1/2]{tree21}}}
\def\4{{\mhpastefig[1/2]{tree22}}}
\def\5{{\!\mhpastefig[1/2]{tree112}}}
\def\6{{\!\mhpastefig[1/2]{tree112b}}}
\def\8{{\mhpastefig[1/2]{tree121}}}
\def\7{{\mhpastefig[1/2]{tree121b}}}

\def\${|\!|\!|}
\definecolor{darkred}{rgb}{0.9,0.1,0.1}

\tikzset{
	bdot/.style={circle,fill=black,draw=black,inner sep=0pt,minimum size=1pt},
	dot/.style={circle,fill=symbols,draw=symbols,inner sep=0pt,minimum size=0.4pt},
	xi/.style={circle,fill=symbols!10,draw=symbols,inner sep=0pt,minimum size=1.5pt},
	xi_old/.style={circle,fill=symbols!10,draw=symbols,inner sep=0pt,minimum size=1.2mm},
	ddd/.style={draw=black,dash pattern=on 0.4pt off 0.4pt, dash phase=0.4pt},
	Ip/.style={draw=symbols},
	I/.style={draw=symbols,
		decorate, decoration={zigzag,amplitude=0.4pt,segment length = 0.6pt,pre length=0.5pt,post length=0.5pt}},
	>=stealth,
	}
\makeatletter
\def\DeclareSymbol#1#2#3{\expandafter\gdef\csname MH@symb@#1\endcsname{\tikz[baseline=#2,scale=0.15,draw=symbols]{#3}}}
\def\<#1>{\csname MH@symb@#1\endcsname}
\makeatother

\DeclareSymbol{X}{-2.4}{\node[xi] {};}
\DeclareSymbol{1}{0}{\draw[white] (-.5,0) -- (.5,0); \draw[Ip] (0,0) node[dot] {} -- (0,1.25); \node[xi] at (0,1.4) {};}
\DeclareSymbol{d}{0}{\draw[white] (-.5,0) -- (.5,0); \draw[I] (0,0) node[dot] {} -- (0,1.25); \node[xi] at (0,1.4) {};}
\DeclareSymbol{2}{0}{\node(root)[dot] at (0,0) {}; \draw[Ip] (-0.5,1.4) node[xi] {} -- (root); \draw[Ip] (0.5,1.4) node[xi] {} -- (root);}

\DeclareSymbol{20}{-3}{\node(root)[dot] at (0,0) {}; \draw[Ip] (root) -- (0,-1) node[dot] {};\draw[Ip] (-.7,1) node[xi] {} -- (root); \draw[Ip](.7,1) node[xi] {} -- (root);}
\DeclareSymbol{2d}{-3}{\node(root)[dot] at (0,0) {}; \draw[I] (root) -- (0,-1) node[dot] {};\draw[Ip] (-.7,1) node[xi] {} -- (root); \draw[Ip] (.7,1) node[xi] {}--  (root);}
\DeclareSymbol{21}{-3}{\node(root)[dot] at (0,-1) {}; \draw[Ip] (-1,1) node[xi] {} -- (-0.5,0); \draw[Ip] (root) -- (-0.5,0); \draw[Ip] (root) -- (1,0) node[xi] {};\draw[Ip] (-0.5,0) node[dot] {} -- (0.5,1) node[xi] {};}

\DeclareSymbol{40}{-3}{\node(right)[dot] at (1.2,0) {};\node(left)[dot] at (-0.4,0) {};\node(root)[dot] at (0.4,-1) {};\draw[Ip] (-.8,1) node[xi] {} -- (left); \draw[Ip] (0,1) node[xi] {} -- (left);\draw[Ip] (.8,1) node[xi] {} -- (right); \draw[Ip] (right) -- (1.6,1) node[xi] {};\draw[Ip] (left) -- (root); \draw[Ip] (right) -- (root);}

\DeclareSymbol{211}{2}{\draw[Ip] (-1.2,2.4) node[xi] {} -- (-0.8,1.6)   node[dot] {} -- (-0.4,0.8)   node[dot] {} -- (0,0)  node[dot] {} -- (0.8,0.8) node[xi] {};\draw[Ip] (-0.4,0.8) -- (0.4,1.6) node[xi] {};\draw[Ip] (-0.8,1.6) -- (0,2.4) node[xi] {};}

\DeclareSymbol{210}{2}{\draw[Ip] (-1.2,2.4) node[xi] {} -- (-0.8,1.6)  node[dot] {} -- (-0.4,0.8)  node[dot] {} -- (0,0)  node[dot] {};\draw[Ip] (-0.4,0.8) -- (0.4,1.6) node[xi] {};\draw[Ip] (-0.8,1.6) -- (0,2.4) node[xi] {};}
\DeclareSymbol{21d}{2}{\draw[Ip] (-1.2,2.4) node[xi] {} -- (-0.8,1.6)  node[dot] {} -- (-0.4,0.8); \draw[I] (-0.4,0.8) node[dot] {} -- (0,0)  node[dot] {};\draw[Ip] (-0.4,0.8) -- (0.4,1.6) node[xi] {};\draw[Ip] (-0.8,1.6) -- (0,2.4) node[xi] {};}

\DeclareSymbol{10}{1}{\draw[Ip] (0,2) node[xi] {} -- (-1,1) node[dot] {} -- (0,0) node[dot] {};}
\DeclareSymbol{1d}{1}{\draw[Ip] (0,2) node[xi] {} -- (-1,1); \draw[I] (-1,1) node[dot] {} -- (0,0) node[dot] {};}
\DeclareSymbol{11}{1}{\draw[Ip] (0,2) node[xi] {} -- (-0.8,1) node[dot] {} -- (0,0) node[dot] {} -- (0.8,1) node[xi] {};}

\input{tcilatex}

\begin{document}
\title[Regularity structures \& rough vol]{A regularity structure for rough volatility } 

\author{C. Bayer, P. K. Friz, P. Gassiat, J. Martin, B. Stemper}
\address{WIAS Berlin, TU and WIAS Berlin, Paris Dauphine University, HU
Berlin, TU and WIAS Berlin}

\begin{abstract}
A new paradigm recently emerged in financial modelling: rough (stochastic) volatility, first observed by Gatheral et al. in high-frequency data, subsequently derived within market microstructure models, also turned out to
 capture parsimoniously key stylized facts of the entire implied volatility surface, including extreme skews that were thought to be outside the scope of stochastic volatility. On the mathematical side, Markovianity and, partially, 
 semi-martingality are lost. In this paper we show that Hairer's regularity structures, a major extension of rough path theory, which caused a revolution in the field of stochastic partial differential equations, also provides a new and powerful tool to analyze rough volatility models. \\
\center{{\it{Dedicated to Professor Jim Gatheral on the occasion of his 60th birthday. }}} 
\end{abstract}


\date{\today }
\maketitle
\tableofcontents



\section{Introduction}

\label{sec:Introduction}

We are interested in stochastic volatility (SV)\ models given in It\^{o} differential form 
\begin{equation}
dS_{t}/S_{t}=\sigma _{t}dB_{t}\equiv \sqrt{v_{t}\left( \omega \right) }%
dB_{t} \ .  \label{equ:S}
\end{equation}%
Here, $B$ is a standard Brownian motion and $\sigma _{t}$ (resp. $v_{t}$)
are known as \textit{stochastic volatility} (resp. \textit{variance})\
process. Many classical Markovian asset price models fall in this framework,
including Dupire's local volatility model, the SABR -, Stein-Stein - and 
Heston model. In all named SV model, one has Markovian dynamics for the
variance process, of the form%
\begin{equation}
dv_{t}=g\left( v_{t}\right) dW_{t}+h\left( v_{t}\right) dt;  \label{dvIntro}
\end{equation}%
constant correlation $\rho :=d\left\langle B,W\right\rangle _{t}/dt$ is
incorporated by working with a 2D standard Brownian motion $\left( W,\bar{W}%
\right) $,%
\begin{equation*}
B:=\rho W+\bar{\rho}\bar{W}\equiv \rho W+\sqrt{1-\rho ^{2}}\bar{W}.
\end{equation*}%
This paper is concerned with an important class of non-Markovian\
(fractional) SV models, dubbed \textbf{rough volatility (RV) models}, in which
case $\sigma _{t}$ (equivalently: $v_{t} \equiv \sigma^2_t$) is modelled via a fractional
Brownian motion (fBM) in the regime $H\in \left( 0,1/2\right) $.\footnote{%
Volatility is not a traded asset, hence its
non-semimartingality (when $H\neq 1/2$) does not imply arbitrage.}
The terminology "rough"\ stems from the fact that in such models
stochastic volatility (variance) sample paths are $H^-$-H\"older, hence ``rougher'' than Brownian
paths. Note the stark contrast to the idea of "trending" fractional
volatility, which amounts to take $H>1/2$. The evidence for the rough regime
(recent calibration suggest $H$ as low as $0.05$) is now overwhelming - both
under the physical and the pricing measure, see e.g.  \cite{ALV07, Fuk11,
Fuk17, GJR14,BFG16, FZ17, MT16}. Much attention in theses reference has in fact been given to "simple"
rough volatility models, by which we mean models of the form

\begin{eqnarray}
\sigma_t &:=&f(\hat{W}_{t}) \ \ \  \ \ \  \dots \ \ \text{ ``simple rough volatility (RV)''} \  \label{eq:sigmasimple} \\
\hat{W}_{t} &=&\int_{0}^{t}K\left( s,t\right) dW_{s}  \ ; \\
& \text{with} & K\left( s,t \right) = {\sqrt{2H}}\left\vert t-s\right\vert ^{H-1/2}\mathbf{1}_{t>s} \ , \ \ \ H \in (0,1/2).
\end{eqnarray}%
%
In other words, volatility is a function of a fractional Brownian motion, with (fixed) Hurst parameter.%
\footnote{Following \cite{BFG16} we work with the Volterra- or Riemann-Liouville fBM, but other choices such as 
the Mandelbrot van Ness fBM, with suitably modified kernel $K$, are possible.}
Note that, in contrast even to classical SV\
models, the stochastic volatility is explicitly given, and no rough /
stochastic differential equation needs to be solved (hence "simple"). Rough volatility not only
provides remarkable fits to both time series and option pricing problems, it also has a market microstructure justification: starting with a Hawkes process model, Rosenbaum and coworkers \cite{EEFR16x, EER16x,EER17x} 
find in the scaling limit $f,g,h$ such that 
\begin{eqnarray}
\sigma_t &:=&f(\hat{Z}_{t}) \ \ \  \ \ \  \dots \ \ \text{ ``non-simple rough volatility (RV)''} \  \label{eq:sigmanotsimple} \\
Z_t &=&  z + \int_0^t K(s,t)  g(Z_s) ds +  \int_0^t K(s,t)  h(Z_s) dW_s \label{eq:Voltintro}  \ ,
\end{eqnarray}%
with stochastic Volterra dynamics that provide a natural generalization of simple rough volatility.
%
%

\subsection{Markovian stochastic volatility models} 
\label{sec:mark-stoch-volat}

For comparison with rough volatility, Section \ref{sec:RVintro} below, we first mention a selection of tools and methods well-known for {\it Markovian} SV models. 
\begin{itemize}
\item PDE methods are ubiquitous in (low-dimensional) pricing problems, as are
\item Monte Carlo methods, noting that knowledge of strong (resp. weak) rate $1/2$ (resp. $1$) is the grist in the mills of modern multilevel methods (MLMC);
\item Quasi Monte Carlo (QMC) methods are widely used; related in spirit we have the Kusuoka--Lyons--Victoir cubature approach, popularized in the form of Ninomiya--Victoir (NV) splitting scheme, nowadays available in standard software packages;
\item Freidlin--Wentzell theory of small noise large deviations is essentially immediately applicable, as are various ``strong`` large deviations (a.k.a. exact asymptotics) results, used e.g. the derive the famous SABR formula. 
\end{itemize}
\noindent  For several reasons it can be useful to write model dynamics in {\it Stratonovich form}: From a PDE perspective, the operators then take sum-square form which can be exploited in many ways (H\"ormander theory, naturally linked to Malliavin calculus ...). 
From a numerical perspective, we note that the cubature / NV scheme \cite{NV} also requires the full dynamics to be rewritten in Stratonovich form. In fact, viewing NV as level-$5$ cubature, in sense of \cite{LV04}, its level-$3$ 
simplification is nothing but the familiar Wong-Zakai approximation result for difffusions. Another financial example that requires a Stratonovich formulation comes from interest rate model validation \cite{DMP07}, based on the 
Stroock--Varadhan support theorem. 
We further note, that QMC (e.g. Sobol') works particularly well if the noise has a multiscale decomposition, as obtained by interpreting a (piece-wise) linear Wong-Zakai approximation, as Haar wavelet expansion of the driving white noise. 

\subsection{Complications with rough volatility}   \label{sec:RVintro}

Due to loss of Markovianity, PDE methods are not applicable, and neither are (off-the-shelf) Freidlin--Wentzell large deviation estimates (but see \cite{FZ17}). Moreover, rough volatility is not a semi-martingale, which complicates, to say the least, the use of several established stochastic analysis tools. In particular, {\it rough volatility admits no Stratonovich form}. Closely related, one lacks a (Wong-Zakai type) approximation theory for rough volatility. To see this, focus on the ``simple'' situation, that is  
(\ref{equ:S}), (\ref{eq:sigmasimple}) so that 
\begin{equation}
S_{t}/S_{0}=\mathcal{E}\left( \int_{0}^{\cdot }f\left(  \hat{W}_s  \right) dB_{s}\right) (t) \ .    \label{equ:SundersimpleRV}
\end{equation}%
Inside the (classical) stochastic exponential $\mathcal{E}(M) (t) = \exp (M_t - \tfrac{1}{2} [M]_t)$ we have the martingale term
\begin{equation}
\int_0^t f(\hat{W})dB=\rho\underbrace{ \int_0^t f ( \hat{W} )  dW_{t}}+\bar{\rho}\int_0^t
f ( \hat{W} ) d\bar{W}_{t}  \label{equ:int1}
\end{equation}%
and, in essence, the trouble is due to underbraced, innocent looking It\^o-integral. 
Indeed, any naive attempt to put it in Stratonovich form, 
\begin{equation}
 \text{``} \int_0^t f ( \hat{W})  \circ dW  :=  \int_0^t f ( \hat{W})   dW + (\text{It\^o-Stratonovich correction}) \text{ ''}
\end{equation}
or, in the spirit of Wong-Zakai approximations, 
\begin{equation}
\text{``}  \int_0^t f ( \hat{W} ) \circ W   :=  \lim_{\varepsilon \to 0} \int_0^t f( \hat{W}^{\varepsilon} ) dW^{\varepsilon}  \text{ ''}
\end{equation}
must fail whenever $H<1/2$. The It\^o-Stratonovich correction is given by the quadratic covariation, 
defined (whenever possible) as the limit, in probability,  of
\begin{equation}
    \sum_{[u,v] \in \pi}  (f(\hat W_v) - f(\hat W_u))(W_v - W_u),
\end{equation}
along any sequence $(\pi^n)$ of partitions with mesh-size tending
to zero. But, disregarding trivial situations, this limit does not exist. For
instance, when $f(x)=x$ fractional scaling immediately gives divergence (at
rate $H-1/2$) of the above bracket approximation. This issues also arises in
the context of option pricing which in fact is readily reduced (Theorem \ref{thm:rPricing} and 
Section~\ref{sec:ChristianBen}) to the sampling of stochastic integrals of the
afore-mentioned type, i.e. with integrands on a fractional scale. All theses problems remain present, of course, 
for the more complicated situation of ``non-simple'' rough volatility (Section \ref{sec:Volt}) .

\bigskip 
\subsection{Description of main results}
\label{sec:main-results}

With motivation from singular SPDE theory, such as Hairer's work on KPZ \cite{Hai13} and the Hairer-Pardoux ``renormalized'' Wong-Zakai theorem \cite{HP15},  we provide the closest there is to a satisfactory approximation theory for rough volatility. This starts with the remark that rough path theory, despite its very purpose to deal with low regularity paths, is not applicable 

\bigskip

To state our basic approximation results, write $\dot W^\eps \equiv \partial_t W^\eps$ for a suitable (details below) approximation at scale $\varepsilon$ to white noise, with induced approximation to fBM, denoted by $\hat W^\eps$. Throughout, the Hurst parameter $H \in (0, 1/2]$ is fixed and $f$ is a smooth function, such that (\ref{equ:SundersimpleRV}) is a (local) martingale, as required by modern financial theory.

\begin{theorem} \label{thmEasy} Consider simple rough volatility with dynamics $dS_t / S_t = f(\hat W_t) dB_t$, i.e. driven by Brownians $B$ and $W$ with constant correlation $\rho$.  There exist $\varepsilon$-peridioc functions $ \mathscr{C}^\varepsilon=  \mathscr{C}^\varepsilon(t)$, with diverging averages $C_{\varepsilon } \ ,$ such that a Wong-Zakai result holds of the form $\tilde S^\eps \to S$ in probability and uniformly on compacts, where
$$
     \partial_t \tilde S^\eps_t / S_t^\eps  =  f ( \hat W^\eps)  \dot B^\eps - \rho \mathscr{C}^\varepsilon(t) f' (\hat W^\eps) - \tfrac{1}{2}f^2 (\hat W^\eps)  \  ,  \ \  S^\eps_0 = S_0  .
$$
Similar results hold for more general (``non-simple'') RV models. 
\end{theorem}

\begin{remark} When $H=1/2$, this result is an easy consequence of It\^o-Stratonovich conversion formulae. In the case $H<1/2$ of interest, Theorem \ref{thmEasy} provides the interesting insight that genuine renormalization, in the sense of subtracting diverging quantities is required if and only if correlation $\rho$ is non-zero. This is the case in equity (and many other) markets \cite{BFG16}. Also note that naive approximations $S^\eps_t$, without subtracting the $\mathscr{C}^\varepsilon$-term, will in general diverge.
\end{remark}

\noindent  In order to formulate implications for option pricing, define the Black-Scholes pricing function 
\begin{equation}
C_{BS}\left( S_{0},K;\sigma ^{2}T\right)  :=   \mathbb{E}\left(
S_{0}\exp \left( \sigma \sqrt{T} Z -\frac{\sigma ^{2}}{2}T\right) -K\right) ^{+} \ ,
\end{equation}%
where $Z$ denotes a standard normal random variable. We then have

\begin{theorem} \label{thm:rPricing} With $ \mathscr{C}^\varepsilon=  \mathscr{C}^\varepsilon(t)$ as in Theorem \ref{thmEasy}, 
define the renormalized integral approximation, 
\begin{equation}  
 \ \tilde{\mathscr{I}}^\eps := \tilde{\mathscr{I}}^\eps_f (T) := \int_{0}^{T}f(\hat{W}^{\varepsilon})dW^{\varepsilon }- \int_{0}^{T} 
  \mathscr{C}^\varepsilon(t)  f^{\prime }( \hat{W}_{t}^{\varepsilon }) dt 
\end{equation}%
and also approximate total variance, 
$$
       {\mathscr{V}}^\eps :=  {\mathscr{V}}^\eps_f (T)  :=\int_0^T f^2 ( \hat{W}_{t}^{\varepsilon }) dt \ .
$$
Then the price of a European call option, under the pricing model
(\ref{equ:S}), (\ref{eq:sigmasimple}), struck at $K$ with time $T$ to
maturity, is given as 
\begin{equation*}  
   \lim_{\eps \to 0 } \mathbb{E} \left[ \Psi ( \tilde{\mathscr{I}}^\eps, \mathscr{V}^\eps )  \right] 
\end{equation*}
where
\begin{equation}
   \label{equ:Psi}
    \Psi \left( \mathscr{I},\mathscr{V} \right) 
     := C_{BS} \left( S_{0}\exp \left( \rho
\mathscr{I} -\frac{\rho ^{2}}{2} \mathscr{V}   \right), 
K, \bar{\rho}^{2}  \mathscr{V} \right)    \ .
\end{equation}
Similar results hold for more general (``non-simple'') RV models. 
\end{theorem}

From a {\bf mathematical perspective}, the key issue in proving the above theorems is to establish convergence of the {\it renormalized} approximate integrals
\begin{equation}  \label{equ:keyconvergence}
 \tilde{\mathscr{I}}^\eps =
 \int_{0}^{T}f(\hat{W}^{\varepsilon})dW^{\varepsilon }- \int_{0}^{T} 
  \mathscr{C}^\varepsilon(t)  f^{\prime }( \hat{W}_{t}^{\varepsilon }) dt \rightarrow (\text{It\^o-integral)}  .
\end{equation}%
It is here that we find much inspiration from singular SPDE theory, which also requires renormalized approximations for convergence to the correct It\^o-object. Specifically, we see that the theory of regularity structures \cite{Hai14}, which essentially emerged from rough paths and Hairer's KPZ analysis (see \cite{FH14} for a discussion and references),  is a very appropriate tool for us.
This adds to the existing instances of regularity structures (polynomials, rough paths, many singular SPDEs \ldots) an interesting new class of examples which on the one hand avoids all considerations related to spatial structure (notably multi-level Schauder estimates; cf. \cite[Ch.5]{Hai14}), yet comes with the genuine need for renormalization. 
In fact, since we do not restrict to mollifier approximations (this would rule out wavelet approximation of white noise!) our analysis naturally leads us to 
 {\it renormalization functions}. In case of mollifier approximations, i.e. $\dot W^\eps$ is the $\eps$-mollifciation obtained by convolution of $\dot W$ with a rescaled mollifier function, say $\delta^\varepsilon(x,y)=\varepsilon^{-1} \rho(\varepsilon^{-1}(y-x))$), which is the usual choice of Hairer and coworkers \cite{Hai13, Hai14, CH16}, the renormalization function turns out to constant (since $\dot W^\eps$ is still stationary); in this case
%
%
$$
       \mathscr{C}^\varepsilon(t) \equiv C_\varepsilon = c \eps^{H-1/2} 
$$
with $c = c(\rho)$ explicitly given as integral, cf. (\ref{eq:MollifiedVolterraKernel}). If, on the other hand, we consider a Haar wavelet approximation of white noise, very natural from a numerical point of view,
\footnote{Other wavelet choices are possible. In particular, in case of fractional noise, {\it Alpert-Rokhlin (AR) wavelets} have been suggested for improved numerical behaviour; cf. \cite{GBJ16} where this is attributed to a series of works of A. Majda and coworkers. A theoretical and numerical study of AR wavelets in the rough vol context is left to future work.}
 \begin{align}
		 \mathscr{C}^\varepsilon(t) =\frac{\sqrt{2H}}{H + 1/2} \frac{ |t-\lfloor t / \eps \rfloor \eps|^{H + 1/2} }{\eps}\ \ \  \text{ with mean } C_\eps = \frac{\sqrt{2H}}{(H+1/2)(H+3/2)}  \eps^{H-1/2}. 
 \end{align}
It is natural to ask if $\mathscr{C}^\varepsilon(t)$ can be replaced, after all, by its (since $H<1/2$: diverging) mean $C_\eps$. For $H>1/4$ the answer yes, with an interesting phase transition when $H=1/4$, cf. Section \ref{subsec:ApproximationTheory}. 


\bigskip

\noindent From a {\bf numerical simulation perspective}, Thereom \ref{thm:rPricing} is a step forward as it avoids
any sampling related to the other factor $\bar{W}$.\ A brute-force approach
then consists in simulating a scalar Brownian motion $W$, followed by
computing $\hat{W}=\int KdW$ by It\^o/Riemann Stieltjes approximations of $\left( \mathscr{I},%
\mathscr{V}\right) $. However, given the singularity of Volterra-kernel $K$,
this is not advisable and it is preferable to simulate the two-dimensional
Gaussian process $(W_{t},\hat{W}_{t}:0\leq t\leq T)$ with covariance readily
available. A remaining problem is that the rate of convergence%
\begin{equation*}
\sum f(\hat{W}_{s})W_{s,t} \rightarrow (\text{It\^o-integral)} \ ,
\end{equation*}%
with $\left[ s,t\right] $ taken in a partition of mesh-size $\sim 1/n$, is
very slow since $\hat{W}$ has little regularity when $H$ is small.
(Gatheral and co-authors \cite{GJR14,BFG16} report $H\approx 0.05$) . It
is here that higher-order approximations come to help and we have included
quantitative estimates, more precisely: {\it strong} rates, throughout. 
An analysis of {\it weak} rates will be conducted elsewhere, as is the investigation of multi-level
algorithms (cf. \cite{BFRS16} for MLMC for general Gaussian rough differential equations). Recall
that the design of MLMC algorithms requires knowledge of strong rates. 
Numerical aspects are further explored in\ Section \ref%
{sec:ChristianBen}. 

\bigskip

The second set of results concerns large deviations for rough volatility. Thanks to the contraction principle and fundamental continuity properties of Hairer's reconstruction map, the problem is reduced 
to understanding a LDP for a suitable enhancement of the noise. This approach requires (sufficiently) smooth coefficients, but comes with no growth restrictions which is indeed quite suitable for
financial modelling: we improve the Forde-Zhang (simple rough vol) short-time large deviations \cite{FZ17} such as to include $f$ of exponential type, a defining feature in the works of Gatheral and coauthors \cite{GJR14,BFG16}. (Such an extension is also subject of a recent preprint \cite{JPS17x} and forthcoming work \cite{Gul17p}.)

\begin{theorem} \label{FZbetter} Let $X_t = log(S_t/S_0)$ be the log-price under simple rough SV, i.e. (\ref{equ:S}), (\ref{eq:sigmasimple}). Then $(t^{H-\frac{1}{2}} X_t: t \ge 0)$ satisfies a short time large deviation principle with speed $t^{2H}$ and rate
function given by 
\begin{equation}  \label{eq:rateZ}
I(y) = \inf_{h \in L^2([0,1])} \{ \frac{1}{2} \|h\|_{L^2}^2 + \frac{%
\left(y- \rho I_1(h)\right)^2}{2I_2(h)} \} 
\end{equation}
with $I_1(h)=\int_0^1 f( \hat{h}(t)) h(t) dt, \  I_2^z(h)= \int_0^1 f( \hat{h}(t) )^2 dt$ where $\hat{h}(t) = \int_0^t K(s,t)  h(s) ds$.
\end{theorem} 

%

%

 \begin{remark}  A potential short-coming is the non-explicit form of the rate
   function, in the sense that even geometric or Hamiltonian descriptions of
   the rate function (classical in Markovian setting, see e.g \cite{ABOBF03,
     BBF04, DFJV14a, DFJV14b,BL14}), which led to the famous SABR volatility smile formula, is lost. A partial remedy here is to move
   from large deviations to (higher order) {\it moderate deviations}, which restores
   analytic tractability and still captures the main feature of the volatiliy
   smile close to the money. This method was introduced in a Markovain setting in \cite{FGP17}, the extension to simple rough volatility was given in \cite{BFGHS17x}, relying either on \cite{FZ17} or the above Theorem \ref{FZbetter}.
\end{remark} 

\medskip

\noindent We next turn to non-simple rough volatility, motivated by Rosenbaum and coworkers \cite{EEFR16x, EER16x,EER17x}, and consider the stochastic  It\^o--Volterra equation 
\begin{equation*} 
Z_t = z + \int_0^t K(s,t) \left( u(Z_s) dWs + v(Z_s) ds \right)
\end{equation*}
with corresponding rough SV log-price process given by 
\begin{equation*}
X_t = \int_0^t f(Z_s) (\rho dW_s + \bar{\rho} d\bar{W}_s ) - \frac{1}{2}
\int_0^t f^2(Z_s) ds \ .
\end{equation*}
(For simplicity, we here consider $f, u , v$ to be bounded, with bounded derivatives of all orders.) For $h$ $\in$ $L^2([0,T])$, let $z^h$ be the unique solution to the integral
equation 
\begin{equation*}
z^h(t) = z + \int_0^t K(s,t) u(z^h(s)) h(s) ds,
\end{equation*}
and define $I_1(h)=\int_0^1 f(z^h(s)) h(s) ds$ and $I_2^z(h)= \int_0^1
f(z^h(s))^2 ds$. Then we have the following extension of Theorem \ref{FZbetter} (and also \cite{FZ17,JPS17x,Gul17p}) to non-simple rough volatility:

\begin{theorem} \label{thm:coolLDP}Let $X_t = log(S_t/S_0)$ be the log-price under non-simple rough SV. Then
$t^{H-\frac{1}{2}} X_t$ satisfies a LDP with speed $t^{2H}$ and rate
function given by 
\begin{equation}  \label{eq:rateZ}
I(x) = \inf_{h \in L^2([0,T])} \{ \frac{1}{2} \|h\|_{L^2}^2 + \frac{%
\left(x- \rho I_1^z(h)\right)^2}{2I_2^z(h)} \}.
\end{equation}
\end{theorem}

\begin{remark}  We showed in \cite[Cor.11]{BFGHS17x} (but see related results by Alos et al. \cite{Alos07} and Fukasawa \cite{Fuk11, Fuk17}) that in the previously considered simple rough volatility models, now writing  $\sigma(.)$ instead of $f(.)$, the {\it implied volatility skew} behaves, in the short time limit, as $
\sim          \rho \tfrac{\sigma' (0)}{\sigma (0)} \langle K1,1 \rangle t^{H-1/2} \ ,
$
where $\langle K1,1 \rangle$ in our setting computes to $c_H := \tfrac{(2H)^{1/2}}{(H+1/2)(H+3/2)}$. (The blowup $t^{H-1/2}$ as $t\to 0$ is a desired feature, in agreement with steep skews seen in the market.) To first order $Z_t \approx z + u(z) \int_0^t K(s,t) dWs = z + u(z) \hat{W} =: \sigma (\hat{W})$, from which one obtains a skew-formula in the non-simple rough volatility case of the form,
$$
       \rho u(z) \frac{f' (z)}{f(z)} c_H t^{H-1/2} \ .
$$
Following the approach of \cite{BFGHS17x}, Theorem \ref{thm:coolLDP} not only allows for rigorous justification but also for the computation of higher order smile features, though this is not pursued in this article.
In the case of classical (Markovian) stochastic volaility, $H=1/2$, and specializing further to $f(x) \equiv x$, so that $Z$ (resp. $z$) models stochastic (resp. spot) volatility, this reduces precisely to the popular skew formula  Gatheral's book \cite[(7.6)]{Gat06}, attributed therein to Medvedev--Scaillet. In the case of {\it rough Heston}, where $Z$ models stochastic variance, cf. (\ref{equ:rHeston}), we have $f=\sqrt{.}, u=\eta\sqrt{.}$ and this leads to the following (rough Heston, implied volatility) short-dated skew formula
$$
         \frac{\rho \eta}{2 \sqrt{v_0}} c_H t^{H-1/2} \ ,
$$
(multiply with $2 \sqrt{v_0}$ to get the implied variance skew, again in agreement with Gatheral \cite[p.35]{Gat06}); this may be independently verified via the characteristic function obtained in \cite{EER16x}.

\end{remark}

\bigskip  
\noindent {\bf Structure of the article.} In Section \ref{sec:red} we reduce the proofs of Theorems \ref{thmEasy} and \ref{thm:rPricing} to the key convergence issue, subject of Section \ref{sec:RPRS}. In Section \ref{sec:full-rough-vol-reg-struc} we consider the structure for two-dimensional noise, necessary to study the asset price process. Section \ref{sec:Volt} then discusses the case of non-trivial dynamics for rough volatility. Some numerical results are presented in \cite{}, followed by several appendices with technical details. From Section \ref{sec:RPRS} all our work relies on the framework of Hairer's regularity structures.
 There seems to be no point in repeating all the necessary definitions and terminology, which the reader can find in \cite{Hai13,Hai14,Hai15,FH14} and a variety of survey papers on the subject. Instead, we find it more instructive to substantiate our KPZ inspiration and in the next section introduce, informally, all relevant objects from regularity structures in this context.

\subsection{Lessons from KPZ and singular SPDE theory}
\label{sec:lessons-from-kpz}

The absence of a good approximation theory is a defining feature of all singular SPDE recently considered by Hairer, Gubinelli et al. (and now many others). 
In particular, approximation of the noise (say, $\varepsilon$-mollification for the sake of argument) typically does {\it not} give rise to convergent approximations. To be specific, it is instructive to recall the universal model for fluctuations of interface growth given by the
Kardar--Parisi--Zhang (KPZ) equation%
\begin{equation*}
\partial _{t}u=\partial _{x}^{2}u+|\partial _{x}u|^{2}+\xi 
\end{equation*}%
with space-time white noise $\xi =\xi \left( x,t;\omega \right) $. As a matter of fact, and without going in further detail, there is a well-defined (``Cole-Hopf'') It\^o-solution $u=u(t,x;\omega)$, \ 
but if one considers the equation with $\varepsilon$-mollified noise, then $u=u^\varepsilon$ diverges with $\varepsilon \to 0$. In this sense, there is a fundamental {\it lack of approximation theory} and {\it no Stratonovich solution} to KPZ exists. To see the problem, take $u_{0}\equiv 0$ for simplicity and write
\begin{equation*}
u=H\star \left( |\partial _{x}u|^{2}+\xi \right) 
\end{equation*}%
with space-time convolution $\star$ and heat-kernel%
\begin{equation*}
H\left( t,x\right) =\frac{1}{\sqrt{4\pi t}}\exp \left( -\frac{x^{2}}{4t}%
\right) \ 1_{ \{ t >0 \}}
\end{equation*}
One can proceed with Picard iteration 
\begin{equation*}
u=H\star \xi +H\star ((H^{\prime }\star \xi )^{2})+... \,
\end{equation*}
but there is an immediate problem with $(H^{\prime }\star \xi )^{2}$, (naively) defined $\varepsilon$-to-zero limit of 
$ (H^{\prime }\star \xi ^{\varepsilon })^{2} $, 
which does not exist. However, there exists a  diverging sequence $\left( C_{\varepsilon}\right) $ such that, in probability,
\begin{equation*}
\exists \lim_{\varepsilon \rightarrow 0}(H^{\prime }\star \xi ^{\varepsilon
})^{2} { - C_{\varepsilon }}\rightarrow \left( \text{new object}\right) =: (H' \star \xi)^{\diamond 2}.
\end{equation*}
The idea of Hairer, following the philosophy of rough paths, was then to accept
\begin{equation*}
H\star \xi, (H' \star \xi)^{\diamond 2} \text{ (and a few more)}
\end{equation*}%
as enhancement of the noise ("{\bf model}") upon which solution depends in pathwise robust fashion. This unlocks the seemingly fixed (and here even non-sensical) relation%
\begin{equation*}
H\star \xi \rightarrow \xi \rightarrow (H^{\prime }\star \xi )^{2}.
\end{equation*}%
Loosely speaking, one has

\begin{theorem}[Hairer]  There exist diverging constants  $C_{\varepsilon }$ such that a Wong-Zakai\footnote{Hairer--Pardoux \cite{HP15} derive the KPZ result as special case of a Wong-Zakai result for It\^o-SPDEs.}  result holds of the form $\tilde{u}^{\varepsilon } \to u$, in probability and uniformly on compacts, where
\begin{equation*}
\partial _{t}\tilde{u}^{\varepsilon }=\partial _{x}^{2}\tilde{u}^{\varepsilon
}+|\partial _{x}\tilde{u}^{\varepsilon }|^{2} \ 
- {C_{\varepsilon}}+\xi
^{\varepsilon }.
\end{equation*}
Similar results hold for a number of other singular semilinear SPDEs.
\end{theorem}

In a sense, this can be traced back to the {\it Milstein-scheme} for SDEs and then {\it rough paths}:
Consider $ dY=f\left( Y\right) dW$, with $Y_{0}=0$ for simplicity, and consider the 2nd order (Milstein) approximation
\begin{equation*}
Y_{t_{i+1}} \approx Y_{t_{i}} + f\left(  Y_{t_{i}} \right) W_{t_i, t_{i+1}}+ff^{\prime }\left(  Y_{t_{i}} \right) \int_{t_i}^{t_{i+1}} W_{t_i,s} \dot W_s ds
\end{equation*}
One has to unlock the seemingly fixed relation%
\begin{equation*}
W\rightarrow \dot{W}\rightarrow \int W\dot{W} ds =: \mathbb{W} \ ,
\end{equation*}%
for there is a choice to be made. For instance, the last term can be understood as It\^o-integral $\int W dW$ or as Stratonovich integral $\int W \circ dW$ (and in fact, there are many other choices, see e.g. the discussion in \cite{FH14}.) It suffices to take this thought one step further to arrive at {\it rough path theory}: accept $\mathbb{W}$ as new (analytic) object, which leads to the main (rough path) insight 
\begin{equation*}
\text{SDE theory = analysis based on }\left( W,\mathbb{W}\right).
\end{equation*}%
In comparison, 
\begin{eqnarray*}
&&\text{SPDE theory \`a la Hairer } \\
&\text{=}&\text{ analysis based on (renormalized) enhanced noise $(\xi, ....)$.}
\end{eqnarray*}

\bigskip

\noindent
{\bf Inside Hairer's theory:}  \footnote{In the section only, following \cite{FH14}, symbols will be coloured.} As motivation, consider the Taylor-expansion (at $x$) of a real-valued smooth function,
$$
       f(y) =  f(x) + f'(x) (y-x) + \frac{1}{2} f''(x) (y-x)^2 + ... \ ,
$$
can be written as abstract polynomial (``jet") at $x$,
$$
     \symbol F (x)  :=     f(x) \, \symbol{1}+ g(x) \symbol{X} +h(x) \symbol{X^2} + ...  \ , \\ 
$$
with, necessarily, $g = f', \ h = f''/2, ...$.      
If we {``realize''} these abstract symbols again as honest monomials, i.e. ${\Pi_x}: \symbol{X^k}   \mapsto ( . - x)^k $
and extend ${ \Pi_x}$ linearly, then we recover the full Taylor expansion:
$$
        { \Pi_x} [ \symbol F (x)] (.)  =  f(x) + g(x) (.-x) + \frac{1}{2} h(x) (.-x)^2 + ... 
$$
Hairer looks for solution of this form: at every space-time point a jet is attached, which in case of KPZ turns out - after solving an abstract fixed point problem - to be of the form
$$
\symbol U (x,s) = u (x,s) \, \symbol{1} + \<d> + \<2d> + v(x,s) \,\symbol{X} + 2 \<21d> + v(x,s) \, \<1d>\;.
$$
As before, every symbol is given concrete meaning by ``realizing'' it as honest function (or Schwartz distribution).
In particular,
\begin{equation}
    \<d> \mapsto 
   \begin{cases}
    H \star  \xi^\epsilon, & \text{mollified noise; {\bf or} }\\
    H \star  \xi & \text{noise } 
  \end{cases}
\end{equation}
and then, more interestingly,
\begin{equation}
    \<2d> \mapsto 
   \begin{cases}
    H \star (H' \star \xi^\epsilon )^2  , & \text{canonically enhanced mollified noise;  {\bf or} }\\
    H \star [ (H' \star \xi^\epsilon )^2  - C_\epsilon ] , & \text{renormalized $\sim$  {\bf or} } \\
     H \star (H' \star \xi)^{\diamond 2}, & \text{renormalized enhanced noise} 
  \end{cases}
\end{equation}
This realization map is called ``model'' and captures exactly a typical, but otherwise fixed, realization of the noise (mollified or not) together with some enhancement thereof, renormalized or not.  For instance, writing {$\Pi_{x,s}$} for the realization map for renormalized enhanced noise, one has 
$$
           {\Pi_{x,s}} [ \symbol U (x,s)] ( .)  =  u (x,s) \,  + H \star  \xi |_{(*)}    + H \star (H' \star \xi)^{\diamond 2} |_{(*_)} + ...
$$
where $(*)$ indicates suitable centering at $(x,s)$. 
Mind that $\symbol U$ takes values in a (finite) linear space spanned by (sufficiently many) symbols,
$$
      \symbol U (x,s) \in \langle ...,\symbol{1}, \, \<d>, \, \<2d>,  \, \symbol{X}, \, \<21d>, \, \<1d>, ...\rangle \, =: { \mathcal{T}}
$$
The map $(x,s) \mapsto \symbol U (x,s)$  is an example of a {\bf modelled distribution}, the precise definition is a mix of suitable analytic and algebraic conditions (similar to the notation of a controlled rough path).

The analysis requires keeping track of the {\it degree} (a.k.a. {\it homogeneity}) of each symbol. For instance, $ \, |\<d>| = 1/2 - \kappa$ (related to the H\"older regularity of the realized object one has in mind),
$|\symbol{X^2}| = 2$ etc. All these degrees are collected in an {\bf index set}.
Last not least, in order to compare jets at different points (think $(\symbol{X} - \delta \symbol{1})^{\symbol 3}  = ...$), use a group of linear maps on ${ \mathcal{T}}$, called {\bf structure group.} Last not least, the {\bf reconstruction map} uniquely maps modelled distributions to function / Schwartz distributions. (This can be seen as generalization of the {\it sewing lemma}, the essence of rough integration, see e.g. \cite{FH14}, which turns a collection of sufficiently compatible local expansions into one function / Schwartz distribution.) In the KPZ context, the (Cole-Hopf It\^o) solution is then indeed obtained as reconstruction of the abstract (modelled distribution) solution $\symbol U$.

\bigskip

\noindent\textbf{Acknowledgment:}   
The authors acknowledge financial support from DFGs research  grants BA5484/1 (CB, BS) and 
FR2943/2 (PKF, BS), the ERC via Grant  CoG-683166 (PKF), the ANR via Grant ANR-16-CE40-0020-01 (PG) 
and DFG Research Training Group RTG 1845 (JM). 

\medskip

\noindent Participants of Global Derivatives 2017 (Barcelona) and Gatheral 60th Birthday conference (CIMS, NYU) are thanked
for the feedback.  

\bigskip

\section{Reduction of Theorems \ref{thmEasy} and \ref{thm:rPricing}} \label{sec:red}

In the context of these theorems, we have
\begin{equation}
S_{t} = S_{0} \exp \left[ \int_{0}^t f\left( \hat{W}_s  \right) dB_{s}  - \tfrac{1}{2}  \int_{0}^{t} f^2 \left(  \hat{W}_s  \right) d{s}\right] .
\end{equation}
where we
recall that
$$ \int_0^t f(\hat{W})dB=\rho  \int_0^t f ( \hat{W} )  dW +\bar{\rho}\int_0^t
f ( \hat{W} ) d\bar{W} .
$$
All approximations, $W^\eps, \bar W^\eps$ and $B^\eps \equiv \rho W^\eps + \bar \rho  \bar W^\eps$ converge 
uniformly to the obvious limits, so that it suffices to understand the convergence of the stochastic integral. Note that $\tilde W$ is heavily 
correlated with $W$ but independent of $\bar W$. The difficult interesting part is then indeed (\ref{equ:keyconvergence}), i.e.
\begin{equation}
   \int_{0}^{t}f(\hat{W}^{\varepsilon})dW^{\varepsilon }- \int_{0}^{t} 
  \mathscr{C}^\varepsilon(s)  f^{\prime }( \hat{W}_{s}^{\varepsilon }) ds \rightarrow  \int_0^t f ( \hat{W} )  dW \ ,
\end{equation}%
which is the purpose of Theorem \ref{thm:ConvergenceDiscreteIntegrals}. For the other part, due to independence no correction terms arise and we have (with details left to the reader)
$
     \int_{0}^{t}f(\hat{W}^{\varepsilon})d \bar W^{\varepsilon } \rightarrow  \int_0^t f ( \hat{W} )  d \bar W \ ,
$
with convergence in probability and uniformly on compacts in $t$. The convergence result of Theorems \ref{thmEasy} then follows readily.

\bigskip

\noindent As for pricing, Theorem {\ref{thm:rPricing}, consider the call payoff 
$\left( S_0 \exp \left[ \int_0^T \sigma \left( t,\omega \right) dB_t -\frac{1}{2}\int_0^T \sigma^2 (t,\omega) dt\right] -K\right) ^{+}$.
An elementary conditioning argument (first used by Romano--Touzi in the context of Markovian SV models) w.r.t. $W$, then shows that the call price is given as expection of 
\begin{equation*}
 C_{BS} \left(  S_{0}\exp \left( \rho
\int_{0}^{T}\sigma \left( t,\omega \right) dW-\frac{\rho ^{2}}{2}%
\int_{0}^{T}\sigma ^{2}\left( t,\omega \right) dt\right) ,K,\frac{\bar{\rho}%
^{2}}{2}\int_{0}^{T}\sigma ^{2}\left( t,\omega \right) dt \right) .
\end{equation*}
Specializing to the case $\sigma = f(\tilde W)$, in combination with Theorem \ref{thm:ConvergenceDiscreteIntegrals}, then yields Theorem {\ref{thm:rPricing} . Remark that extensions to non-simple RV are immediate from suitable extensions of  Theorem \ref{thm:ConvergenceDiscreteIntegrals}, as discussed in \ref{sec:VolterraStuff}. 


\section{The rough pricing regularity structure}

\label{sec:RPRS}

In this section we develop the approximation theory for integrals of the
type $\int f(\tilde W) dW$. In the first part we present the regularity structure
and the associated models we will use. In the second part we apply the
reconstruction theorem from regularity structures to conclude our main
result, Theorem \ref{thm:ConvergenceDiscreteIntegrals}.


\subsection{Basic pricing setup}
\label{sec:basic-pricing-setup}

We are given a Hurst parameter $H \in (0, 1/2]$, associated to a fractional
Brownian motion (in the Riemann-Liouville sense) $\hat{W}$, and fix an arbitrary $\kappa\in (0,H)$ and an
integer 
\begin{align*}
M\geq \max\{m\in \mathbb{N}\,\vert \, m\cdot (H-\kappa)-1/2-\kappa\leq 0\}
\end{align*}
so that 
\begin{align}  \label{eq:condonkappa}
(M+1)(H-\kappa)-1/2-\kappa >0\,.
\end{align}

At this stage, we can introduce the ``level-$(M+1)$'' \textit{model space} 
\begin{align}     \label{equ:SimpleModelSpace}
\mathcal{T}=\left\langle \{ \Xi, \Xi \mathcal{I}(\Xi),\ldots,\Xi \mathcal{I}%
(\Xi)^M,\mathbf{1},\mathcal{I}(\Xi),\ldots,\mathcal{I}(\Xi)^M \} \right\rangle
\,,
\end{align}
where $\langle\ldots \rangle$ denotes the vector space generated by the
(purely abstract) symbols in $\{\ldots \}$. 
We will sometimes write 
\begin{equation*}
S=S^{(M)}:=\{\Xi, \Xi \mathcal{I}(\Xi),\ldots,\Xi \mathcal{I}(\Xi)^M,\mathbf{%
1},\mathcal{I}(\Xi),\ldots,\mathcal{I}(\Xi)^M\}
\end{equation*}
so that $\mathcal{T}=\mathcal{T}^{(M)}=\bigoplus_{\tau\in S} \mathbb{R}  \tau$.

\begin{remark}
It is useful here and in the sequel to consider as sanity check the special case $H=1/2$ in which case we
recover the ``level-$2$'' rough path structure as introduced in \cite[Ch.13]{FH14}.
More specifically, if take H\"older exponent $\alpha := 1/2 - \kappa < 1/2$ and (and then $M=1$) condition %
\eqref{eq:condonkappa} is precisely the familiar condition $\alpha
>1/3$.
\end{remark}

The interpretation for the symbols in $S$ is as follows: $\Xi$ should be
understood as an abstract representation of the white noise $\xi$ belonging
to the Brownian motion $W$, i.e. $\xi=\dot{W}$ where the derivative is taken
in the distributional sense. Note that since we set $W(x)=0$ for $x\leq 0$
we have $\dot{W}(\varphi)=0$ for $\varphi\in C^\infty_c((-\infty,0))$. The
symbol $\mathcal{I}(\ldots)$ has the intuitive meaning ``integration against the
Volterra kernel'', so that $\mathcal{I}(\Xi)$ represents the integration of
white noise against the Volterra kernel 
\begin{align*}
\sqrt{2H}\int_0^t|t-r|^{H-1/2} \mathrm{d}  W(r)\,,
\end{align*}
which is nothing but the fractional Brownian motion $\hat{W}(t)$. Symbols
like $\Xi \mathcal{I}(\Xi)^m =\Xi \cdot \mathcal{I}(\Xi)\cdot \ldots \cdot 
\mathcal{I}(\Xi)$ or $\mathcal{I}(\Xi)^m=\mathcal{I}(\Xi)\cdot \ldots \cdot 
\mathcal{I}(\Xi)$ should be read as products between the objects above.
These interpretations of the symbols generating $\mathcal{T}$ will be made
rigorous by the model $(\Pi,\Gamma)$ in the next subsection.
Every symbol in $S$ is assigned a homogeneity, which we define by 
\begin{align*}
|\Xi \mathcal{I}(\Xi)^m| &= - 1/2 - \kappa + m (H-\kappa),\,m\geq 0 \\
|\mathcal{I}(\Xi)^m|&=m (H-\kappa) ,\,m>0 \\
|\mathbf{1}|&=0\,,
\end{align*}
We collect the homogeneities of elements of $S$ in a set $%
A:=\{|\tau|\,\vert\, \tau \in S\}$, whose minimum is $|\Xi|=-1/2-\kappa$.
Note that the homogeneities are  
multiplicative in the sense that, $|\tau\cdot \tau^{\prime }| = |\tau| + |\tau^{\prime }|$
for $\tau,\tau^{\prime }\in S$. 


At last, our regularity comes with a {\it structure group} $G$, an (abstract) group of linear operators on the model space $\mathcal{T}$ which should satisfy $\Gamma \tau- \tau=\bigoplus_{\tau'\in S:\,|\tau'|<|\tau|} \RR\tau'$ and $\Gamma \mathbf{1}=\mathbf{1}$ for $\tau \in S$ and $\Gamma \in G$. We will choose $G=\{\Gamma_h \,\vert\, h\in (\RR,+)\}$ given by 
\begin{align*}
\Gamma_h \mathbf{1}=\mathbf{1},\,\Gamma_h \Xi=\Xi,\,\Gamma_h \mathcal{I}(\Xi)=\mathcal{I}(\Xi)+h \mathbf{1}\,. 
\end{align*}
and $\Gamma_h (\tau'\cdot \tau)=\Gamma_h \tau' \cdot \Gamma_h \tau$ for $\tau',\,\tau \in S$ for which $\tau\cdot \tau'\in S$ is defined. 
\bigskip

\subsubsection*{The limiting model $(\Pi,\Gamma)$}

Let $W$ be a Brownian motion on $\mathbb{R} _+$ and extend it to all of $%
\mathbb{R} $ by requiring $W(x)=0$ for $x\leq 0$. We will frequently use the
notations 
\begin{align}  \label{eq:NotationStochasticIntegral}
\int_0^t f(t) \mathrm{d}  W(t) ,\, \int_0^t f(t) \diamond \mathrm{d}  W (t)
\end{align}
which denote the It\^o integral and the Skohorod integral (which boils down
to an It\^o integral whenever the integrand is adapted). From $W$ we
construct now the fractional Riemann-Liouville Brownian motion $\hat{W}$
with Hurst index $H\in (0,1/2]$ as 
\begin{align*}
\hat{W}(t)=\dot{W}\star K(t)=\sqrt{2H}\int_{0}^{t}|t-r|^{H-1/2}\,\mathrm{d} 
W(r)\,,
\end{align*}
where $K(t)=\sqrt{2H}\mathbf{1}_{t>0}\cdot t^{H-1/2}$ denotes the Volterra
kernel. We also write $K(s,t) := K (t-s)$.

To give a meaning to the product terms $\Xi\mathcal{I}(\Xi)^k$ we follow the
ideas from rough paths and define an ``iterated integral'' for $s,t\in 
\mathbb{R} , s\leq t$ as 
\begin{align}  \label{eq:DefinitionWW}
\mathbb{W}^m(s,t)=\int_s^t (\hat{W}(r)-\hat{W}(s))^m\,\mathrm{d}  W(r)
\end{align}
$\mathbb{W}^m(s,t)$ satisfies a modification of Chen's relation

\begin{lemma}
\label{lem:ChensRelation}  $\mathbb{W}^m$ as defined in %
\eqref{eq:DefinitionWW} satisfies 
\begin{align}  \label{eq:ChensRelation}
\mathbb{W}^m(s,t)=\mathbb{W}^m(s,u)+\sum_{l=0}^m \binom{m}{l} (\hat{W}(u)-%
\hat{W}(s))^l \mathbb{W}^{m-l}(u,t)
\end{align}
for $s,u,t\in \mathbb{R} ,\,s\leq u\leq t$.
\end{lemma}

\begin{proof}
Direct consequence of the binomial theorem.
\end{proof}

We extend the domain of $\mathbb{W}^m$ to all of $\mathbb{R} ^2$ by imposing
Chen's relation for all $s,u,t\in \mathbb{R} $, i.e. we set for $t,s\in 
\mathbb{R} ,\, t\leq s$ 
\begin{align}  \label{eq:reversedorder}
\mathbb{W}^m(s,t)=-\sum_{l=0}^m \binom{m}{l} (\hat{W}(t)-\hat{W}(s))^l 
\mathbb{W}^{m-l}(t,s)
\end{align}

We are now in the position to define a model $(\Pi,\Gamma)$ that gives a
rigorous meaning to the interpretation we gave above for $\Xi,\mathcal{I}%
(\Xi),\Xi\mathcal{I}(\Xi),\ldots\,$. Recall that in the theory of regularity
structures a model is a collection of linear maps $\Pi_s: \mathcal{T}%
\rightarrow C^1_c(\mathbb{R} )^{\prime }$, $\Gamma_{st}\in G$ for indices $s,t\in\mathbb{R} $ that satisfiy 
\begin{align}
&\Pi_t =\Pi_s \Gamma_{st},\,  \label{eq:ModelDefinition1} \\
&|\Pi_s \tau (\varphi^\lambda_s)| \lesssim \lambda^{|\tau|}\,,
\label{eq:ModelDefinition2} \\
&\Gamma_{st} \tau =\tau+ \sum_{\tau^{\prime }\in S:\,|\tau^{\prime }|<\tau}
c_{\tau^{\prime }}(s,t) \tau^{\prime },\,|c_{\tau^{\prime }}(s,t)|\lesssim
|s-t|^{|\tau|-|\tau^{\prime }|}  \label{eq:ModelDefinition3}
\end{align}
where the bounds hold uniformly for $\tau\in S$, any $s,t$ in a compact set
and for $\varphi^\lambda_s:=\lambda^{-1}\varphi(\lambda^{-1} (\cdot-s))$
with $\lambda\in(0,1]$ and $\varphi\in C^1$ with compact support in the ball 
$B(0,1)$.

We will work with the following ``It\^o'' model $(\Pi,\Gamma)$, and (occasionally) write $(\Pi^{\text{It\^o}},\Gamma^{\text{It\^o}})$ to avoid confusion with a generic model, also denoted by $(\Pi,\Gamma)$,
which renders more
precisely our interpretations of the elements of $S$. 
\label{LimitingModel}
\begin{equation*}
\begin{array}{ll}
\Pi_s \mathbf{1}=1 & \Gamma_{ts} \mathbf{1} =\mathbf{1} \\ 
\Pi_s \Xi = \dot{W} \qquad & \Gamma_{ts} \Xi =\Xi \\ 
\Pi_s \mathcal{I}(\Xi)^m = \left(\hat{W}(\cdot)-\hat{W}(s)\right)^m & 
\Gamma_{ts} \mathcal{I}(\Xi)=\mathcal{I}(\Xi)+(\hat{W}(t)-\hat{W}(s))\mathbf{%
1} \\ 
\Pi_s \Xi \mathcal{I}(\Xi)^m = \{ t \mapsto \frac{\mathrm{d} }{\mathrm{d}  t}\mathbb{W}%
^m(s, t) \} & \Gamma_{ts} \tau \tau^{\prime }=\Gamma_{ts}\tau \cdot
\Gamma_{ts}\tau^{\prime }\,,\,\, \mbox{ for } \tau,\tau^{\prime }\in S\mbox{
with } \tau\tau^{\prime }\in S%
\end{array}
\end{equation*} 
We extend both maps from $S$ to $\mathcal{T}$ by imposing linearity.

\begin{lemma}
\label{lem:LimitingModel} The pair $(\Pi,\Gamma)$ as defined above defines
(a.s.) a model on $(\mathcal{T},A)$.
\end{lemma}

\begin{proof}
The only symbol in $S$ on which \eqref{eq:ModelDefinition1} is not
straightforward is $\Xi \mathcal{I}(\Xi)^m$, where the statement follows by
Chen's relation. The bounds \eqref{eq:ModelDefinition2} and %
\eqref{eq:ModelDefinition3} follow for $\mathbf{1}$ trivially and for $%
\mathcal{I}(\Xi)^m$ by the $H-\kappa^{\prime },\,\kappa^{\prime }\in (0,H)$
H\"older regularity of $\hat{W}$. It is further straightforward to check the
condition \eqref{eq:ModelDefinition3} by using the rule $\Gamma_{ts}\tau%
\tau^{\prime }=\Gamma_{ts}\tau \cdot \Gamma_{ts}\tau^{\prime }$ so that we
are only left with the task to bound $\Pi_s \Xi\mathcal{I}%
(\Xi)^m(\varphi^\lambda_s)$. Following along the lines of proof \cite[%
Theorem 3.1]{FH14} it follows $|\mathbb{W}^m(s,t)|\leq C
|s-t|^{mH+1/2-(m+1)\kappa}$ (where $C>0$ denotes a random constant with $C\in \bigcup_{p<\infty} L^p$), so that 

\begin{align*}
|\Pi_s \mathcal{I}(\Xi)^m\Xi(\varphi^\lambda_s)|&=\left|\int
\left(\varphi^\lambda_s\right)^{\prime }(t) \mathbb{W}^m(s,t) \,\mathrm{d} 
t\right| \leq C \int \varphi^{\prime -1}(t-s)) |s-t|^{mH+1/2-(m+1)\kappa} 
\frac{\mathrm{d}  t}{\lambda^2} \\
&\leq C \lambda^{mH-1/2-(m+1)\kappa}= C\lambda^{|\mathcal{I}(\Xi)^m\Xi|}\,.
\end{align*}
\end{proof}

As we will see below in subsection \ref{subsec:ApproximationTheory} this
model is the toolbox from which we can build pathwise It\^o integrals of the
type $\int_0^t f(r,\hat{W}(r)) \, \mathrm{d}  W(r)$. For an approximation
theory for such expressions we are in need of a comparable setup that
describes approximations, which will be achieved by introducing a model $%
(\Pi^\varepsilon,\Gamma^\varepsilon)$.

\subsubsection*{The approximating model $(\Pi^\protect\varepsilon,\Gamma^%
\protect\varepsilon)$}

The whole definition of the model $(\Pi,\Gamma)$ is based on the object $%
\dot{W}$. It is therefore natural to build an approximating model by
replacing $\dot{W}$ by some modification $\dot{W}^\varepsilon$ that
converges (as a distribution) to $\dot{W}$ as $\varepsilon\rightarrow 0$.

The definition of $\dot{W}^\varepsilon$ will be based on an object $%
\delta^\varepsilon$ which should be thought of as an approximation to the
Delta dirac distribution. Our purpose to build $\delta^\varepsilon$ from
wavelets, which can be as irregular as the Haar functions. We find it
therefore convenient to allow $\delta^\varepsilon$ to take values in the
Besov space $\mathcal{B}^{\beta}_{1,\infty}(\mathbb{R} ),\, \beta >1/2+\kappa
$ which covers functions like $\mathbf{1}_{[0,1]}\in \mathcal{B}%
^{1}_{1,\infty}(\mathbb{R} )$.

\begin{remark}
We shortly recall the definition of the Besov space $\mathcal{B}%
^{\beta}_{1,\infty}(\mathbb{R} )$ (see for example \cite{MeyerWavelets})
although this will here only be explicitely used in the proof of Lemma \ref%
{lem:BesovBounds} in the appendix. Given a compactly supported wavelet basis 
$\phi_y=\phi(\cdot-y),\,y\in \mathbb{Z}$, $\psi^j_y=2^{j/2}\,\psi(2^j(
\cdot-y)),\,j \geq 0,\,y\in 2^{-j}\mathbb{Z}$ we set 
\begin{align*}
\|g\|_{\mathcal{B}^{\beta}_{1,\infty}}:=\sum_{y\in \mathbb{Z}}
|(g,\phi_y)_{L^2}| +\sup_{j\geq 0} 2^{j\beta } \sum_{y\in 2^{-j}\mathbb{Z}}
2^{- j/2} |(g,\psi^j_y)_{L^2}|
\end{align*}
and define $\mathcal{B}^{\beta}_{1,\infty}(\mathbb{R} )$ to be those $L^1$
functions $g$ (or $(C^{-\lceil \beta \rceil+1 }_c(\mathbb{R} ))^{\prime }$
distributions if $\beta\leq 0$) for which this norm is finite.
\end{remark}

\begin{definition}
\label{def:Dirac} In the following we call $\delta^{\varepsilon}:\mathbb{R}
^{2}\rightarrow\mathbb{R} $ a measurable, bounded function with the
following properties

\begin{itemize}
\item $\delta^{\varepsilon}(x,y)=\delta^{\varepsilon}(y,x)$ for all $x,y\in 
\mathbb{R} $.,

\item the map $\mathbb{R} \ni x\mapsto \delta^\varepsilon (x,\cdot)\in \mathcal{B}%
^{\beta}_{1,\infty}(\mathbb{R} ) $ is bounded and measurable for some $%
\beta>-|\Xi|=1/2+\kappa$. 

\item $\int_{\mathbb{R} }\delta^{\varepsilon}(x,\cdot)\,\mathrm{d}  x=1$,

\item $\sup_{\mathbb{R} ^{2}}|\delta^{\varepsilon}|\lesssim\varepsilon^{-1}$,

\item $\mathrm{supp}\, \delta^{\varepsilon}(x,\cdot)\subseteq B(x,c\cdot
\varepsilon)$ for any $x\in\mathbb{R} $ and some $c>0$.
\end{itemize}
\end{definition}

\begin{example}
\label{rem:Mollifier} There are two examples which are of particular
interest for our purposes

\begin{itemize}
\item We say that $\delta^\varepsilon$ ``comes from a mollifier'', by which
we mean that there is symmetric, compactly supported $L^\infty\cap \mathcal{B%
}^{\beta}_{1,\infty}(\mathbb{R} )$-function $\rho$, which integrates to $1$
such that 
\begin{align*}
\delta^\varepsilon(x,y)=\varepsilon^{-1}\cdot \rho(\varepsilon^{-1}(y-x))
\end{align*}

\item A further interesting example is the case where $\delta^\varepsilon$
``comes from a wavelet basis''. Consider only $\varepsilon=2^{-N}$ and
choose compactly supported $L^\infty\cap\mathcal{B}^\beta_{1,\infty}$-valued
father wavelets $(\phi_{k,N})_{k\in \mathbb{Z}}$ (e.g. the Haar father
wavelets $\phi_{k,N}=2^{N/2}\cdot \mathbf{1}_{[k2^{-N},(k+1)2^{-N})}$) and
set 
\begin{align*}
\delta^\varepsilon(x,y)=\sum_{k\in \mathbb{Z}} \phi_{k,N}(x) \phi_{k,N}(y)
\end{align*}
Note that we could also add some generations of mother wavelets in this
choice.
\end{itemize}
\end{example}

Note that (locally) $\dot{W}$ is contained in $\mathcal{B}%
^{|\Xi|}_{\infty,\infty}(\mathbb{R} )$ (recall: $|\Xi|=-1/2-\kappa$), so
that due to $\mathcal{B}^{|\Xi|}_{\infty,\infty}(\mathbb{R} ) \subseteq (%
\mathcal{B}^{\beta}_{1,\infty}(\mathbb{R} ))^{\prime }$ we can set 
\begin{align*}
\dot{W}^\varepsilon(t)&:=\langle \dot{W},\delta^\varepsilon(t,\cdot)\rangle 
\mathbf{1}_{\mathbb{R} _+}(t)
\end{align*}
which is a Gaussian process and pathwise measurable and locally bounded. For
(maybe stochastic) integrands $f$ we introduce the notations 
\begin{align*}
\int_0^t f(r)\, \mathrm{d}  W^\varepsilon(r)&:=\int_0^t f(r) \dot{W}%
^\varepsilon(r)\, \mathrm{d}  r
\end{align*}
and if $f$ takes values in some (non-homogeneous) Wiener chaos induced by $%
\dot{W}$ we also introduce 
\begin{align}  \label{eq:NotationApproximateStochasticIntegral}
\int_0^t f(r) \diamond \mathrm{d}  W^\varepsilon(r)&:=\int_0^t f(r) \diamond 
\dot{W}^\varepsilon(r)\, \mathrm{d}  r \,,
\end{align}
where $\diamond$ denotes the Wick product. Note that these two objects do in
general not coincide. The motive for using the same symbol ``$\diamond$'' as
in \eqref{eq:NotationStochasticIntegral} is that %
\eqref{eq:NotationApproximateStochasticIntegral} can be seen as the Skohorod
integral with respect to the Gaussian stochastic measure induced by the
Gaussian process $\dot{W}^\varepsilon$ (for the notion of Wick products and
Skohorod integrals and their links see e.g. \cite{Jan97}).

We now define an approximate fractional Brownian motion by setting 
\begin{align*}
\hat{W}^\varepsilon(t)&=K\star\dot{W}^\varepsilon=\sqrt{2H} \int_0^t
|t-r|^{H-1/2} \, \mathrm{d}  W^\varepsilon( r)
\end{align*}
which has the expected regularity as it is shown in the following lemma.

\begin{lemma}
\label{lem:ConvergenceW} On every compact time intervall $[0,T]$ we have the
estimates 
\begin{align*}
|\hat{W}^\varepsilon(t)-\hat{W}^\varepsilon(s)|\lesssim C_\varepsilon
|t-s|^{H-\kappa^{\prime }},\, |\hat{W}^\varepsilon(t)-\hat{W}%
^\varepsilon(s)-(\hat{W}(t)-\hat{W}(s))|\lesssim C |t-s|^{H-\kappa^{\prime
}}\varepsilon^{\delta \kappa^{\prime }}\,.
\end{align*}
uniformly in $\varepsilon\in (0,1]$ for any $\delta\in (0,1)$ and $%
\kappa^{\prime }\in (0,H)$ and where $C_\varepsilon,C>0$ are random constants that are (uniformly) bounded in $L^p$ for $p\in [1,\infty)$.
\end{lemma}

\begin{proof}
The proof is elementary but a bit bulky and therefore postponed to the
appendix.
\end{proof}

Finally we can give the definition of the approximative model $%
(\Pi^\varepsilon,\Gamma^\varepsilon)$, the ``canonical'' model built from
the approximate (and hence regular) noise $W^\varepsilon$. 
\begin{equation*}
\begin{array}{ll}
\Pi^\varepsilon_s \mathbf{1} =1 & \Gamma_{st}^\varepsilon \mathbf{1}=1 \\ 
\Pi^\varepsilon_s \Xi =\dot{W}^\varepsilon & \Gamma_{st}^\varepsilon \Xi=\Xi
\\ 
\Pi^\varepsilon_s \mathcal{I}(\Xi)^m = \left(\hat{W}^\varepsilon(\cdot)-\hat{%
W}^\varepsilon(s) \right)^m & \Gamma_{st}^\varepsilon \mathcal{I}(\Xi)= 
\mathcal{I}(\Xi) +\left(\hat{W}^\varepsilon(t)-\hat{W}^\varepsilon(s)
\right) \mathbf{1} \\ 
\Pi^\varepsilon_s \mathcal{I}(\Xi)^m\Xi = (\hat{W}^\varepsilon(\cdot)-\hat{W}%
^\varepsilon(s))^m \, \dot{W}^\varepsilon(\cdot) & \Gamma_{st}^\varepsilon
\tau \tau^{\prime} = \Gamma^{\varepsilon}_{st}\tau \cdot
\Gamma^\varepsilon_{st}\tau^{\prime }\,,\,\, \tau,\tau^{\prime },\tau\cdot
\tau^{\prime }\in S%
\end{array}
\end{equation*}

\begin{lemma}
\label{lem:LimitingModel} The pair $(\Pi^\varepsilon,\Gamma^\varepsilon)$ as
defined above is a model on $(\mathcal{T},A)$.
\end{lemma}

\begin{proof}
The identity $\Pi_t=\Gamma_{ts}\Pi_s$ is straightforward to check. The
bounds \eqref{eq:ModelDefinition2} and \eqref{eq:ModelDefinition3} on $%
\Gamma_{st}$ and on $\Pi_s \mathcal{I}(\Xi)^m$ follow from the regularity of 
$\hat{W}^\varepsilon$ as proved in Lemma \ref{lem:ConvergenceW}. The blow-up
of $\Pi_s \Xi\mathcal{I}(\Xi)^m(\varphi^\lambda_s)$ however is even better
than we need, since by the choice of $\delta^\varepsilon$ we have $|\dot{W}%
^\varepsilon|\leq C_\varepsilon $, for some random constant $C_\varepsilon$, on compact sets.
\end{proof}

The definition of this model is justified by the fact that application of
the reconstruction operator (as in Lemma \ref{lem:ReconstructionIdentity})
yields integrals 
\begin{align}  \label{eq:ApproximateIntegralNotRenormalized}
\int_0^t f(r,\hat{W}^\varepsilon(r))\, \mathrm{d}  W^\varepsilon(r)\,.
\end{align}
As we pointed out already in section \ref{sec:Introduction}, there is no
hope that integrals of this type will converge as $\varepsilon\rightarrow 0$
if $H<1/2$. This can be cured by working with a renormalized model $(\hat{\Pi%
}^\varepsilon,\Gamma^\varepsilon)$ instead.

\subsubsection*{The renormalized model $\hat{\Pi}^\protect\varepsilon$}

From the perspective of regularity structures the fundamental reason why
integrals like \eqref{eq:ApproximateIntegralNotRenormalized} fail to
converge to 
\begin{align*}
\int_0^t f(r,\hat{W}(r)) \,\mathrm{d}  W( r)
\end{align*}
lies in the fact that the corresponding models will not satisfy $%
(\Pi^\varepsilon,\Gamma^\varepsilon) \rightarrow (\Pi,\Gamma)$ in a suitable
norm. To see what is going on we will first rewrite $\Pi_s\Xi\mathcal{I}(\Xi)^k$

\begin{lemma}
\label{lem:ReshapingLimit} For $\varphi\in C^\infty_c(\mathbb{R} ),\,s\in 
\mathbb{R} ,\,m\in \{1,\ldots,M\}$ we have 
\begin{eqnarray*}
\Pi_{s}\Xi\mathcal{I}(\Xi)^m(\varphi) & = & \int_{0}^{\infty}\varphi(t)\,(%
\hat{W}(t)-\hat{W}(s))^{m}\diamond \mathrm{d}  W(t) \\
& & -m\int_{0}^{\infty}\,\varphi(t)\,K(s-t)\,(\hat{W}(t)-\hat{W}(s))^{m-1}\,%
\mathrm{d}  t
\end{eqnarray*}

where $\diamond$ denotes the Skorokhod integral and $K(t)=\sqrt{2H}\mathbf{1}%
_{t>0} t^{H-1/2}$ denotes the Volterra kernel. Note that in the second term
the domain of integration is actually $(0,s)$.
\end{lemma}
\begin{remark} Our notation reflects a close relation between the Skorokhod
  integral and the Wick product. Indeed, when $g = \sum X_s
  \mathbf{1}_{[s,t]}$, with summation over a finite partition of $[0,T]$, and
  each $X_s$ a (non-adapted) random variable in a finite Wiener-It\^o chaos,
  it follows from \cite[Thm 7.40]{Jan97} that $\int g \delta W = \sum X_s \diamond W_{s,t}$. Passage to $L^2$-limits is then standard. See also \cite{Nua13} and the references therein. 
\end{remark} 
\begin{proof}
We prove this by reexpressing 
$\mathbb{W}^k(s,t)$. For $s<t$ we have already
$$
\mathbb{W}^{k}(s,t)=\int_s^t \mathrm{d}  W(r)\diamond(\hat{W}(r)-\hat{W}%
(s))^{k}
$$ so that it remains to see what happens for $t<s$. With relation %
\eqref{eq:reversedorder} we have in this case 
\begin{align*}
\mathbb{W}^{k}(s,t) = -\sum_{l=0}^{k}\binom{k}{l}(\hat{W}(t)-\hat{W}%
(s))^{l}\cdot\int \mathrm{d}  r \,\dot{W}(r)\diamond(\hat{W}(r)-\hat{W}%
(t))^{k-l}\mathbf{1}_{t<r<s}\,,
\end{align*}
where we use for the sake of concision formal notation, which is easy to
translate to a rigorous formulation. Using the fact that for Gaussians $%
U_1,V,U_2$ we have 
\begin{equation}
U_1^l \cdot (V\diamond U_2^{k-l})=V\diamond (U_1^l U_2^{k-l})+l\mathbb{E}[V
U_1]\,U_1^{l-1} U_2^{k-l}         \label{U1VU2}
\end{equation}
(a consequence of \cite[Theorems 3.15, 7.33]{Jan97}), we obtain 
\begin{align*}
\mathbb{W}^{k}(s,t) & = -\int \mathrm{d}  r \,\dot{W}(r)\diamond(\hat{W}(r)-%
\hat{W}(s))^{k}\mathbf{1}_{t<r<s} \\
& -\sum_{l=0}^{k}\binom{k}{l}l\cdot\int\mathrm{d}  r\,\mathbb{E} [\dot{W}%
(r)\cdot(\hat{W}(t)-\hat{W}(s))]\cdot(\hat{W}(t)-\hat{W}(s))^{l-1}\cdot(\hat{%
W}(r)-\hat{W}(t))^{k-l}\,.
\end{align*}
Using $\binom{k}{l}=k\binom{k-1}{l-1}$ and $\mathbb{E} [\dot{W}(r)\cdot(\hat{%
W}(t)-\hat{W}(s))]= - K(s-r) \mathbf{1}_{r>0}$ 
for $t<r<s$ we can
reformulate this and obtain 
\begin{align*}
\mathbb{W}^{k}(s,t) & = -\int\mathrm{d}  W(r)\diamond(\hat{W}(r)-\hat{W}%
(s))^{k}\mathbf{1}_{t<r<s} +k\int\mathrm{d}  r K(s-r) (\hat{W}(r)-%
\hat{W}(s))^{k-1}\mathbf{1}_{r>0}\,.  \notag
\end{align*}
(An alternative derivation of the above Skorokohod form can be given in terms
of \cite[Thm 3.2]{NP88}.)  Since $\Pi_s \Xi \mathcal{I}(\Xi)^m (\varphi)=\int \varphi(t) \,\mathrm{d}_t  
\mathbb{W}^{m}(s,t)$ the claim follows.
\end{proof}

Let us also reexpress the approximating model in suitable form.

\begin{lemma}
\label{lem:ReshapingApproximation} For $\varphi\in C^\infty_c(\mathbb{R}
),\,s\in\mathbb{R} ,\,m\in\{1,\ldots,M\}$ we have 
\begin{align*}
\Pi_{s}^{\varepsilon}\Xi\mathcal{I}(\Xi)^m(\varphi) & =
\int_0^\infty\varphi(t)\,(\hat{W}^{\varepsilon}(t)-\hat{W}%
^{\varepsilon}(s))^{m}\diamond \mathrm{d}  W^\varepsilon(t) \\
& -m\int_0^\infty\,\varphi(t)\,\mathscr{K}^{\varepsilon}(s,t)(\hat{W}%
^{\varepsilon}(t)-\hat{W}^{\varepsilon}(s))^{m-1} \mathrm{d}  t \\
& +m\int_0^\infty \varphi(t)\,\mathscr{K}^{\varepsilon}(t,t)(\hat{W}%
^{\varepsilon}(t)-\hat{W}^{\varepsilon}(s))^{m-1}\,\mathrm{d}  t
\end{align*}
where $\diamond$ is defined as in %
\eqref{eq:NotationApproximateStochasticIntegral} and where 
\begin{align}  \label{eq:MollifiedVolterraKernel}
\mathscr{K}^\varepsilon(u,v):=\mathbb{E}[\hat{W}^\varepsilon(u)\dot{W}%
^\varepsilon(v)]=\mathbf{1}_{u,v\geq 0}\int_0^\infty \int_0^\infty
\delta^\varepsilon(v,x_1) \delta^\varepsilon(x_1,x_2) K(u-x_2) \,\mathrm{d} 
x_1 \mathrm{d}  x_2\,.
\end{align}
\end{lemma}

\begin{proof}
Using that for Gaussian $V,\,U$ we have $V U^m=V\diamond U^m+m \mathbb{E}%
[VU] U^{m-1}$ (
this is (\ref{U1VU2}) with $U_2=1$) we can rewrite 
\begin{align*}
\Pi_{s}^{\varepsilon}\Xi\mathcal{I}(\Xi)^m(\varphi) & =
\int_0^\infty\varphi(t)\,(\hat{W}^{\varepsilon}(t)-\hat{W}%
^{\varepsilon}(s))^{m}\diamond \mathrm{d}  W^\varepsilon(t) \\
& +m\int_0^\infty \mathrm{d}  t\,\varphi(t)\,\mathbb{E} [\dot{W%
}^{\varepsilon}(t)\,(\hat{W}^{\varepsilon}(t)-\hat{W}^{\varepsilon}(s))](%
\hat{W}^{\varepsilon}(t)-\hat{W}^{\varepsilon}(s))^{m-1}\cdot
\end{align*}
Inserting $\mathbb{E} [\dot{W}^{\varepsilon}(t)\,(\hat{W}^{\varepsilon}(t)-%
\hat{W}^{\varepsilon}(s))] =\mathscr{K}^{\varepsilon}(t,t)-\mathscr{K}%
^{\varepsilon}(s,t) $ 
shows
the identity.
\end{proof}

Comparing the expressions in Lemma \ref{lem:ReshapingApproximation} and \ref%
{lem:ReshapingLimit} we see that we morally have to subtract 
\begin{align*}
m\int\,\varphi(t)\,\mathscr{K}^{\varepsilon}(t,t)(\hat{W}^{\varepsilon}(t)-%
\hat{W}^{\varepsilon}(s))^{m-1}\,\mathrm{d}  t
\end{align*}
from the model, which will give us a new model $\hat{\Pi}^\varepsilon$. Of
course we have to be careful that this step preserves ``Chen's relation'' $%
\hat{\Pi}_s^\varepsilon\Gamma_{st}=\hat{\Pi}_t^\varepsilon$, see Theorem \ref{prop:ModelConvergence} below.

If we interpret $\mathscr{K}^\varepsilon$ as an approximation to the
Volterra-kernel we see that the expression $$\mathscr{C}^\varepsilon(t) := \mathscr{K}^\varepsilon(t,t),\,t%
\geq 0$$

will correspond to something like ``$0^{H-1/2}=\infty$'' in the
limit $\varepsilon\rightarrow 0$. We have indeed the following upper bound.

\begin{lemma}
\label{lem:ApproximateVolterraEstimate} For all $s,t\in \mathbb{R} $ we have 
\begin{align*}
|\mathscr{K}^\varepsilon(s,t)|\lesssim \varepsilon^{H-1/2}\,.
\end{align*}

\end{lemma}

\begin{proof}
$|\mathscr{K}^{\varepsilon}(s,t)| \lesssim
\varepsilon^{-2}\int_{B(t,c\varepsilon)}\mathrm{d}  x\int_{B(x,c\varepsilon)}%
\mathrm{d}  u\,|s-u|^{H-1/2} \lesssim \varepsilon^{H-1/2}\,. $
\end{proof}

Our hope is now that the new model $\hat{\Pi}^\varepsilon$ converges to $\Pi$
in a suitable sense. Similar to \cite[(2.17)]{Hai14}
we define the distance between two models $(\Pi,\Gamma)$ and $(\tilde{\Pi},%
\tilde{\Gamma})$ on a compact time interval $[0,T]$ as 
\begin{align}  \label{eq:ModelDistance}
{\vert\kern-0.25ex\vert\kern-0.25ex\vert (\Pi,\Gamma);(\tilde{\Pi},\tilde{%
\Gamma}) \vert\kern-0.25ex\vert\kern-0.25ex\vert}_{T}:=\sup_{{\scriptsize 
\begin{array}{c}
\mathrm{supp}\,\varphi\subseteq B(0,1), \\ 
\lambda\in (0,1], \\ 
s\in [0,T],\tau\in S%
\end{array}%
}} \lambda^{-|\tau|} |(\Pi_s-\tilde{\Pi}_s)\tau(\varphi^\lambda_s)|+\sup_{%
{\scriptsize 
\begin{array}{c}
t,s\in [0,T], \\ 
\tau \in S,A\ni\beta<|\tau|%
\end{array}%
}} \frac{|\Gamma_{ts}\tau - \tilde\Gamma_{ts}\tau|_\beta}{%
|t-s|^{|\tau|-\beta}}\,,
\end{align}
where $|\cdot|_\beta$ denotes the absolute value of the coefficient of the
symbol $\tau^{\prime }$ with $|\tau^{\prime }|=\beta$ and where the first
supremum runs over $\varphi\in C^1_c$ with $\|\varphi\|_{C^1}\leq 1$. We
will also need 
\begin{align*}
\|\Pi\|_T=\sup_{{\scriptsize 
\begin{array}{c}
\mathrm{supp}\,\varphi\subseteq B(0,1), \\ 
\lambda\in (0,1], \\ 
s\in [0,T],\tau\in S%
\end{array}%
}} \lambda^{-|\tau|} |\Pi_s \tau(\varphi^\lambda_s)|\,.
\end{align*}

We are now ready to give the fundamental result of this subsection which
plays a key role in our approximation theory. Recall that the (minimal)
homogeneity $|\Xi| = -1/2 - \kappa$ which corresponds to $W$ being H\"older
with exponent $1/2 - \kappa$.

\begin{theorem}
\label{prop:ModelConvergence} Define, for every $s\in[0,T]$, 
the linear map $\hat{\Pi}_s^\varepsilon: \mathcal{T} \to C_c^1 (\mathbb{R})^\prime$ given by, for $%
m\in\{1,\ldots,M\}$ 
\begin{align*}
\hat{\Pi}^\varepsilon_s \Xi\mathcal{I}(\Xi)^m=\Pi^\varepsilon_s\Xi\mathcal{I}%
(\Xi)^m-m\mathscr {C}^\varepsilon(\cdot) \Pi_s^\varepsilon(\mathcal{I}%
(\Xi)^{m-1})
\end{align*}
and $\hat{\Pi}_s^\varepsilon=\Pi_s^\varepsilon$ on all remaining symbols in $S$.
Then 
\begin{equation*}
(\hat{\Pi}^\varepsilon,\hat \Gamma^\varepsilon) := (\hat{\Pi}^\varepsilon,
\Gamma^\varepsilon)
\end{equation*}
defines a (``renormalized'') model on $(\mathcal{T},A)$ and on compact time
intervals we have 
\begin{align}  \label{eq:ConvergenceModel}
 \left\|  {\vert\kern-0.25ex\vert\kern-0.25ex\vert (\hat{\Pi}^\varepsilon,\hat
\Gamma^\varepsilon);(\Pi,\Gamma) \vert\kern-0.25ex\vert\kern-0.25ex\vert}%
_T \right\|_{L^p}\lesssim \varepsilon^{\delta\kappa}\,.
\end{align}
for any $\delta\in (0,1)$ and $p\in [1,\infty)$. In particular, we have ``almost rate $H$'' for $%
M=M(\kappa,H)$ large enough.
\end{theorem}

\begin{remark}
In the special case of the level-2 Brownian rough path (i.e. $H = 1/2,\, M=1$%
) the above result is in precise agreement with known results (even though
the situation here is simpler since we are dealing with \textit{scalar}
Brownian). More specifically, we don't see the usual (strong) rate
``almost'' $1/2$ but have to subtract the H\"older exponent used in the
rough path / model topology (here: $1/2 - \kappa$) which exactly leads to
the rate ``almost $\kappa$''. Since $M=1$ entails the condition $1/2 -
\kappa > 1/3$, we see that $\kappa < 1/6$, exactly as given e.g. in in \cite[%
Ex. 10.14]{FH14}. A better rate can be achieved by working with
higher-level rough path (here: $M>1$) and indeed the special case of $H=1/2$%
, but general $M$, can be seen as a consequence of \cite{FrizRiedel}: at the
price of working with $\sim 1/(1/2 - \kappa)$ levels, one can choose $\kappa$
arbitrarily close to $1/2$ and so recover the usual ``almost'' $1/2$ rate.
Of course, the case $H<1/2$ is out of reach of rough path
considerations. 
\end{remark}

\begin{proof}
Since due to Lemma \ref{lem:ApproximateVolterraEstimate} we have, for fixed $\varepsilon$, that
 $\sup_{ t \in [0,T] } |\mathscr{C}^\varepsilon(t)|< \infty $ 
 and $|\Pi_s\mathcal{I}%
(\Xi)^m|\lesssim |\cdot-s|^{mH}$ the bound \eqref{eq:ModelDefinition2} is
still satisfied. The modification $\hat{\Pi}^\varepsilon_s \Xi\mathcal{I}%
(\Xi)^m-\Pi^\varepsilon_s \Xi\mathcal{I}(\Xi)^m$ does not lead to a violation of ``Chen's
relation''.  
Indeed, using validity of \eqref{eq:ModelDefinition1} for the original model, we have 
\begin{align*}
&\hat{\Pi}^\varepsilon_t\Gamma_{ts}^\varepsilon(\Xi \mathcal{I} (\Xi )^k)=%
\hat{\Pi}^\varepsilon_t\left(\sum_{l=0}^k \binom{k}{l} (\hat{W}%
^\varepsilon(t)-\hat{W}^\varepsilon(s))^l \Xi\mathcal{I} (\Xi )^{k-l} \right)
\\
&=\Pi_s^\varepsilon(\Xi\mathcal{I} (\Xi )^k)-\sum_{l=0}^k \binom{k}{l} (\hat{W}%
^\varepsilon(t)-\hat{W}^\varepsilon(s))^l (k-l)\mathscr{C}%
^\varepsilon(\cdot)(\hat{W}^\varepsilon(\cdot)-\hat{W}%
^\varepsilon(t))^{k-l-1} \\
&=\Pi_s^\varepsilon(\Xi\mathcal{I} (\Xi )^k)-k\mathscr{C}^\varepsilon(\cdot)\,\sum_{l=0}^{k-1}\binom{k-1}{l}\,(\hat{W}^\varepsilon(t)-\hat{W}%
^\varepsilon(s))^l\,(\hat{W}^\varepsilon(\cdot)-\hat{W}%
^\varepsilon(t))^{k-l-1} \\
&=\Pi_s^\varepsilon(\Xi \mathcal{I} (\Xi )^k)-k\mathscr{C}^\varepsilon(\cdot)(\hat{W}^\varepsilon(\cdot)-\hat{W}^\varepsilon(s))^k =\hat{\Pi}%
^\varepsilon_s(\Xi(\mathcal{I} (\Xi )^k)\,.
\end{align*} 
We so see that \eqref{eq:ModelDefinition1} is also satisfied after our modification, and then
easily conclude that $(\hat{\Pi}^\varepsilon,\Gamma^\varepsilon)$ is still
a model on $(\mathcal{T},A)$. At last, the bound \eqref{eq:ConvergenceModel} is a bit technical and left to Appendix~\ref{app:AR}.\end{proof}

\subsection{Approximation and renormalization theory}

\label{subsec:ApproximationTheory} 

We now address to central question of how the integral $\int_0^t
f(\hat{W}^\varepsilon(r),r)\, \mathrm{d%
}W^\varepsilon(r)$ has to be modified to make it convergent against $%
\int_0^t f(W(r),r) \mathrm{d} W(r)$.

The key idea is to combine the convergence result from Theorem \ref{prop:ModelConvergence} 
with Hairer's reconstruction theorem, which we state
below.

We first recall the notion of a modelled distribution, compare \cite[%
Definition 3.1]{Hai14}. We say that a map $F:\mathbb{R}
\rightarrow \mathcal{T}$ is in the space $\mathcal{D}_T^\gamma(\Gamma),\,%
\gamma>0$ for some time horizon $T>0$ if 
\begin{align}  \label{eq:ModelledDistribution}
\|F\|_{\mathcal{D}_T^\gamma(\Gamma)}:=\sup_{A\ni\beta<\gamma, s\in [0,T]}
|F(s)|_{\beta}+\sup_{A\ni\beta <\gamma,s,t\in[0,T],\,s\neq t} \frac{%
|F(t)-\Gamma_{ts}F(s)|_\beta}{|t-s|^{\gamma-\beta}}<\infty \,,
\end{align}
where as above $|\cdot|_\beta$ denotes the absolute value of the coefficient
of the vector $\tau$ with $|\tau|=\beta$. Given two models $(\Pi,\Gamma)$
and $(\overline{\Pi},\overline{\Gamma})$ and two $F,\overline{F}:\,\mathbb{R}
\mapsto \mathcal{T}$ it is also usefull to have the notion of a distance 
\begin{align*}
{\vert\kern-0.25ex\vert\kern-0.25ex\vert F ; \overline{F} \vert\kern%
-0.25ex\vert\kern-0.25ex\vert}_{\mathcal{D}_T^\gamma(\Gamma),\mathcal{D}%
_T^\gamma(\overline{\Gamma})}:=&\sup_{A\ni\beta<\gamma,\, t\in [0,T]} |F(t)-%
\overline{F}(t)|_{\beta} \\
&+\sup_{A\ni\beta<\gamma,\, s,t\in [0,T],\,s\neq t} \frac{%
|F(t)-\Gamma_{ts}F(s)-(\overline{F}(t)-\overline{\Gamma}_{ts}\overline{F}%
(s))|_\beta}{|t-s|^{\gamma-\beta}}\,.
\end{align*}

The reconstruction theorem now states that for $\gamma>0$ a map $F\in 
\mathcal{D}^\gamma_T(\Gamma)$ can be uniquely identified with a distribution
that behaves locally like $\Pi_\cdot F(\cdot)$.

\begin{theorem}
\cite[Theorem 3.10]{Hai14} \label{thm:Reconstruction}

Given a model $(\Pi,\Gamma)$, $\gamma>0$ and a $T>0$ there is a unique
continuous operator\footnote{$\mathcal{C}^{|\Xi|}(\mathbb{R} )$ denotes the
space of distributions that are locally in the Besov space $\mathcal{B}%
^{|\Xi|}_{\infty,\infty}(\mathbb{R} )$ (cmp.~\cite[Remark 3.8]{Hai14}).}
$\mathcal{R}:\,\mathcal{D}_T^{\gamma}(\Gamma)
\rightarrow \mathcal{C}^{|\Xi|}(\mathbb{R} )$ such that for any $s\in [0,T]$
and $\varphi \in C^1_c (B(0,1))$ 
\begin{equation}
|(\mathcal{R}F-\Pi_{s}F(s))(\varphi_{s}^{\lambda})|\lesssim \|\Pi\|_T\,\lambda^{\gamma}%
\,.  \label{eq:ReconstructionEstimate1}
\end{equation}
For two different models $(\Pi,\Gamma)$ and $(\overline{\text{$\Pi$}},%
\overline{\Gamma})$ we further have 
\begin{align}  \label{eq:ReconstructionEstimate2}
&\left|(\mathcal{R}F-\Pi_{s}F(s)-(\overline{\mathcal{R}}\,\overline{F}%
-\Pi_{s}\overline{F}(s)))(\varphi_{s}^{\lambda})\right|  \notag \\
&\lesssim\lambda^{\gamma}\,\left(\|F\|_{\mathcal{D}^{\gamma}_T(\Gamma)}\,{%
\vert\kern-0.25ex\vert\kern-0.25ex\vert (\Pi,\Gamma);(\overline{\Pi},%
\overline{\Gamma}) \vert\kern-0.25ex\vert\kern-0.25ex\vert}_{T}+ \|\Pi\|_T{%
\vert\kern-0.25ex\vert\kern-0.25ex\vert F;\overline{F} \vert\kern-0.25ex\vert%
\kern-0.25ex\vert}_{\mathcal{D}^{\gamma}_T(\Gamma);\mathcal{D}^\gamma_T(%
\overline{\Gamma})}\right)
\end{align}
for $F\in\mathcal{D}_T^{\gamma}(\Gamma),\,\overline{F}\in\mathcal{D}%
_T^{\gamma}(\overline{\Gamma})$.
\end{theorem}

As mentioned earlier we want ourselves to work with compactly supported
functions $\varphi \in \mathcal{B}^{\beta}_{1,\infty}(\mathbb{R}
^d),\,\beta>-|\Xi|$ which includes objects like the Haar wavelets. The
following Lemma allows us to carry over all bounds.

\begin{lemma}
\label{lem:BesovBounds} The bounds \eqref{eq:ModelDefinition2}, %
\eqref{eq:ModelDistance}, \eqref{eq:ReconstructionEstimate1} and %
\eqref{eq:ReconstructionEstimate2} do still hold for $\varphi\in \mathcal{B}%
^{\beta}_{1,\infty}(\mathbb{R} ^d),\beta>-|\Xi|$ with compact support in $%
B(0,1)$ (after a change of constants).
\end{lemma}
\begin{remark}
This covers in particular functions like $\mathbf{1}_{[0,1]}\in \mathcal{B}^{1}_{1,\infty} (\mathbb{R})$.
\end{remark}

\begin{proof}
We prove this via wavelet methods in the appendix.
\end{proof}

By the notation $X^{(\varepsilon)}$ we mean in the following both $X$ and $%
X^{\varepsilon}$.

To study objects like $\int_0^t f(\hat{W}^{(\varepsilon)}(r),r) \, \mathrm{d}
W^{(\varepsilon)}( r)$ with the reconstrution theorem we first ``expand''
the integrand $f(\hat{W}^{(\varepsilon)}(r),r)$ in the regularity structure $%
\mathcal{T}$ 
\begin{align*}
F^{(\varepsilon)}(s):=\sum_{m=0}^M \frac{1}{m!} \partial_{1}^{m} f(\hat{W}%
^{(\varepsilon)}(s),s) \mathcal{I}(\Xi)^{m}
\end{align*}
On the level of the regularity structure these objects can be multiplied
with ``noise'' $\Xi$ which gives a modelled distribution on $\mathcal{T}$.

We will analyze $F^{(\varepsilon)}$ by writing it as composition of a (random) modelled distribution with the smooth function $f$.  To this end we need
\begin{lemma} \label{lem:kerneleasy}
On the regularity structure $(\mathcal{T},A,G)$ introduced in Section \ref{sec:basic-pricing-setup}, consider a model $(\Pi,\Gamma)$ which is {\it admissible} in the sense $$\Pi_t\mathcal{I} (\Xi) = (K \ast\Pi_t\Xi)(\cdot) - (K \ast\Pi_t\Xi)(t) \ . $$
Then \begin{equation} \label{equ:firstK}
\mathcal{K} \Xi (t) = \mathcal{I} (\Xi) + (K \ast \Pi_t \Xi)(t) \mathbf{1}
\end{equation}
defines a modelled distribution. More precisely, $\mathcal{K} \Xi \in \mathcal{D}_T^\infty := \bigcup_{\gamma < \infty} \mathcal{D}^\gamma_T$.
\end{lemma} 

\begin{remark}  Our notion of admissibility mimics \cite[Def 5.9]{Hai14}, which however is not directly applicable here (due to failure of Assumption 5.4 in \cite{Hai14}). 
\end{remark} 

\begin{proof}
By definition of the modelled distribution space we need to understand the action of $\Gamma_{st}$ on all constituting symbols. Since $\{ \mathbf{1}, \mathcal{I} (\Xi) \}$ span a sector, i.e. a space invariant by the action of the structure group, it is clear that 
$$
      \Gamma_{st} \mathcal{I} (\Xi)  = \mathcal{I} (\Xi)  + (...) \mathbf{1} .
$$
Application of the realization map $\Pi_s$, followed by evaluation at $s$, immediately identifies $(....)$ as 
$$ 
\Pi_t \mathcal{I} (\Xi)(s) - \Pi_s \mathcal{I} (\Xi)(s) = \Pi_t \mathcal{I} (\Xi)(s) = (K \ast\Pi_s\Xi)(s) - (K \ast\Pi_t\Xi)(t)
$$ 
where we used admissibility and $\Pi_s\Xi = \Pi_t\Xi$ in the last step, a general fact due to the trivial action of the structure group on the symbol with lowest degree. As a consequence $\Gamma_{st} \mathcal{K} \Xi (t) \equiv \mathcal{K} \Xi (s)$, so that, trivially, $\mathcal{K} \Xi \in \mathcal{D}_T^\gamma$ for any $\gamma < \infty$.
\end{proof} 

%
%
For a given (sufficiently smooth) function $f$, and a generic model $(\Pi,\Gamma)$ on our regularity structure, define
\begin{equation*}
F^\Pi : s \mapsto \sum_{m=0}^M \frac{1}{m!} \partial_{1}^{m} f(
(\mathcal{R} \mathcal{K} \Xi(s),s ) 
\mathcal{I}(\Xi)^{m}    \ .
\end{equation*}
Remark that $\mathcal{K} \Xi(s)$ is function-like, i.e. with values in the span of symbols with non-negative degree. From \cite[Prop. 3.28]{Hai14} we then have 
$$ 
\mathcal{R} \mathcal{K} \Xi(s)  =  \langle   \mathcal{K} \Xi(s),  \mathbf{1}  \rangle =  K \ast\Pi_s\Xi \ .   
$$
(In particular, we see that $F^{(\varepsilon)}(s)$ coincides with $F^\Pi$ when $\Pi$ is taken as either approximate or renormalized approximate model.) 
We can also define $\Xi F^\Pi$ simply obtained by multiplying it with $\Xi$. The properties of $F^\Pi$ and $\Xi F^\Pi$ are summarized in the
following lemma. 

\begin{lemma}
\label{lem:FDgamma} Given $f\in C_b^{2M+3}([0,T]\times \mathbb{R} )$, there
exists $N>0$ such that, for all $\gamma$ $\in$ $(1/2+\kappa ,1) , $ 
\begin{align*}
\|F^\Pi\|_{\mathcal{D}_T^\gamma(\Gamma)} \lesssim \|\Pi\|_T^N , &\;\;\;\;
\|\Xi F^\Pi\|_{\mathcal{D}_T^{\gamma+|\Xi|}(\Gamma)} \lesssim \|\Pi\|_T^N \ .
\end{align*}
We have further for two given models $(\Pi,\Gamma)$ and $(\Pi^{\prime
},\Gamma^{\prime })$, 
\begin{align}  \label{eq:ConvergenceF}
{\vert\kern-0.25ex\vert\kern-0.25ex\vert F^\Pi; F^{\Pi^{\prime }} \vert\kern%
-0.25ex\vert\kern-0.25ex\vert}_{\mathcal{D}^\gamma_T(\Gamma);\mathcal{D}%
^\gamma_T(\Gamma^{\prime })} &\lesssim \left( \|\Pi\|_T^N + \|\Pi^{\prime
}\|_T^N\right) {\vert\kern-0.25ex\vert\kern-0.25ex\vert
(\Pi,\Gamma);(\Pi^{\prime },\Gamma^{\prime }) \vert\kern-0.25ex\vert\kern%
-0.25ex\vert}_{T}, \\
{\vert\kern-0.25ex\vert\kern-0.25ex\vert \Xi F^\Pi; \Xi F^{\Pi^{\prime }}
\vert\kern-0.25ex\vert\kern-0.25ex\vert}_{\mathcal{D}^{\gamma+|\Xi|}_T(%
\Gamma);\mathcal{D}^{\gamma+|\Xi|}_T(\Gamma^{\prime })} &\lesssim \left(
\|\Pi\|_T^N + \|\Pi^{\prime }\|_T^N\right) {\vert\kern-0.25ex\vert\kern%
-0.25ex\vert (\Pi,\Gamma);(\Pi^{\prime },\Gamma^{\prime }) \vert\kern%
-0.25ex\vert\kern-0.25ex\vert}_{T} \ ,
\end{align}
where the proportionality constants are, in particular, uniform over all $f$ with bounded $C^{2M+3}$-norm.
\end{lemma}

\begin{proof}
  The map $F^\Pi$ is simply the composition (in the sense of
  \cite[Sec. 4.2]{Hai14}) of the function $f$ with the, thanks to the previous lemma, modelled distributions $%
  \mathcal{K} \Xi$ and $s \mapsto s \mathbf{1}$. The result then follows from
  \cite[Thm 4.16]{Hai14} (polynomial dependence in $\|\Pi\|_T$
  is not stated there but is clear from the proof).
\end{proof}

\begin{remark} \label{rem:unbdf}
In the case when $f$ $\in$ $C^{2M+3}$ but with no global bounds, the result
still holds since we only consider the values of $f$ on the range of the
continuous function $\mathcal{R} \mathcal{K} \Xi$ (which is bounded by some $%
R\geq 0$). The resulting bounds then depend linearly on $\|f\|_{C^{2M+3}(B_R
\times [0,T])}$. 
\end{remark}

In the case of the It\^o model $(\Pi,\Gamma)$ (resp. the approximating
renormalized models $(\hat{\Pi}^\varepsilon,\Gamma^\varepsilon)$) we simply
denote $F^\Pi$ by $F$ (resp. $F^\varepsilon$). We are then allowed to apply
Hairer's reconstruction Theorem \ref{thm:Reconstruction}. Note that since we
have two models we have two reconstruction operators $\mathcal{R}$ and $%
\mathcal{R}^\varepsilon$. The objects $\mathcal{R}^{(\varepsilon)}\Xi
F^{(\varepsilon)}$ can be written down explicitely.

\begin{lemma}
\label{lem:ReconstructionIdentity}We have (a.s.) 
\begin{eqnarray*}
\mathcal{R}F\Xi(\varphi) & = & \int\,\varphi(t)\,f(\hat{W}(t),t) \,\mathrm{d}
W(t)\,, \\
\mathcal{R}^{\varepsilon}F^{\varepsilon}\Xi(\varphi) & = &
\int\,\varphi(t)\,f(\hat{W}^{\varepsilon}(t),t)\, \mathrm{d} 
W^\varepsilon(t)-\int\,\mathscr{K}^{\varepsilon}(t,t)\partial_{1}f(\hat{W}%
^\varepsilon(t),t)\varphi(t)\,\mathrm{d}  t\,.
\end{eqnarray*}
\end{lemma}

\begin{proof}
The proof is in the appendix.
\end{proof}

If we take $\varphi=\mathbf{1}_{[0,T)}$ we obtain $\mathcal{R}F\Xi(\mathbf{1}%
_{[0,T)})=\int_0^T\,\,f(\hat{W}(t),t) \,\mathrm{d}  W(t)$, so that it is
natural to choose 
$ {\tilde{\mathscr{I}}^\eps_f(T)} =\mathcal{R}^\varepsilon
\Xi F^{\varepsilon}(\mathbf{1}_{[0,T)})$ as an approximation. However, note
that the key property of the reconstruction operator $\mathcal{R}%
^{(\varepsilon)}$ is that it is locally close to the corresponding model $%
\Pi^{(\varepsilon)}$ so that we have in fact two natural approximations:

\begin{definition}
\label{def:DiscreteIntegral}For $F,\,F^{\varepsilon}$ as in Lemma \ref%
{lem:FDgamma} and $t\geq0$ we set 
\begin{align*}
{\tilde{\mathscr{I}}^\eps_f(t)} & := \mathcal{R}^{\varepsilon}\Xi F^{\varepsilon}(%
\mathbf{1}_{[0,t]})= \int_{0}^{t}\,f(\hat{W}^{\varepsilon}(r),r)\, \mathrm{d}
W^\varepsilon(r)-\int_{0}^{t}\mathscr{C}^{\varepsilon}(r)\partial_{1}f(%
\hat{W}^\varepsilon(r),r)\, \mathrm{d}  r\,.
\end{align*} 

For a (fixed) partition $\{[t_{l}^{\varepsilon},t_{l+1}^{\varepsilon})\}$ of $[0,t)$
with $\left|t_{l+1}^{\varepsilon}-t_{l}^{\varepsilon}\right|\lesssim%
\varepsilon$ we further set 
\begin{align*}
{\tilde{\mathscr{J}}^\eps_{f,M}(t)} & = \sum_{[t_{l}^\varepsilon,t_{l+1}^\varepsilon)} \hat{\Pi}%
_{t_{l}}^{\varepsilon} \Xi F_{t_{l}}^{\varepsilon}(\mathbf{1}_{[t_{l}^{\varepsilon},t_{l+1}^{%
\varepsilon})})  \label{eq:IDiscrete} \\
& = \sum_{[t_{l}^\varepsilon,t_{l+1}^\varepsilon)}\sum_{m=0}^{M}\frac{1}{m!}\partial_{1}^{m}f(\hat{W}%
^{\varepsilon}(t_{l}^\varepsilon),t_{l}^\varepsilon)\int_{t_{l}^\varepsilon}^{t_{l+1}^\varepsilon}\left(\hat{W}%
^{\varepsilon}(r)-\hat{W}^{\varepsilon}(t_{l}^\varepsilon)\right)^{m} \mathrm{d} 
W^\varepsilon(r) -  \notag \\
& -\sum_{m=1}^{M}\frac{1}{(m-1)!}\partial_{1}^{m}f(\hat{W}%
^{\varepsilon}(t_{l}^\varepsilon),t_{l}^\varepsilon)\int_{t_{l}^\varepsilon}^{t_{l+1}^\varepsilon}\mathscr{C}%
^{\varepsilon}(r)\,\left(\hat{W}^{\varepsilon}(r)-\hat{W}%
^{\varepsilon}(t_{l}^\varepsilon)\right)^{m-1}\,\mathrm{d}  r \,.  \notag
\end{align*}
\end{definition}
We might drop the indices $f$ and $f,M$ on $\tilde{\mathscr{I}}^\varepsilon$ and $\tilde{\mathscr{J}}^\varepsilon$ if there is no risk of confusion. 

The following theorem, which can be seen as the fundamental theorem of our
regularity structure approach to rough pricing shows that these
approximations do both converge.

\begin{theorem} 
\label{thm:ConvergenceDiscreteIntegrals} Fix $T>0$. For $f$ smooth, bounded with bounded derivatives, and $\tilde{\mathscr{I}}^\eps_f, \ \tilde{\mathscr{J}}^\eps_{f,M}$ as in Definition %
\ref{def:DiscreteIntegral} we have 

\begin{enumerate}[(i)]

\item for any $\delta\in (0,1)$ and any $p< \infty$ there exists $C$ such that
\begin{equation}
\left\|
\sup_{t\in[0,T]}\left|{\tilde{\mathscr{I}}^\eps_f(t)}-\int_{0}^{t}f(\hat{W}%
(r),r)\mathrm{d}  W(r)\right|   \right\|_{L^p}   \le C   \varepsilon^{\delta H}\,,
\label{eq:ConvergenceContinuousApproximation}
\end{equation}

\item for every $\delta \in (0,1)$ we can pick $M=M(\delta, H)$ large enough, such that, for any $p< \infty$ there exists $C$ such that
\begin{equation}
\left\|
\sup_{t\in[0,T]}\left|{\tilde{\mathscr{J}}^\eps_{f,M}(t)}-\int_{0}^{t}f(\hat{W}(r),r)%
\mathrm{d}  W(r)\right|   \right\|_{L^p}   \le C \varepsilon^{\delta H}\,.
\label{eq:ConvergenceDiscreteApproximation}
\end{equation}
\end{enumerate}
\end{theorem}

\begin{remark} With regard to (i): although ${\tilde{\mathscr{I}}^\eps_f(t)}$ does not depend on any choice of $M$, and nor does its (It\^o) limit, 
the choice of $M$ affects the entire regularity structure and so, implicitly also the reconstruction operator $\mathcal{R}^{\varepsilon}$
used in the definition of $\tilde{\mathscr{I}}^\eps_f$, as well as the modelled distribution $F^\eps$. The latter, in turn, requires $f \in C^M$ 
for the construction to make sense. If $\delta$ is chosen arbitrarily close to one, $f$ needs to have derivatives of arbitrary order, hence our smoothness assumption.
\end{remark} 

\begin{remark} ($f$ of exponential form; \cite{GJR14}) By an easy localization argument one shows that for $f$ smooth (but without any further bounds) ones still has 
$$
   \sup_{\varepsilon \in (0,1]} \mathbb{P} \left(   \sup_{t\in[0,T]}\left|{\tilde{\mathscr{I}}^\eps_f(t)}-\int_{0}^{t}f(\hat{W}%
(r),r)\mathrm{d}  W(r)\right|   \le C   \varepsilon^{\delta H}   \right) \to 0 
$$ with $C \to \infty$. 
The original rough vol model due to \cite{GJR14} makes a point that $f$ should be of exponential form. Now, the result with $L^p$-estimates
still holds since we only consider the values of $f$ on the range of the
continuous function $\mathcal{R} \mathcal{K} \Xi$ (which is bounded by some $%
R\geq 0$). As pointed out in Remark \ref{rem:unbdf}, the bounds then depend linearly on $\|f\|_{C^{M+2}(B_R
\times [0,T])}$. Since, for us, $(\Pi,\Gamma)$ is always a Gaussian
model, $\mathcal{R} \mathcal{K} \Xi$ is a Gaussian process (say, $\hat W$ or $\hat{W}^\varepsilon)$ hence
we have (Fernique) Gaussian concentration for $\sup_{t \in [0,T]}  |\mathcal{R} \mathcal{K} \Xi (t)  | $.
So, for instance if $f$ and its derivatives have exponential growth we do have the $L^p$ bounds of the above theorem, for all $p<\infty$.
This remark justifies in particular the choice $f(x) = \exp (x)$ and $p=2$ in the numerical discussion of Section \ref{sec:ChristianBen}. 
\end{remark}

\begin{proof}
Without loss of generality $T\leq 1$, otherwise split $[0,T]$ in
subintervals. Let us show \eqref{eq:ConvergenceContinuousApproximation}.  
\begin{align*}
&{\tilde{\mathscr{I}}^\eps_f(t)}-\int_0^t f(\hat{W}(r),r) \mathrm{d}  W(r)=(%
\mathcal{R}^\varepsilon (F^\varepsilon\Xi)-\mathcal{R}(F\Xi))(\mathbf{1}%
_{[0,t]}) \\
&=t \left(\hat{\Pi}^\varepsilon_{0} \Xi F^\varepsilon(0)-\Pi_0 \Xi
F(0))\right)( t^{-1}\mathbf{1}_{[0,t]}) \\
&+ t \left(\mathcal{R}^\varepsilon \Xi F^\varepsilon-\hat{\Pi}%
^\varepsilon_{0} \Xi F^\varepsilon(0) -(\mathcal{R} \Xi F-\Pi_0 \Xi
F(0))\right)( t^{-1}\mathbf{1}_{[0,t)})\,.
\end{align*}
We then obtain the rate $\varepsilon^{\delta \kappa},\delta\in (0,1)$ using
Theorem \ref{prop:ModelConvergence}, Lemma \ref{lem:FDgamma} and %
\eqref{eq:ModelDistance} for the first term and also Theorem \ref%
{thm:Reconstruction} for the second term. 
Letting $\kappa \uparrow H$ and $M\uparrow \infty$ our total rate can be
chosen arbitrary close to $H$.

To obtain the second estimate we can bound 
$\tilde{\mathscr{I}}^\eps_f (t) -  \tilde{\mathscr{J}}^\eps_{f,M} (t)$
with the first inequality in Theorem \ref%
{thm:Reconstruction}.
\end{proof}

\bigskip

{\bf Non-constant vs. constant renormalization}   

If $\delta^\varepsilon$ comes from a mollifier (cf. Example \ref{rem:Mollifier}) the renormalization
$  \mathscr{C}^\varepsilon = \mathscr{K}^\varepsilon(\cdot,\cdot)$ that was applied in Theorem \ref{prop:ModelConvergence} 
and thus in Definition \ref{def:DiscreteIntegral}
is a constant, which is the familiar concept one encounters in the study of
singular SPDE \cite{Hai13, Hai14, CH16}. 
If $\delta^\varepsilon$ comes from wavelets such as the Haar
basis, $\mathscr{K}^\varepsilon(\cdot,\cdot)$ is usually not constant but a
periodic function with period $\varepsilon$. Thus we see that our analysis
gives rise to a ``non-constant renormalization''. It is natural to ask if one can do with 
constant renormalization after all. For the sake of argument, consider $  \mathscr{C}^\varepsilon $, periodic with period $\varepsilon$, with mean
$$
     C_\eps = \frac{1}{\eps} \int_0^\eps \mathscr{C}^\varepsilon (t) dt .
$$
From Lemma \ref{lem:ApproximateVolterraEstimate} it follows that $  \mathscr{C}^\varepsilon $ (and its mean) are bounded by $\eps^{H-1/2}$, uniformly in $t$. Putting all this together
it easily follows that 
$ | \langle \mathscr{C}^\varepsilon - C_\eps, \varphi  \rangle | \lesssim \eps^{\alpha + H -1/2}$, 
uniformly over all $\varphi$ bounded in $C^\alpha$, with convergence to zero when $\alpha  > 1/2-H$. 
As a consequence, taking $\varphi (t) = f (\hat W^\eps)$, for smooth $f$, we clearly can apply this with any $\alpha < H$. 
Hence, by equating the constraints on $\alpha$, we arrive at $H > 1/4$. 
The practical consequence then is, with focus on the convergence stated in part (i) of 
Theorem \ref{eq:ConvergenceContinuousApproximation} that we can indeed replace non-constant renormalization by a constant, 
however at the prize of restricting to $H>1/4$ and with an according loss on the convergence rate. 
Interestingly, our numerical simulation suggest that no loss occurs and constant renormalization works for any $H>0$. While we have refrained from investigation this (technical) point further,
\footnote{Some computations led us to believe that this question can be settled with the aid of mixed $(1,\rho)$-variation of the covariance function of the Volterra process, cf. \cite{FGGR16}, which we expected to hold uniformly over approximation. However the amount of work seems in no relation to the main theme of this article.}  we can understand the mechanism at work by looking at the following toy example: Consider the Ito-integral $\int_0^1 W^h dW$ where $W^H$ is a fBM, but now with Hurst parameter $H>1/2$, built, say, as Volterra process over $W$. Using Young integration theory, one can give a pathwise argument that shows that Riemann-Stieltjes approximation converge a.s. (with vanishing rate as $H \to 1/2^+$). However, we know from stochastic theory (It\^o integration) that this convergence works in $L^2$ (and then in probability) for any $H>0$. We would thus expect that, when $H \in (0,1/4]$, constant renormalization is still valid, but now the difference only vanishes in mean-square sense (which is what we did in the numerics section).

\subsection{The case of the Haar basis}
\label{subsec:haar}

The following special case of the approximations above to $\int_0^t f(\hat{W}%
(r),r) \mathrm{d}  W(r)$ is of particular interest for our purposes. We here
collect some more concrete formulas that arise in this case.

Let $\varepsilon=2^{-N}$, $\phi:=\mathbf{1}_{[0,1)}$ and $\phi_{l,N}=2^{N/2}
\phi(2^N \cdot-l),l\in\mathbb{Z} $ and the corresponding $\delta^\varepsilon$
coming from this wavelet is then for $x,y\in\mathbb{R} $. 
\begin{align*}
\delta^\varepsilon(x,y)=\sum_{l\in \mathbb{Z} }
\phi_{l,N}(x)\phi_{l,N}(y)=2^{N} \mathbf{1}_{[\lfloor x 2^N \rfloor
2^{-N},(\lfloor x 2^N \rfloor+1) 2^{-N})}(y)
\end{align*}
The mollified Volterra-kernel \eqref{eq:MollifiedVolterraKernel} then takes
the form 
\begin{align*}
\mathscr{K}^\varepsilon(u,v) &=\int_0^\infty \int_0^\infty
\delta^\varepsilon(v,x_1) \delta^\varepsilon(x_1,x_2) K(u-x_2) \mathrm{d} 
x_1 \mathrm{d}  x_2 \\
&= \sqrt{2H} \cdot 2^{N} \int_{[\lfloor v 2^N \rfloor 2^{-N},(\lfloor v 2^N
\rfloor+1) 2^{-N}\wedge u )} |u-x|^{H-1/2} \mathbf{1}_{\lfloor v 2^N \rfloor
2^{-N}\leq u} \,\mathrm{d}  x \\
&=\frac{\sqrt{2H}}{1/2+H} 2^N \times \\
&\times \left( |u-\lfloor v 2^N \rfloor 2^{-N}|^{1/2+H}-|u-(\lfloor v 2^N
\rfloor+1) 2^{-N}\wedge u )|^{1/2+H} \right) \mathbf{1}_{\lfloor v 2^N
\rfloor 2^{-N}\leq u}\,.
\end{align*}
A special role is played by \textit{diagonal function} as a renormalization,
\begin{align}  \label{eq:RenormalizationHaarBasis}
\mathscr{C}^\varepsilon(t) = \mathscr{K}^\varepsilon(t,t) =\frac{\sqrt{2H}\,2^N}{1/2+H} |t-\lfloor t
2^N \rfloor 2^{-N}|^{1/2+H}\,.
\end{align}
We have moreover 
\begin{align*}
\hat{W}^\varepsilon(t) &=\int_0^t K(t-r) \,\mathrm{d} 
W^\varepsilon(r)=\sum_{l=0}^\infty Z_l \int_0^t K(t-r) \phi_{k,N}(r) \mathrm{%
d}  r \\
&=\sum_{l=0}^\infty 2^{-N/2} \mathscr{K}^\varepsilon(t,l2^{-N}) Z_l
=\sum_{l=0}^{\lfloor t 2^{N} \rfloor} 2^{-N/2} \mathscr{K}%
^\varepsilon(t,l2^{-N}) Z_l\,,
\end{align*}
where $Z_l=\langle \dot{W}, \phi_{l,N} \rangle$ are i.i.d. $N(0,1)$
variables. As approximation we can finally take $\mathscr{I}_f^{\varepsilon} (t)$ from
Definition \ref{def:DiscreteIntegral} with partition $\{[t_l,t_{l+1})\}=%
\{[l2^{-N},(l+1)2^{-N} \wedge t)\}$ which gives us 
\begin{align*}
{\tilde{\mathscr{J}}^\eps_{f,M}(t)} &= \sum_{l=0}^{\lceil t 2^N \rceil-1}\sum_{m=0}^{M}%
\frac{1}{m!}\partial_{1}^{m}f(\hat{W}^{\varepsilon}(t_{l}),t_{l}) 2^{N/2}
Z_l \int_{t_{l}}^{t_{l+1}}\,\left(\hat{W}^{\varepsilon}(r)-\hat{W}%
^{\varepsilon}(t_{l})\right)^{m}\,\mathrm{d}  r - \\
& -\sum_{m=1}^{M}\frac{1}{(m-1)!}\partial_{1}^{m}f(\hat{W}%
^{\varepsilon}(t_{l}),t_{l+1})\int_{t_{l}}^{t_{l+1}}\mathscr{C}
^{\varepsilon}(r)\,\left(\hat{W}^{\varepsilon}(r)-\hat{W}%
^{\varepsilon}(t_{l})\right)^{m-1}\,\mathrm{d}  r
\end{align*}
and 
\begin{align*}
{\tilde{\mathscr{I}}^\eps_f(t)}=\sum_{l=0}^{\lceil t 2^N \rceil -1}
\int_{t_l}^{t_{l+1}} [2^{N/2} Z_l \cdot f(\hat{W}^\varepsilon(r),r) \,%
\mathrm{d}  r -\mathscr{C}^\varepsilon(r) \,\partial_1 f(\hat{W}^
\varepsilon(r),r) ]\,\mathrm{d}  r \,.
\end{align*}
As explained at the end of the last section, $\mathscr{C}^\varepsilon(r)$ in these formulas could be replaced by its local mean, the constant
\begin{equation*}
2^N \int_{0}^{2^{-N}} \mathscr{C}^\varepsilon(r) \, \mathrm{d}  r = \frac{%
\sqrt{2H}}{(H+1/2)(H+3/2)} 2^{N(1/2-H)}\,.
\end{equation*}

\section{The full rough volatility regularity structure}  
\label{sec:full-rough-vol-reg-struc}

\subsection{Basic setup}
\label{sec:basic-setup-full-rough-vol-reg-struc}

We want to add an independent Brownian motion, so that we take an additional
symbol $\bar{\Xi}$. We again fix $M$ and define a (larger) collection of symbols $\bar{S}$, with $S \subset \bar{S}$, and then

\begin{align}
\bar{\mathcal{T}} = \bigoplus_{\tau \in \bar{S}} \mathbb{R}\tau \cong \mathcal{T} + \left\langle \{ \bar{\Xi}, \bar{\Xi} \mathcal{%
I}({\Xi}),\ldots,\bar{\Xi} \mathcal{I}(\Xi)^M \}\right\rangle .
\end{align}
Again we fix $|\bar{\Xi}| =-1/2-\kappa$ and the homogeneity of the other
symbols are defined multiplicatively as before.

Also as before, we set $\hat{W}%
_t=\int_0^t K(s,t) dW_s$ with $K(s,t) = \sqrt{2H} |t-s|^{H-1/2}\mathbf{1}_{t>s}$, where $W$ and also $\bar{W}$ are independent Brownian motions.

We extend the canonical model $(\Pi,\Gamma)$ to this regularity structure by
defining 
\begin{equation*}
\Pi_s \bar{\Xi} \mathcal{I}(\Xi)^m  =
\left\{ t \mapsto 
 \frac{\mathrm{d} }{\mathrm{d}  t}
\left(\int_s^t \left(\hat{W}(u) - \hat{W}(s)\right)^m d\bar{W}(u) \right) \right\}
\end{equation*}
(the above integral being in It\^o sense), and \footnote{Upon setting $\Gamma_{ts}\left( \bar{\Xi} \right) = \bar{\Xi}$, the given relation is precisely implied by multiplicativity of $\Gamma$.}
\begin{equation*}
\Gamma_{ts}\left( \bar{\Xi} \mathcal{I}(\Xi)^m\right) = \bar{\Xi}
\Gamma_{ts}\left( \mathcal{I}(\Xi)^m\right).
\end{equation*}

Arguments similar to the proof of Lemma \ref{lem:LimitingModel} show that
this indeed defines a model on $\bar{\mathcal{T}}$.

\subsection{Small noise model large deviation}
\label{sec:small-noise-model-ld}

Given $\delta>0$ we consider the "small-noise" model $(\Pi^\delta,\Gamma^%
\delta)$ on $\tilde{T}$ obtained by replacing $W, \bar{W}$ by $\delta W,
\delta \bar{W}$, which simply means that 
\begin{align*}
\Pi^\delta \mathbf{1}&=1 \\
\Pi^\delta \mathcal{I}(\Xi)^m &= \delta^{m} \Pi\mathcal{I}(\Xi)^m \\
\Pi^\delta \Xi \mathcal{I}(\Xi)^m &=\delta^{m + 1} \Pi \Xi \mathcal{I}(\Xi)^m
\\
\Pi^\delta \bar{\Xi} \mathcal{I}(\Xi)^m &=\delta^{m + 1} \Pi \bar{\Xi} 
\mathcal{I}(\Xi)^m,
\end{align*}
and 
\begin{align*}
\Gamma^\delta_{ts} \mathbf{1}& =\mathbf{1},\Gamma^\delta_{ts} \Xi
=\Xi,\Gamma^\delta_{ts} \bar{\Xi}=\bar{\Xi}, \\
\Gamma^\delta_{ts} \mathcal{I}(\Xi)&=\mathcal{I}(\Xi)+\delta (\hat{W}(t)-%
\hat{W}(s))\mathbf{1} \\
\Gamma^\delta_{ts} \tau \tau^{\prime} &= \Gamma^\delta_{ts}\tau \cdot
\Gamma^\delta_{ts}\tau^{\prime }\,,\,\,\mbox{ for } \tau,\tau^{\prime }\in 
\bar{S} .
\end{align*}

Finally, for $h=(h_1,h_2)$ in $\mathcal{H}:=L^2([0,T])^2$, we consider the
deterministic model $(\Pi^h,\Gamma^h)$ defined by 
\begin{align*}
\Pi^h\mathbf{1}&=1, \\
\Pi^h_s \Xi = h_1, & \;\; \Pi^h_s \bar{\Xi} = h_2, \\
\Pi^h_s \mathcal{I}(\Xi) (t)&= \int_0^{t \vee s} (K(u,t) - K(u,s)) h_1(u) du,
\\
\Pi^h \tau \tau^{\prime} &= \Pi^h\tau \Pi^h \tau ^{\prime }\mbox{ for }
\tau,\tau^{\prime }\in \bar{S}
\end{align*}
and 
\begin{align*}
\Gamma^h_{ts} \mathbf{1}& =\mathbf{1},\Gamma^h_{ts} \Xi =\Xi,\Gamma^h_{ts} 
\bar{\Xi}=\bar{\Xi}, \\
\Gamma^h_{ts} \mathcal{I}(\Xi)&=\mathcal{I}(\Xi)+ (\int_0^{t \vee s} (K(u,t)
- K(u,s)) h_1(u) du)\mathbf{1} \\
\Gamma^h_{ts} \tau \tau^{\prime} &= \Gamma^h_{ts}\tau \cdot \Gamma^h_{ts}\tau^{\prime
}\,,\,\,\mbox{ for } \tau,\tau^{\prime }\in \bar{S}.
\end{align*}
The following lemma and theorem are proved in Appendix \ref{app:LD}.
\begin{lemma}
\label{lem:Pih} For each $h \in \mathcal{H}$, $\Pi^h$ does define a model.
In addition, the map $h \in \mathcal{H} \mapsto \Pi^h$ is continuous.
\end{lemma}

\begin{theorem}
\label{thm:LD} The models ${\Pi}^{\delta}$ satisfy a large deviation
principle (LDP) in the space of models with rate $\delta^2$ and rate
function given by 
\begin{equation*}
J(\Pi) = \left\{%
\begin{array}{ll}
\frac{1}{2} \|h\|_{\mathcal{H}}^2 & \mbox{if }\Pi = \Pi^h \mbox{ for some }
h \in \mathcal{H}, \\ 
+\infty, & otherwise.%
\end{array}%
\right..
\end{equation*}
\end{theorem}

As an immediate corollary we have 

\begin{corollary}
\label{cor:LD} For $\delta$ small, $  \mathbb{P} ( Y_1^\delta \approx y)  \approx \exp [ - I (y) / \delta^2 ] $, 
in the precise sense of a large deviation principle (LDP) for  
\begin{equation*}
Y^\delta_1 :=\int_0^1 f( \delta^H \hat{W}_s) \delta \left( \rho dW_s + \bar{%
\rho} d\bar{W}_s \right)
\end{equation*}
with speed $\delta^2$, and rate function given by 
\begin{equation}  \label{eq:rate}
I(y) = \inf_{h_1 \in L^2([0,1])} \{ \frac{1}{2} \|h_1\|_{L^2}^2 + \frac{%
\left(y-I_1(h_1)\right)^2}{2I_2(h_1)} \} 
\end{equation}
where $$I_1(h_1)= \rho \int_0^1 f\left(\int_0^s K(u,s) h_1(u) du\right) h_1(s)
ds, \ \ \ \ I_2(h_1)= \int_0^1 f\left(\int_0^s K(u,s) h_1(u) du\right)^2 ds \ . $$
\end{corollary}

\begin{remark} This improves a similar result in \cite{FZ17} in the sense that $f$ of exponential form, as required in rough volatility modelling 
\cite{GJR14, BFG16, BFGHS17x}, is now covered.
\end{remark}
\begin{proof}
Note that 
\begin{equation*}
Y^\delta_1 = \left\langle \mathcal{R}^\delta F^\delta \cdot (\rho \Xi + \bar{%
\rho} \bar{\Xi}), 1_{[0,1]} \right\rangle
\end{equation*}
where $F^\delta \equiv F^{\Pi^\delta}$ as defined in Lemma \ref{lem:FDgamma}. By
the contraction principle and the continuity estimate from Theorem \ref%
{thm:Reconstruction}, it holds that $Y^\delta_1$ satisfies a LDP, with rate
function given by 
\begin{equation*}
I(y) = \inf \{ \frac{1}{2} \left( \|h_1\|_{L^2}^2 + \|h_2\|^2_{L^2} \right),
\;\;\; y = \left\langle \mathcal{R}^h F^h \cdot (\rho \Xi + \bar{\rho} \bar{%
\Xi}), 1_{[0,1]} \right\rangle \},
\end{equation*}

where we used $F^h \equiv F^{\Pi^h}$.  It then suffices to note that
\begin{equation*}
\left\langle 
\mathcal{R}^h \left( F^h \cdot (\rho \Xi + \bar{\rho} \bar{\Xi})\right), 1_{[0,1]}
\right\rangle = \int_0^1 f\left(\int_0^s K(u,s) h_1(u) du\right) \left(\rho
h_1(s) ds + \bar{\rho} h_2(s) ds\right)
\end{equation*}
and optimizing over $h_2$ for fixed $h_1$ we obtain \eqref{eq:rate}.
\end{proof}

We note that thanks to Brownian resp. fractional Brownian scaling, small noise large deviations translate 
immediately to short time large deviations, cf. \cite{FZ17}.

Although the rate function here is not given in a very useful form, it is
possible \cite{BFGHS17x} to expand it in small $y$ and so compute
(explicitly in terms of the model parameters) higher order moderate
deviations which relate to implied volatility skew expansions.

\section{Rough Volterra dynamics for volatility}   \label{sec:Volt}

\subsection{Motivation from market micro-structure} 

Rosenbaum and coworkers, \cite{EEFR16x, EER16x,EER17x}, show that stylized facts of modern market microstructure naturally give rise to fractional dynamics and leverage effects. Specifically, they construct a sequence of Hawkes processes suitably rescaled in time and space that converges in law to a rough volatility model of rough Heston form
\begin{eqnarray}
dS_{t}/S_{t} &=&\sqrt{v_t}dB_{t}\equiv \sqrt{v}\left( \rho dW_{t}+\bar{\rho}d%
\bar{W}_{t}\right) \text{ } \ , \label{equ:rHeston} \\
v_t &=&  v_0 + \int_0^t \frac{a-bv_s}{(t-s)^{1/2-H}} ds + \int_0^t 
\frac{c\sqrt{v_s}}{(t-s)^{1/2-H}}dW_s \nonumber \ .
\end{eqnarray}
(As earlier, $W,\bar{W}$ independent Brownians.) Similar to the case of the classical Heston model, the square-root provides both pain (with regard to any methods that rely on sufficient smooth coefficients) and comfort (an {\it affine structure}, here infinite-dimensional, which allows for closed form computations of moment-generating functions).  Arguably, there is no real financial reason for the square-root dynamics\footnote{This is also a frequent remark for the classical Heston model.} and ongoing work attempts to modify the above square-root dynamics, such as to obtain (something close to) {\it log-normal volatility}, put forward as important rough volatility feature by Gatheral et al. \cite{GJR14}. This motivates the study of more general dynamic rough volatility models of the form

\begin{eqnarray}
dS_{t}/S_{t} &=&f({Z_t})dB_{t}\equiv f({Z_t}) \left( \rho dW_{t}+\bar{\rho}d\bar{W}_{t}\right) \text{ } \ , \label{equ:grHeston} \\
Z_t &=&  z + \int_0^t K(s,t)  v(Z_s) ds +  \int_0^t K(s,t)  u(Z_s) dW_s \label{eq:Volt}
\end{eqnarray}%
with sufficiently nice functions $f, u,v$. (While $f(x)=\sqrt{x}$ is still OK in what follows, we assume $u,v \in C^3$ for a local solution theory and then in fact impose $u,v \in C_b^3$ for global existence. (One clearly expects non-explosion under e.g. linear growth, but in order not to stray too far from our main line of investigation we refrain from a discussion.) Remark that $f(z)$ plays the role of spot-volatility. Further note that the 
choice $z=0, \ v \equiv 0, \ u \equiv 1$ brings us back to the ``simple'' case with (rough stochastic) volatility $f(Z_t) = f(\hat W_t)$ considered in earlier sections.

\bigskip

With some good will,\footnote{We are not aware of any literature on mixed It\^o-Volterra systems (although expect no difficulties). Here of course, it suffices to first solves for $Z$ and then construct $S$ as stochastic exponential.}
equation (\ref{equ:grHeston}) fits into the existing theory of stochastic Volterra equations with singular kernels (e.g. \cite{PP90} or \cite{CD01}). 
%

\subsection{Regularity structure approach} \label{sec:VolterraStuff}

We insist that (\ref{equ:grHeston}) is not a classical It\^o-SDE (solutions will not be semimartingales), nor a rough differential equations (in the sense of rough paths, driven by a Gaussian rough path as in \cite[Ch.10]{FH14}).
If rough paths have established themselves as a powerful tool to analyze classical It\^o-SDE, we here make the point that Hairer's theroy is an equally powerful tool to analyze stochastic Volterra (resp. mixed It\^o-Volterra) equations in the singular regime of interest.

As preliminary step, we have to have to find the correct model space, spanned by symbols which arise by formal Picard iteration. To this end, rewrite  \eqref{equ:grHeston} formally, or as equation for modelled distributions,
\begin{equation}  \label{eq:VoltD}
\mathcal{Z} =  \mathcal{I} (U(\mathcal{Z}) \cdot \Xi  ) + (....)  \mathbf{1}
\end{equation}
from which one can guess (or formally derive along \cite[Sec. 8.1]{Hai14}) the need for the symbols
$$
          \mathbf{1},   \mathcal{I} (\Xi) ,  \mathcal{I} (\Xi)^2, \mathcal{I} (\Xi \mathcal{I} (\Xi) ), ...
$$
We have degrees $|\mathbf{1}|=0, \, |\mathcal{I}(\Xi)| = H - \kappa$ and then, for subsequent symbols, degree computed as
$$
                       ( 1/2+ H) \times \{ \text{number of $\mathcal{I}$} \} + (-1/2 + \kappa)  \times \{ \text{number of $\Xi$} \}.
$$ 
For a modelled distribution, $\mathcal{Z} (t)$ takes values in the linear span of sufficiently many symbols, the (minimal) number of which is dictated by the Hurst parameter $H$. Loosely speaking, $\mathcal{Z} \in \mathcal{D}^\gamma$ indicates an expansions with $\gamma$-error estimate, in practice easy to see from the degree of the lowest degree symbols that do not figure in the expansion. For example, in case of a ``level-$2$ expansion'' we can expect 
$$ 
         \mathcal{Z} (t) = (....)  \mathbf{1} + (....) \mathcal{I} (\Xi)   \in  \mathcal{D}_0^{2 (H - \kappa)}
$$
since  $| \mathcal{I} (\Xi)^2 | = |  \mathcal{I} (\Xi \mathcal{I} (\Xi) )| = 2H - 2\kappa.$ It follows from general theory \cite[Thm 4.16]{Hai14} that if $\mathcal{Z}  \in  \mathcal{D}_0^\gamma$, then so is $U(\mathcal{Z})$, the composition with a smooth function, and by \cite[Thm 4.7]{Hai14} the product with $\Xi \in  \mathcal{D}_{-1/2-\kappa}^\infty$ is a modelled distribution in $ \mathcal{D}^{\gamma-1/2-\kappa}$. For both reconstruction and convolution with singular kernels, one needs modelled distributions with positive degree $\gamma-1/2-\kappa>0$. Given $H \in (0,1/2]$ we can then determine which symbols (up to which degree) are required in the expansion. As earlier, fix an integer 
$$
        M \ge  \max \{ m \in \mathbb{N} | m  \cdot  (H - \kappa) - 1/2 - \kappa \le 0 \}
$$
(so that $(M+1). (H - \kappa) - 1/2 - \kappa >  0$) and see that $ \mathcal{Z} \in \mathcal{D}_0^{(M+1). (H - \kappa)}$ will do. When $H > 1/4$, and by choosing $\kappa >0$ small enough, we see that $M=1$ will do. That is, the symbols required to describe $\mathcal{Z}$ are $\{ \mathbf{1},   \mathcal{I} (\Xi)  \}$ and if one adds the symbols required to describe the right-hand side, one ends up with the level-$2$ model space spanned by
$$
    \{   \Xi,    \Xi \mathcal{I} (\Xi) ,  \mathbf{1},   \mathcal{I} (\Xi) \}
$$
which is exactly the model space for the ``simple'' rough pricing regularity structure, (\ref{equ:SimpleModelSpace}) in case $M=1$. When $H \le 1/4$ this precise correspondence is no longer true. To wit, in case $H \in (1/3, 1/4]$, taking $M=2$ accordingly, solving (\ref{eq:Volt}) on the level of modelled distributions will require a (``level-$3$'') model space given by 

$$
    \langle   \Xi,    \Xi \mathcal{I} (\Xi) ,  \Xi \mathcal{I} (\Xi)^2, \Xi \mathcal{I} (\Xi \mathcal{I} (\Xi) ),  \mathbf{1},   \mathcal{I} (\Xi), \mathcal{I} (\Xi)^2, \mathcal{I} (\Xi \mathcal{I} (\Xi) ) \rangle
$$
which is strictly larger than the corresponding level-$3$ simple model space given in (\ref{equ:SimpleModelSpace}).  In general, one needs to consider an extended model space $\hat{T}= \langle \hat{S} \rangle$, so as to have 
\begin{equation*}
\tau \in \hat{S} \Rightarrow \Xi  \mathcal{I}(\tau)^m , \mathcal{I}(\tau)^m \in \hat{S}, m \geq 0,
\end{equation*}
(with the understanding that only finitely many such symbols are needed, depending on $H$ as explained above). As a result, symbols such as
\begin{equation*}
\Xi \mathcal{I}(\Xi (\mathcal{I}(\Xi))^m), \; m \geq 0,\;\;  \mathcal{I}%
(\Xi (\mathcal{I}(\Xi (\mathcal{I}(\Xi))^{m})^{m^{\prime }}), m, m^{\prime
}\geq 0, \;\; \ldots 
\end{equation*}
will appear. At this stage a tree notation (omnipresent in the works of Hairer) would come in handy and we refer to \cite{BCFP17x} (and the references therein) for a recent attempt to reconcile the tree formalism of branched
 rough path \cite{Gub10, HK15} and the most recent algebraic formalism of regularity structures. (In a nutshell, the simple case (\ref{equ:SimpleModelSpace}) corresponds to trees where one node has $m$ branches; in the 
 present non-simple case symbols branching can happen everywhere.) ) 
Carrying out the following construction in the general case, $H>0$, is certainly possible.\footnote{We note that, as $H \to$ the number of symbols tends to infinity. In comparison, as far as we know, among all recently studied singular SPDEs, only the sine-Gordon equation \cite{HS16} exhibits arbitrarily many symbols. } However, the algebraic complexity is essentially the one from branched rough paths and hence the general case requires a Hopf algebraic (Connes-Kreimer, Grossman-Larson ...) construction of the structure group (a.k.a. positive renormalization). Although this, and negative renormalization, is well understood (\cite{Hai14,BHZ16x}, also \cite{BCFP17x} for a rough path perspective, all complete exposition would lead us to far astray from the main topic of this paper. Hence, for simplicity only, we shall restrict from here on to the level-$2$ case $H>1/4$ (with $M=1$ accordingly) but will mention general results whenever useful.

\subsection{Solving for rough volatility}

We rewrite \eqref{eq:Volt} as equation for modelled distributions in $\mathcal{D}^\gamma$, 
\begin{equation}  \label{eq:VoltD2}
\mathcal{Z} = z \mathbf{1} + \mathcal{K} (U(\mathcal{Z}) \cdot \Xi + V(%
\mathcal{Z}) ).
\end{equation}
(Here $U,V$ are the operators associated to composition with$u,v \in C^{M+2}$ respectively.) We also impose
$$
           \gamma \in (1/2 + \kappa, 1)
$$           
which is clearly necessary such as to have the product  $U(\mathcal{Z}) \cdot \Xi  $ in a modelled distribution space of positive parameter, so that reconstruction, convolution etc. makes sense.
Let $H>1/4, M=1$ and pick $\kappa \in (0, \tfrac{4H - 1}{6})$ so that $(M+1). (H - \kappa) - 1/2 - \kappa >  0$. As explained in the previous section, this exactly allows us to work in the familiar 
structure of Section \ref{sec:basic-pricing-setup}. That is,  with $M=1$,
$$
    \mathcal{T} = \langle   \Xi,    \Xi \mathcal{I} (\Xi) ,  \mathbf{1},   \mathcal{I} (\Xi) \rangle \ .
$$
with index set and structure group as given in that section. This structure is equipped with the It\^o-model, 
and its (renormalization) approximations. Equation (\ref{eq:VoltD2}) critically involves the convolution operator $\mathcal{K}$ acting on $\mathcal{D}^\gamma$. The general construction \cite[Sec. 5]{Hai14} is among the most technical in Hairer's work, and in fact not directly applicable (our kernel $K$, although $\beta$-regularizing with $\beta = 1/2 + H$)  fails the Assumption 5.4 in \cite{Hai14}) so we shall be rather explicit.
\begin{lemma} \label{lem:KisK}
On the regularity structure $(\mathcal{T},A,G)$ of Section \ref{sec:basic-pricing-setup} with $M=1$, consider a model $(\Pi,\Gamma)$ which is {\it admissible} in the sense $$\Pi_t\mathcal{I} (\Xi) = (K \ast\Pi_t\Xi)(\cdot) - (K \ast\Pi_t\Xi)(t) \ . $$ Let $\gamma > 0, \ F \in \mathcal{D^\gamma}$ and set \footnote{$\ \mathcal{I}$ is extended linearly to all of $\mathcal{T}$ by taking $\mathcal{I%
} \tau =0$ for symbols $\tau \neq \Xi$).} 
\begin{equation*}
\mathcal{K} F : s \in [0,T] \mapsto \mathcal{I}(F(s)) + (K \ast \mathcal{R}  F)(s) \mathbf{1}
\end{equation*}
Then (i) $\mathcal{K}$ maps $\mathcal{D}^\gamma \to \mathcal{D}^{\min \{ \gamma+\beta, 1 \} }$ and (ii) $ \mathcal{R} ( \mathcal{K} F ) = K * \mathcal{R} F$, i.e. convolution commutes with reconstruction.
\end{lemma} 
\begin{remark} \cite[Thm 5.2]{Hai14} suggests the estimate $\mathcal{K}$ maps $\mathcal{D}^\gamma \to \mathcal{D}^{\gamma+\beta}$. The difference to our baby Schauder estimate  stems from the fact, unlike Assumption 5.3 in \cite[p.64]{Hai14} we do not assume that our regularity structure contains the polynomial structure.
\end{remark} 

\begin{proof} (Sketch) The special case $F \equiv \Xi \in \mathcal{D}^\infty $ was already treated in Lemma \ref{lem:kerneleasy}. We only show that, in the general case, $\mathcal{K}$ necessarily has the stated form but will not check the properties. It is enough to consider $F$ with values in $ \langle \Xi, \Xi \mathcal{I} \Xi \rangle $ and make the ansatz
$$
      (\mathcal{K} F )(s):= \mathcal{I} F(s) + (...) \mathbf{1}\,.
$$
Applying reconstruction, together with \cite[Prop. 3.28]{Hai14}  we see that $ \mathcal{R} ( \mathcal{K} F) \equiv (...)$ which in turn must equal  $K * \mathcal{R} F$, provided we postulate validity of (ii). This is the given definition of $\mathcal{K} F$ .
\end{proof}

We return to our goal of solving  
\begin{equation}  \label{eq:VoltD}
\mathcal{Z} = z \mathbf{1} + \mathcal{K} \left(U(\mathcal{Z}) \cdot \Xi + V(%
\mathcal{Z}) \right) \,,
\end{equation}
noting perhaps that $U(\mathcal{Z})$ makes sense for every function-like modelled distribution, say $F(t) = F_0(t) \mathbf{1} +
\sum_{k=1}^M F_k(t) (\mathcal{I} \Xi)^k \in \mathcal{T}_+ := \left\langle \mathbf{1}, \mathcal{I}(\Xi), \ldots, (\mathcal{I}\Xi)^M \right\rangle$, in which case 
\begin{equation} \label{equ:compMD}
U(F)(t) = u(F_0(t)) \mathbf{1} + u^{\prime }(F_0(t)) \sum_{k=1}^M F_k(t) 
\mathcal{I} (\Xi)^k \,.
\end{equation}
(Similar remarks apply to $V$, the composition operator associated to $v \in C^{M+2}$). Recall $M=1$.

\begin{theorem}
\label{thm:Volt} For any admissible model $(\Pi,\Gamma)$ and $u$, $v$ $\in$ $C_b^{M+2}(\mathbb{R}%
)$, for any $T>0$, the equation \eqref{eq:VoltD} has a unique solution in $%
\mathcal{D}^\gamma(\mathcal{T}_+)$, and the map $(u,v,\Pi) \mapsto \mathcal{Z%
}$ is locally Lipschitz in the sense that if $\mathcal{Z}$ and $\tilde{%
\mathcal{Z}}$ are the solutions corresponding respectively to $(u,v,\Pi)$
and $(\tilde{u},\tilde{v},\tilde{\Pi})$, 
\begin{equation*}
{\vert\kern-0.25ex\vert\kern-0.25ex\vert \mathcal{Z};\tilde{\mathcal{Z}}
\vert\kern-0.25ex\vert\kern-0.25ex\vert}_{\mathcal{D}^\gamma_T} \lesssim \|u-%
\tilde{u}\|_{C^{M+2}_b} + \|v-\tilde{v}\|_{C^{M+2}_b} + {\vert\kern%
-0.25ex\vert\kern-0.25ex\vert (\Pi,\Gamma);(\tilde{\Pi},\tilde{\Gamma}) \vert%
\kern-0.25ex\vert\kern-0.25ex\vert}_T\,,
\end{equation*}
with the proportionality constant being bounded when the (resp. $C^{M+2}_b$
and model) norms of the arguments stay bounded.

In addition, if $(\Pi,\Gamma)$ is the canonical It\^o model (associated to Brownian resp. fractional Brownian motion, $H > 1/4$) then $Z=\mathcal{R} 
\mathcal{Z}$ is solves \eqref{equ:grHeston} in the It\^o-sense. 
\end{theorem}

\begin{remark} $Z=\mathcal{R} \mathcal{Z}$ is clearly the (unique) reconstruction of the (unique) solution to the abstract problem. We also checked 
that $Z$ is indeed a solution for the It\^o-Volterra equation. However, if one desires to know that $Z$ is the unique strong solution to the stochastic It\^o-Volterra
equation, it is clear that one has to resort to uniqueness results of the stochastic theory, see e.g. \cite{CD01}.
\end{remark}

\begin{proof}
The well-posedness and continuous dependence on the parameters essentially
follows from results of \cite{Hai14}, details are spelled out the
details in Appendix \ref{app:RH}. 

The fact that the reconstruction of the solution solves the It\^o equation
can be obtained by considering approximations as is done in \cite[Thm 6.2]{HP15}
or \cite[Ch. 5]{FH14}.
\end{proof}

Using the large deviation results obtained in the previous subsection, we
can directly obtain a LDP for the log-price 
\begin{equation*}
X_t = \int_0^t f(Z_s) (\rho dW_s + \bar{\rho} d\bar{W}_s ) - \frac{1}{2}
\int_0^t f^2(Z_s) ds.
\end{equation*}
For square-integrable $h$, let $z^h$ be the unique solution to the integral
equation 
\begin{equation*}
z^h(t) = z + \int_0^t K(s,t) u(z^h(s)) h(s) ds \ .
\end{equation*}

\begin{corollary} Let $H \in (0,1/2]$ and $f$ smooth (without boundedness assumption). Then
$t^{H-\frac{1}{2}} X_t$ satisfies a LDP with speed $t^{2H}$ and rate
function given by 
\begin{equation}  \label{eq:rateZ}
I(x) = \inf_{h \in L^2([0,1])} \{ \frac{1}{2} \|h\|_{L^2}^2 + \frac{%
\left(x-I_1^z(h)\right)^2}{2I_2^z(h)} \}
\end{equation}
where $$I^z_1(h)= \rho \int_0^1 f(z^h(s)) h(s) ds \ , \ \ \ \ I_2^z(h)= \int_0^1
f(z^h(s))^2 ds \ . $$
\end{corollary}

\begin{remark} Despite our previous limitation to $H>1/4$, to approach extends to any $H>0$ and yields the result as stated.
\end{remark}

\begin{proof}
Ignoring the second part $\int_0^t (...) ds$ in $X_t$ which is $O(t) = o(t^{%
\frac{1}{2}-H})$ since $f$ is bounded, we let $\hat{X}_t = \int_0^t f(Z_s)
(\rho dW_s + \bar{\rho} d\bar{W}_s )$ and by scaling we see that 
\begin{equation*}
t^{H-\frac{1}{2}}\hat{X}_t =^d \hat{X}^\delta_1,
\end{equation*}
where $\delta= t^H$ and $X^\delta$, $Z^\delta$ are defined in the same way
as $X$, $Z$ with $W, \bar{W}$ replaced by $\delta W, \delta \bar{W}$ and $v$
replaced by $v^\delta = \delta^{1+\frac{1}{2H}} h$.

We then note that 
\begin{equation*}
X^\delta_1 = \left\langle \mathcal{R}^\delta F(\mathcal{Z}^\delta)(\rho \Xi
+ \bar{\rho} \bar{\Xi}), 1_{[0,1]} \right\rangle =: \Psi(\Pi^\delta,
v^\delta)
\end{equation*}
where $\Psi$ is locally Lipschitz by Theorem \ref{thm:Volt}. We can then
directly use the fact that $\Pi^\delta$ satisfy a LDP (Theorem \ref{thm:LD})
with a contraction principle such as Lemma 3.3 in \cite{HW15} to obtain that 
$X^\delta_1$ satisfies a LDP with rate function 
\begin{equation*}
I(x) = \inf\left\{ \frac{1}{2}(\|h\|_{L^2}^2 + \|\bar{h}\|_{L^2}^2, \;\;\;x
= \Psi(\Pi^{(h,\bar h)}, 0) \right\}.
\end{equation*}
It then suffices to note that $z^h$ is exactly $\mathcal{R} \mathcal{Z}$ for 
$\mathcal{Z}$ the solution to \eqref{eq:VoltD} corresponding to a model $%
\Pi^{(h,\bar h)}$ and with $h \equiv 0$, and to optimize separately over $%
\bar{h}$ as in the proof of Corollary \ref{cor:LD}.
\end{proof}

%

We also have an approximation result :

\begin{corollary} Let $H > 1/4$ (for simplicity, but see remark below). Then
$Z = \lim Z^\varepsilon$, uniformly on compacts and in probability, where
\begin{equation}  \label{eq:Volteps}
Z^\varepsilon_t = z + \int_0^t K(s,t) \left( u(Z^\varepsilon_s)
dW^\varepsilon_s + (v(Z^\varepsilon_s) - \mathscr{C}^\varepsilon(s)
uu^{\prime}(Z^\varepsilon_s) ds \right) \ .
\end{equation}
\end{corollary}

\begin{remark} Replacing the renormalization function $ \mathscr{C}^\varepsilon$by its mean is possible, provided $H>1/4$. However, unlike the discussion at the end of Section \ref{subsec:ApproximationTheory}, this is no more a
consequence of quantifying the distributional convergence. In the present context, this is achieved by checking directly model-convergence, which, fortunately, is not much harder. We leave details to the 
interested reader. \end{remark}

\begin{remark} In contrast to the previous statement, the above result is more involved for $H \in (0,1/4]$ and additional terms renormalization terms appear, the general description of which would benefit from pre-Lie products, as recently introduced \cite{BCFP17x}.
 
\end{remark}

\begin{proof}
Thanks to Theorem \ref{prop:ModelConvergence} and Theorem \ref{thm:Volt} it follows from continuity of reconstruction 
that 
\begin{equation*}
Z = \mathcal{R} \mathcal{Z} = 
\lim_{\varepsilon \to 0} \mathcal{R^\eps} 
\mathcal{Z}^\eps , 
\end{equation*} 
so that the only thing to do is check that $Z^\varepsilon$
solves \eqref{eq:Volteps}. Note that %
\eqref{eq:VoltD} implies that one has (omitting upper $\eps$'s at all normal and caligraphic $Z$ ...)
\begin{equation*}
\mathcal{Z} (t) = Z_t \mathbf{1} + u(Z_t) \mathcal{I} (\Xi),
\end{equation*}
and, with (\ref{equ:compMD}),
\begin{equation*}
U(\mathcal{Z}(t)) \Xi = u(Z_t) \Xi + u^{\prime }u (Z_t) \mathcal{I}(\Xi) \Xi.
\end{equation*}
But then since $\hat{\Pi}^\varepsilon$ is a ``smooth'' model, in the sense of Remark 3.15. in \cite{Hai14}, one has 
\begin{eqnarray*}
\mathcal{R}^\varepsilon( U(\mathcal{Z}^\varepsilon) \Xi )(t) 
&=&
\hat\Pi^\varepsilon_t( U(\mathcal{Z}^\varepsilon (t)) \Xi ) (t)\\
&=& u(Z^\varepsilon_t) (\hat\Pi^\varepsilon_t \Xi)(t) 
+ u^{\prime}u (Z^\varepsilon_t) (\hat\Pi^\varepsilon_t \Xi \mathcal{I}(\Xi))(t) \\
&=& 
u(Z^\varepsilon_t) \dot{W}^\varepsilon(t) 
- u^{\prime}u (Z^\varepsilon_t) \mathscr{K}^\varepsilon(t,t) \ .
\end{eqnarray*}
Since convolution commutes with reconstruction, cf. Lemma \ref{lem:KisK}, it follows that $Z^\varepsilon$ is indeed a solution to %
\eqref{eq:Volteps}.
\end{proof}

\section{Numerical results}  
\label{sec:ChristianBen}

We will now resume where we left off in Section \ref{subsec:haar} and revisit the case of European option pricing under rough volatility.  Building on the theoretical underpinnings of Section \ref{sec:RPRS}, we present a concise description of the central algorithm of this paper - for simplicity restricted to the unit time interval - and complement the theoretical convergence rates obtained in previous chapters with numerical counterparts. The code used to run the simulations has been made available on \url{https://www.github.com/RoughStochVol}. 

\bigskip

\noindent {\bf Concise description. }  Without loss of generality, set time to maturity $T=1$. We are interested in
pricing a European call option with spot $S_0$ and strike $K$ under rough
volatility. From Theorem \ref{thm:rPricing}, we have
\begin{equation}
\label{eq:call_rough_def}
C\left( S_{0},K,1\right) =\mathbb{E}\left[ C_{BS}\left(S_{0}\exp \left( \rho
\mathscr{I}-\frac{\rho ^{2}}{2} \mathscr{V} \right)
,K,\frac{\bar{\rho}^{2}}{2}\mathscr{V} \right) \right]
\end{equation}
where the computational challenge obviously lies in the efficient simulation of

\begin{equation*}
\left( \mathscr{I},\mathscr{V}\right) = \left( \int_{0}^{1}f(\hat{W}_t,t
) \dd W_t,\int_{0}^{1}f^{2}(\hat{W}_t, t) \dd t\right) .
\end{equation*}%
As explored in Subsection \ref{subsec:haar}, we take a Wong-Zakai-style approach to simulating $\mathscr{I}$, that is, we approximate the White noise process $\dot{W}$ on the Haar grid as follows: \\
Let $\{Z_i\}_{i=1, \ldots 2^N-1} \sim iid \;\mathcal{N}(0,1)$ and choose a Haar grid level $N\in \mathbb{N}$ such that the stepsize of the Haargrid $\varepsilon = 2^{-N}$. Then, for all $t \in [0,1]$ and $i=0, \ldots, 2^N-1$, we set
\begin{align}
\dot{W}^{\varepsilon}(t) = \sum_{i=0}^{2^N-1} Z_i e_i^{\varepsilon}(t) \quad \text{where} \quad 
e_i^{\varepsilon}(t) = 2^{N/2} \mathbf{1}_{[i2^{-N},(i+1)2^{-N})}(t)
\end{align}
which induces an approximation of the fBm
\begin{align}
\label{eq:num_w_hat}
\hat{W}^{\varepsilon}(t) &= \sum_{i=0}^{2^N-1} Z_i \hat{e}_i^{\varepsilon}(t) \quad \text{where} \\
\hat{e}_i^{\varepsilon}(t) &= \mathbf{1}_{t>i2^{-N}} \frac{\sqrt{2H}2^{N/2}}{H+1/2} \Big(| t-i2^{-N}|^{H+1/2} - | t - \min((i+1)2^{-N},t)|^{H+1/2}\Big). \label{eq:e_hat_simu}
\end{align}
As outlined before, the central issue is that the object $\int_0^1 f(\hat{W}^{\varepsilon}(t),t) \dot{W}^{\varepsilon}(t) \dd t$ does \emph{not} converge in an appropriate sense to the object of interest $ \mathscr{I}$ as $\varepsilon \to 0$. This is overcome by \emph{renormalizing} the object, two possible approaches of which are explored in Subsection \ref{subsec:haar}. For the remainder, we will consider the 'simpler' renormalized object given by
\begin{align}
\label{eq:renormalized_I_standard}
\tilde{\mathscr{I}}^\eps = \int_0^1 f(\hat{W}^{\varepsilon}(t),t) \dot{W}^{\varepsilon}(t) \dd t - \int_{0}^{1} \mathscr{C}^{\varepsilon}(t) \pa_1 f(\hat{W}^{\varepsilon}(t),t) \dd t 
\end{align}
where the renormalization object $\mathscr{C}^\varepsilon(t)$ can be one of 
\begin{align}
\label{eq:num_K}
\mathscr{C}^{\varepsilon}(t) = \begin{cases}
\frac{2^N \sqrt{2H}}{H+1/2}|t-\floor{t2^N}2^{-N}|^{H+1/2}\\
\frac{\sqrt{2H}}{(H+1/2)(H+3/2)} 2^{N(1/2-H)}.
\end{cases}
\end{align}
Coming back to the original question of simulating $\left(\mathscr{I}, \mathscr{V}\right)$, we just argued that what we really need to simulate to achieve convergence in a suitable sense is the object $\left(\tilde{\mathscr{I}}^\eps, \mathscr{V}^{\varepsilon}\right)$, the expressions of which are collected below (note that under an assumed non-constant renormalization the expression \eqref{eq:renormalized_I_standard} for $\tilde{\mathscr{I}}^\eps$  has been rewritten to a form more suitable for efficient simulation):

\begin{align}
\tilde{\mathscr{I}}^\eps &=\sum_{i=0}^{2^N-1} \int_{i2^{-N}}^{(i+1)2^{-N}}\Big[Z_i 2^{N/2}f(\hat{W}^{\varepsilon}(t),t) -\frac{\sqrt{2H}2^{N}}{H+1/2}| t-i2^{-N}|^{H+1/2} \pa_1 f(\hat{W}^{\varepsilon}(t),t) \Big]\dd t \label{eq:I_approx_simu}\\
\mathscr{V}^{\varepsilon} &=\sum_{i=0}^{2^N-1} \int_{i2^{-N}}^{(i+1)2^{-N}}f^{2}(\hat{W}^{\varepsilon}(t),t)\dd t. \label{eq:V_approx_simu}
\end{align}

\begin{algorithm}[h]
	\label{algo:algorithm}
	\SetAlgoLined
	\SetKwInOut{Parameters}{Parameters}
	\Parameters           {$M \in \mathbb{N}$: \# Monte Carlo simulations  \\
		$N \in \mathbb{N}$: Haar grid 'level' such that $\varepsilon = 2^{-N}$ \\
		$d \in \mathbb{N}$: \# discretisation points of trapezoidal rule in each Haar subinterval}
	\KwOut{$M$ samples of bivariate object $(\tilde{\mathscr{I}}^{\varepsilon}, \mathscr{V}^{\varepsilon})$}
	\BlankLine
	initialize $\tilde{\mathscr{I}}^{\varepsilon} =  \mathscr{V}^{\varepsilon} =  \mathbf{0} \in \mathbb{R}^M$\;
	simulate array $\mathbf{Z} \in \mathbb{R}^{M \times 2^N}$ of \emph{iid} standard normals\;
	\For{{\upshape each	Haar subinterval }$[i2^{-N}, (i+1)2^{-N})$ {\upshape where} $i \in \{0, \ldots, 2^N -1\}$ }{
		choose discretization grid $\mathcal{D}^i$ with $d$ points on the Haar subinterval\;
		evaluate functions $\hat{e}_k^{\epsilon}, k=0, \ldots, i,$ from \eqref{eq:e_hat_simu} on $\mathcal{D}^i$ to obtain $\mathbf{\hat{e}}^{\epsilon} \in \mathbb{R}^{(i+1) \times d}$\;
		compute $\hat{\mathbf{W}}^{\epsilon} = \mathbf{Z}^* \times \mathbf{\hat{e}}^{\epsilon} \in \mathbb{R}^{M \times d} $ where $\mathbf{Z}^* \in \mathbb{R}^{M \times (i+1)}$ is the truncation of $\mathbf{Z}$ to its first $i+1$ columns such that $\hat{\mathbf{W}}^{\epsilon} $ is an approximation of the fBM on $\mathcal{D}^i$\;
		evaluate integrands from equations (\ref{eq:I_approx_simu}, \ref{eq:V_approx_simu}) on $\mathcal{D}^i$ using $\hat{\mathbf{W}}^{\epsilon} $ and the last column of $\mathbf{Z}^*$\;
		approximate respective integrals on subinterval by trapezoidal rule \;
		add obtained estimates to running sums $\tilde{\mathscr{I}}^{\varepsilon}$ and $\mathscr{V}^{\varepsilon}$\;
	}
	\Return{$\tilde{\mathscr{I}}^{\varepsilon}, \mathscr{V}^{\varepsilon}$}
	\caption{Simulation of $M$ samples of $(\tilde{\mathscr{I}}^{\varepsilon}, \mathscr{V}^{\varepsilon})$}
\end{algorithm}

{\bf Numerical convergence rates.} 

In this subsection, we will discuss strong convergence of the approximative object $\tilde{\mathscr{I}}^{\varepsilon}$ to the actual object of interest $\mathscr{I}$ as well as weak convergence of the option price itself as the Haar grid interval size $\varepsilon \to 0$. Specifically, we will be looking at Monte Carlo estimates of our errors, that is, in order to approximate some quantity $\mathbb{E}[X]$ for some random variable $X$, we will instead be looking at $\frac{1}{M}\sum_{i=1}^M X_i$ where the $X_i$ are $M$ \emph{iid} samples drawn from the same distribution as $X$. In other words, we need to generate $M$ realisations of the bivariate stochastic object $\left(\tilde{\mathscr{I}}^\eps, \mathscr{V}^{\varepsilon}\right)$, a task that can be vectorized as described below, thus avoiding expensive looping through realisations.

\begin{figure}
	\centering
	\captionsetup{width=\linewidth}
	\includegraphics[width=1\linewidth]{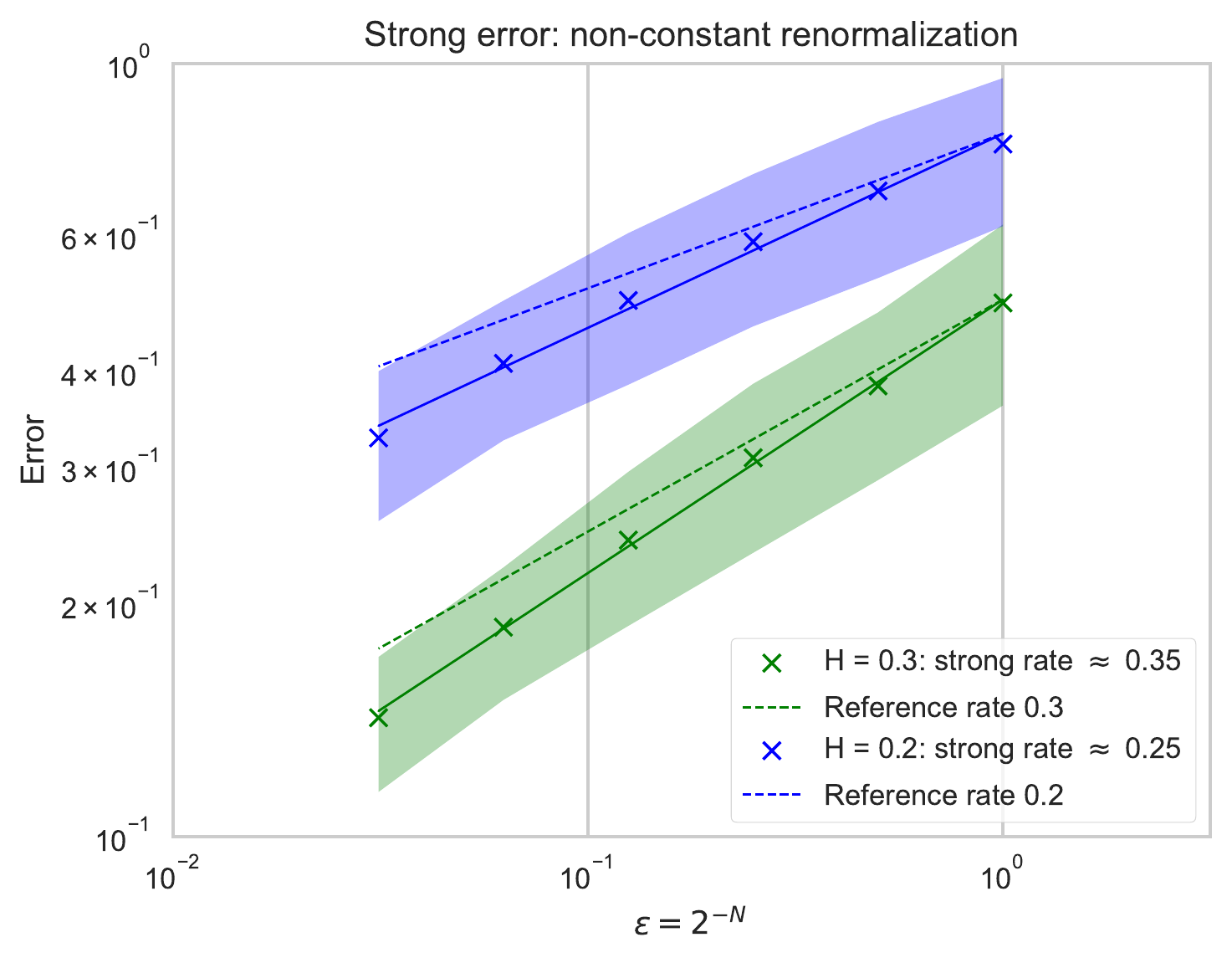}
	\caption{Empirical strong \eqref{eq:strong_objective} errors on a log-log-scale under a
          non-constant renormalization, obtained through $M=10^5$ Monte Carlo
          samples with a trapezoidal rule delta of $\Delta =
          2^{-17}$ and fineness of the reference Haar grid $\varepsilon' = 2^{-8}$. Solid lines visualize the empirical rates of convergence
          obtained by least squares regression, dashed lines provide visual reference rates. Shaded colour bands show interpolated 95\% confidence levels based on normality of Monte Carlo estimator.}
	\label{fig:non_const_rates}
\end{figure}

\begin{figure}
	\centering
	\captionsetup{width=\linewidth}
	\includegraphics[width=1\linewidth]{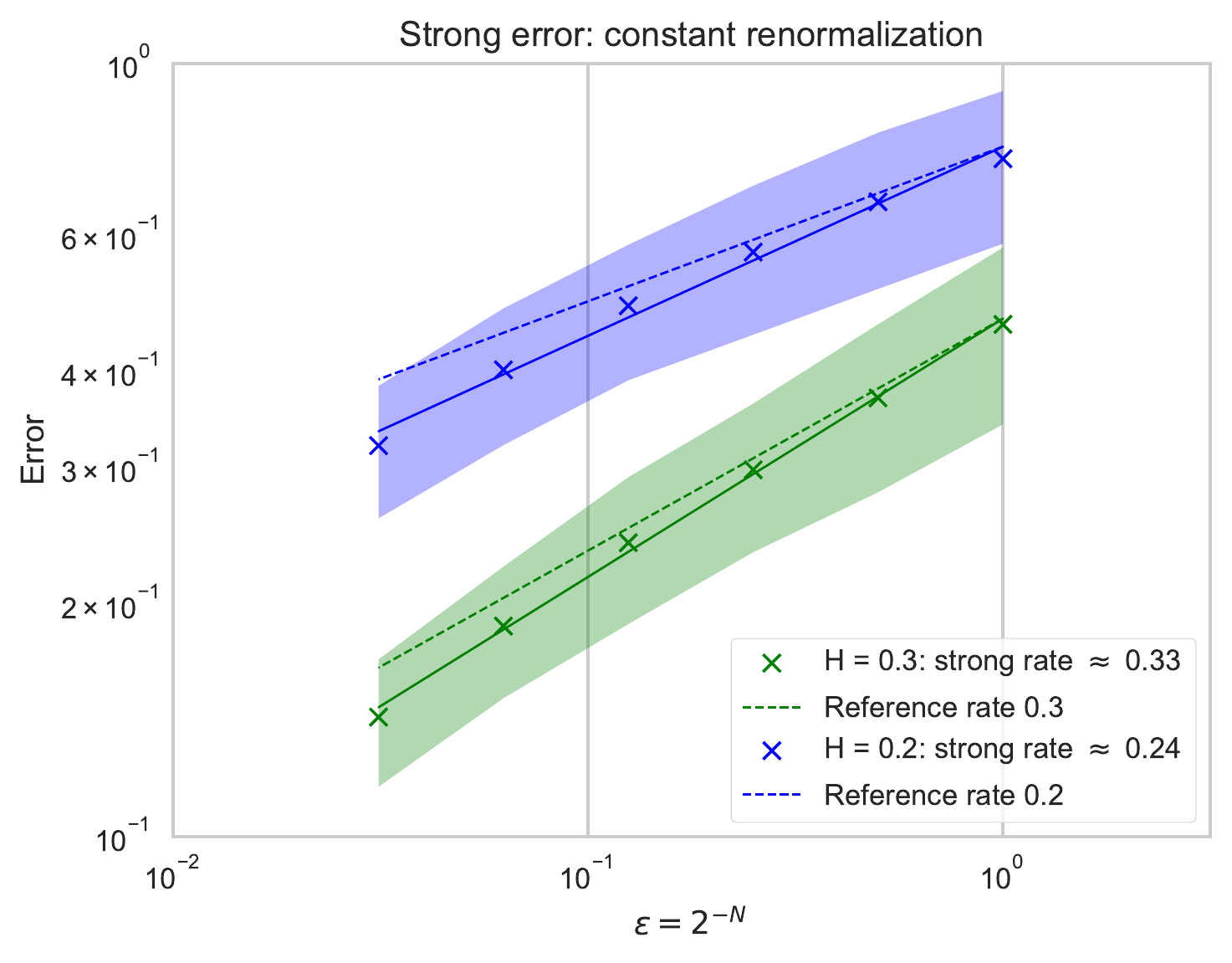}
	\caption{Empirical strong \eqref{eq:strong_objective} errors on a log-log-scale under a constant renormalization, obtained through $M=10^5$ Monte Carlo samples with a trapezoidal rule delta of $\Delta = 2^{-17}$ and fineness of the reference Haar grid $\varepsilon'=2^{-8}$. Solid lines visualize the empirical rates of convergence
          obtained by least squares regression, dashed lines provide visual reference rates. Shaded colour bands show interpolated 95\% confidence levels based on normality of Monte Carlo estimator.}
	\label{fig:const_rates}
\end{figure}


{\it Strong convergence}.
We verify Theorem \ref{thm:ConvergenceDiscreteIntegrals} (i) numerically , albeit in the $L^{2}(\Omega)$-sense and - for simplicity - with $f\left(x,t\right) =\exp(x)$, i.e. with no explicit time dependence. That is, we are concerned with Monte Carlo approximations of%
\begin{equation*}
\norm{  \tilde{ \mathscr{I}}^\eps 
-\int_{0}^{1}\exp(\hat{W}_t)
	\dd W_{t}}_{L^2{(\Omega})}
\end{equation*}%
and we expect an error almost of order $\varepsilon ^{H}$.

\begin{remark}
	
	We choose $f\left(x,t\right) =\exp(x)$ because this closely resembles the \emph{rough Bergomi} model (see \cite{BFG16} and below). Also, for the simplest non-trivial choice, $f\left(x, t\right) = x$, the discretization error is overshadowed by the Monte Carlo error, even for very coarse grids.
\end{remark}

Since $\left( W,\hat{W}\right) $ is a two-dimensional Gaussian process with known covariance structure, it is possible to use the Cholesky algorithm (cf. \cite{BFG16, BFGHS17x}) to simulate the joint paths on some grid and then use standard Riemann sums to approximate the integral. The value obtained in this way could serve as a reference value for our scheme. However - for strong convergence -  we need both objects to be based on the same stochastic sample. For this reason, we find it easier to construct a reference value by the wavelet-based scheme itself, i.e. we simply pick some $\varepsilon ^{\prime }\ll\varepsilon $ and consider
\begin{equation}
\label{eq:strong_objective}
\norm{\tilde{\mathscr{I}}^\eps-\tilde{\mathscr{I}}^{\eps'}}_{L^2(\Omega)}
\end{equation}%
as $\eps \rightarrow \eps^\prime$. As can be seen in Figures \ref{fig:non_const_rates} and \ref{fig:const_rates},  both renormalization approaches stated in \eqref{eq:num_K} are consistent with a theoretical strong rate of almost $H$ across the full range of $0< H < 0.5$ (cf. discussion at the end of Section \ref{subsec:ApproximationTheory}).

\begin{remark}[Weak convergence]
In absence of a Markovian structure, a proper weak convergence analysis proves to be subtle, that is, an analysis that - for suitable test functions $\varphi$ - yields a rate of convergence for
\begin{equation*}
\left\vert \mathbb{E}\left[\varphi \left( \tilde{\mathscr{I}}^\eps\right)\right]  - \mathbb{E}\left[\varphi \left( \int_{0}^{1}\exp(\hat{W}_t)\dd W_{t}\right) \right] \right\vert 
\end{equation*}%
as $\epsilon \rightarrow 0$, remains an open problem. However, picking $\varphi(x) = x^2$, Ito's isometry yields
\begin{equation}
\mathbb{E}\left[ \left(\int_{0}^{1}\exp(\hat{W}_t)\dd W_{t}\right)^2 \right] = \int_0^1 \mathbb{E}\left[\exp\left(2\hat{W}_t\right)\right] \dd t = \int_0^1 \exp\left(2t^{2H}\right) \dd t
\end{equation}
which we can be approximated numerically. So we can consider
\begin{equation}
\label{eq:weak_objective}
\left\vert \mathbb{E}\left[\left(\tilde{\mathscr{I}}^\eps\right)^2\right]  - \int_0^1 \exp\left(2t^{2H}\right) \dd t \right\vert 
\end{equation}%
as $\varepsilon \to 0$. Our preliminary results indicate that for both renormalization approaches the weak rate seems to be around the strong rate $H$.
\end{remark}


{\it Option pricing}. We pick a simplified version of the \emph{rough Bergomi} model \cite{BFG16} where the instantaneous variance is given by
\begin{equation*}
f^{2}\left( x\right) =\sigma _{0}^{2}\exp \left( \eta x\right)
\end{equation*}%
with $\sigma _{0}$ and $\eta$ denoting spot volatility and volatility of volatility respectively. Let $C^{\varepsilon}$ denote the approximation of the call price  \eqref{eq:call_rough_def} based on $\left(\tilde{\mathscr{I}}^\eps, \mathscr{V}^{\varepsilon}\right)$, fix some $\varepsilon ^{\prime }\ll\varepsilon$ and consider
\begin{equation}
\label{eq:weak_option_error}
\left\lvert{C^{\varepsilon}\left( S_{0},K,T=1\right) - C^{\varepsilon^\prime}\left( S_{0},K,T=1\right)}\right\rvert
\end{equation}
as $\varepsilon \to \varepsilon^\prime$. Empirical results displayed in Figure \ref{fig:option_rates} indicate a weak rate of $2H$ across the full range of $0<H<\frac{1}{2}$. 

\begin{figure}
	\centering
	\captionsetup{width=\linewidth}
	\includegraphics[width=1\linewidth]{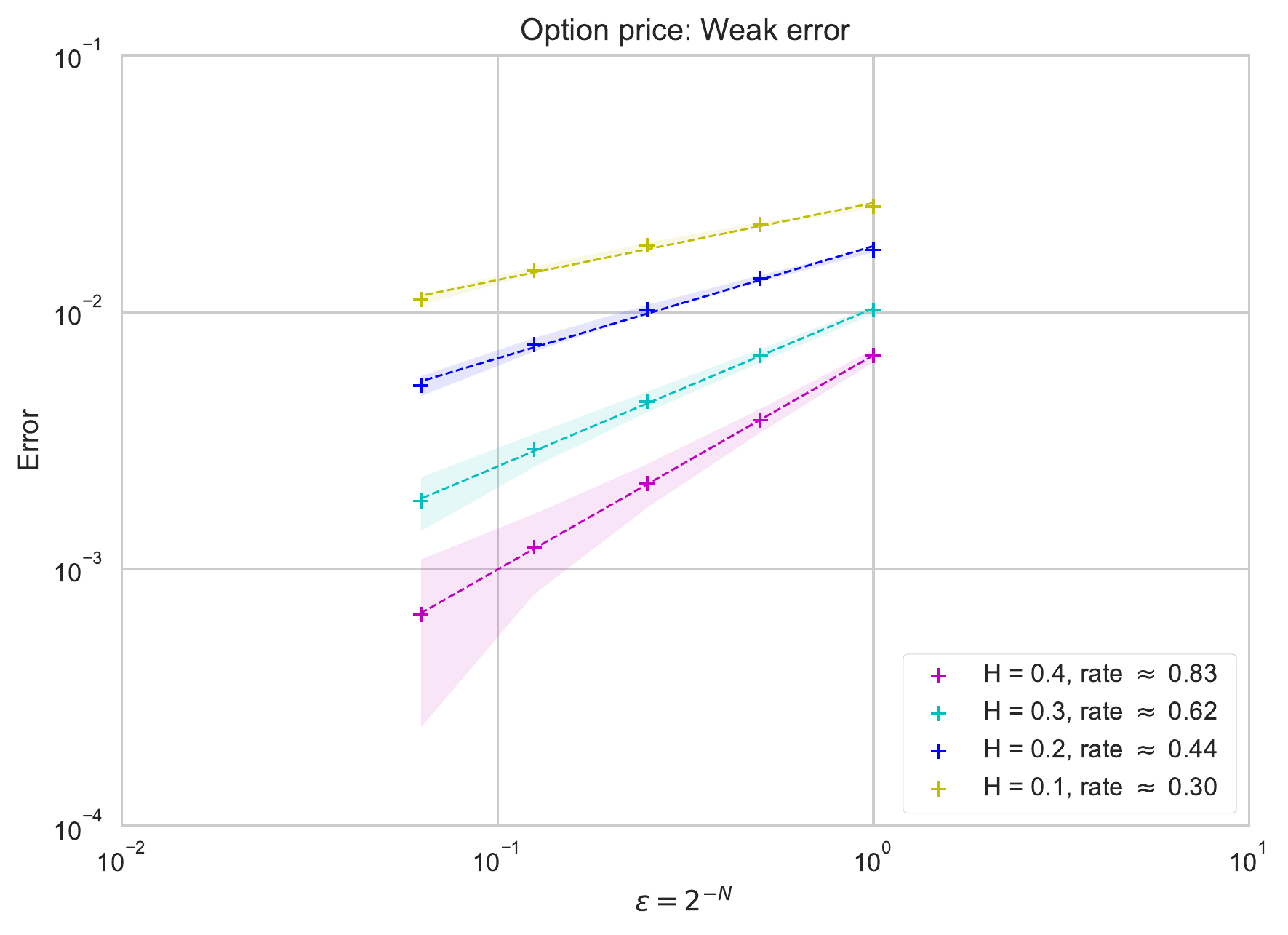}
	\caption{Empirical weak \eqref{eq:weak_option_error} errors on a log-log-scale as $\varepsilon \to \varepsilon^{\prime}=2^{-8}$, obtained through $M=10^5$ MC samples with spot $S_0=1$, strike $K=1$, correlation $\rho=-0.8$, spot vol $\sigma_0 = 0.2$, vvol $\eta=2$ and trapezoidal rule delta $\Delta = 2^{-17}$. Dashed lines represent LS estimates for rate estimation, shaded colour bands show confidence levels based on normality of Monte Carlo estimator.}
	\label{fig:option_rates}
\end{figure}

\appendix

\section{Approximation and renormalization (Proofs)}  
\label{app:AR}

\begin{lemma}
\label{lem:FractionalDifferentiation} For $a,b>0$ and $\delta\in [0,1]$ we
have for $x\notin[0,1)$ 
\begin{align*}
|a^x-b^x|\leq 2^{1-\delta} |x|^\delta (a^{x-\delta}\vee b^{x-\delta}) \cdot
|a-b|^{\delta}
\end{align*}
and for $x\in (0,1)$ 
\begin{align*}
|a^x-b^x|\leq 2^{1-\delta} |x|^\delta (a^{(x-1)\delta}b^{x(1-\delta)} \vee
b^{(x-1)\delta}a^{x(1-\delta)}) \cdot |a-b|^{\delta}\,.
\end{align*}
\end{lemma}

\begin{proof}
This follows from interpolation between $|a^x-b^x|\leq |x| \sup_{z\in [a,b]}
z^{x-1}|a-b|\leq |x| a^{x-1}\vee b^{x-1} |a-b|$ and $|a^x-b^x|\leq
a^x+b^x\leq 2 a^x\vee b^x$.
\end{proof}

\begin{proof}[Proof of Lemma \protect\ref{lem:ConvergenceW} ]
Rewriting $\hat{W}^{\varepsilon}(t) =\sqrt{2H}\int_{0}^{\infty}\mathrm{d} 
W(u)\,\int_{0}^{\infty}\mathrm{d} 
r\,\delta^{\varepsilon}(r,u)\,|t-r|^{H-1/2}\mathbf{1}_{r<t}$ we have 
\begin{align*}
\mathbb{E} \left|\hat{W}^\varepsilon(t)-\hat{W}^\varepsilon(s) \right|^2
&=2H\int_0^\infty \mathrm{d}  u \left(\int_0^\infty \mathrm{d}  r
\delta^\varepsilon(r,u) (\mathbf{1}_{r<t} |t-r|^{H-1/2}-\mathbf{1}_{r<s}
|s-r|^{H-1/2})\right)^2 \\
&\lesssim\int_0^\infty \mathrm{d}  u \int_0^\infty \mathrm{d}  r
|\delta^\varepsilon(r,u)| \left(\mathbf{1}_{r<t} |t-r|^{H-1/2}-\mathbf{1}%
_{r<s} |s-r|^{H-1/2}\right)^2 \\
&\lesssim \int_0^{s\vee t}\mathrm{d}  r \left(\mathbf{1}_{r<t} |t-r|^{H-1/2}-%
\mathbf{1}_{r<s} |s-r|^{H-1/2}\right)^2 \,,
\end{align*}
where we used the It\^o isometry in the first and Jensen's inequality in the
second step. Assuming $s<t$ we can split the integral in domains $[0,s]$ and 
$[s,t]$ which yields the bound $|t-s|^{2H} \int_0^s |s-r|^{4H-1} +|t-s|^{2H}
\lesssim |t-s|^{2H}$. Application of equivalence of moments for Gaussian
random variables and Kolmogorov's criterion then shows the first inequality.

The second one follows by interpolation (and once more Kolmogorov) if we can
prove that 
\begin{align}  \label{eq:proof:ConvergenceW1}
\mathbb{E}|\hat{W}^\varepsilon(t)-\hat{W}(t)|^2\lesssim
\varepsilon^{2H-\kappa^{\prime }}.
\end{align}
We have, by It\^o's isometry, 
\begin{equation*}
\mathbb{E} [\left|\hat{W}^{\varepsilon}(t)-\hat{W}(t)\right|^{2}]=2H%
\int_{0}^{\infty}\mathrm{d}  u\,\left(\int_{0}^{\infty}\mathrm{d} 
r\,\delta^{\varepsilon}(r,u)\,|t-r|^{H-1/2}\mathbf{1}_{r<t}-|t-u|^{H-1/2}%
\mathbf{1}_{u<t}\right)^{2} .
\end{equation*}
We can enlarge the inner integral such that $\int\delta^{\varepsilon}(r,u)=1$
by negleting an error term which can be estimated by $\int_{B(0,c%
\varepsilon)}\mathrm{d}  u(\int_{B(0,c\varepsilon)}\mathrm{d} 
r\,\varepsilon^{-1}|t-r|^{H-1/2})^{2}\lesssim\varepsilon^{2H}$. Application
of Jensen's inequality then yields 
\begin{align*}
\int_{0}^{\infty}\mathrm{d}  u\,\int_{-\infty}^{\infty}\mathrm{d} 
r\,|\delta^{\varepsilon}(r,u)|\,\left(|t-r|^{H-1/2}\mathbf{1}%
_{r<t}-|t-u|^{H-1/2}\mathbf{1}_{u<t}\right)^{2} .
\end{align*}
The cases where either $r>u$ or $u>t$ yield an $\varepsilon^{2H}$ error term
as above so that bounding with Lemma \ref{lem:FractionalDifferentiation} 
\begin{equation*}
\left||t-r|^{H-1/2}-|t-u|^{H-1/2}\right|\lesssim(|t-r|^{-1/2+%
\kappa}+|t-u|^{-1/2+\kappa})\cdot|u-r|^{H-\kappa} 
\end{equation*}
proves \eqref{eq:proof:ConvergenceW1}.
\end{proof}

\begin{proof}[Proof of \eqref{eq:ConvergenceModel}]
We only consider the symbols $\Xi\mathcal{I}^{m}(\Xi)$, the symbols $%
\mathcal{I}(\Xi)^m$ can be handled with Lemma \ref{lem:ConvergenceW}. In
view of Lemma \ref{lem:ReshapingLimit} and \ref{lem:ReshapingApproximation}
we have to controll (for $m\geq 0$ in the first equation and $m>0$ in the
second equation) 
\begin{align}  \label{eq:proof:ConvergenceModel1}
&\mathbb{E}\left|\int_0^\infty \mathrm{d}  W^\varepsilon(t)\diamond
\varphi^\lambda_s(t) (\hat{W}^\varepsilon_{st})^m-\int_0^\infty \mathrm{d} 
W(t)\diamond \varphi^\lambda_s(t) (\hat{W}_{st})^m \right|^2\lesssim
\varepsilon^{2\delta\kappa^{\prime }} \lambda^{2mH-1-2\kappa^{\prime }}\,, \\
&\mathbb{E}\left|\int_0^\infty \mathrm{d}  t \,\varphi^\lambda_s(t)\left(%
\mathscr{K}^\varepsilon(s,t)(\hat{W}^\varepsilon_{st})^{m-1}-K(s-t)(\hat{W}%
_{st})^{m-1}\right)\right|^2\lesssim \varepsilon^{2\delta\kappa^{\prime }}
\lambda^{2mH-1-2\kappa^{\prime }}\, ,  \label{eq:proof:ConvergenceModel2}
\end{align}
where $\hat{W}^{(\varepsilon)}_{st}=\hat{W}^{(\varepsilon)}(t)-\hat{W}%
^{(\varepsilon)}(s)$ and where $\delta\in (0,1),\,\kappa^{\prime }\in (0,H)$
is arbitrary. Equivalence of norms in the Wiener chaos and a version of
Kolmogorov's criterion for models (\cite[Proposition 3.32]{Hai14}) then gives
\eqref{eq:ConvergenceModel} (note that this gives for a better homogeneity
then we actually need since we only subtract $2\kappa^{\prime }$ and not
$2m\kappa^{\prime }$ in the exponent of $\lambda\in (0,1]$). 
We can rewrite the random variable of \eqref{eq:proof:ConvergenceModel1} as 
\begin{align*}
\int_0^{T+1} \mathrm{d}  W(t)\diamond \int \mathrm{d}  u\,
\delta^\varepsilon(t,u) \left(\mathbf{1}_{u\geq 0}\varphi^\lambda_s(u) (\hat{%
W}^\varepsilon_{su})^m-\varphi^\lambda_s(t) (\hat{W}_{st})^m\right)
\end{align*}

Using \cite[Theorem 7.39]{Jan97}  and Jensen's inequality we can estimate
the second moment of this Skorohod integral by 
\begin{align*}
\mathbb{E}|\eqref{eq:proof:ConvergenceModel1}|^2\lesssim \int_0^{T+1}\mathrm{%
d}  t\,\int \mathrm{d}  u\, |\delta^\varepsilon(t,u)| \,\mathbb{E}\left(%
\mathbf{1}_{u\geq 0}\varphi^\lambda_s(u) (\hat{W}^\varepsilon_{su})^m-%
\varphi^\lambda_s(t) (\hat{W}_{st})^m\right)^2 \ .
\end{align*}
In the regime $\lambda\leq \varepsilon$ every term in the squared
parentheses can simply be bounded (using Lemma \ref{lem:ConvergenceW}) by $%
\lambda^{2H-1}\lesssim \lambda^{2H-1-2\kappa^{\prime }}
\varepsilon^{\kappa^{\prime }} $. If on the other hand $\varepsilon<\lambda$
we can split off a term of order $\int_{B(0,c\varepsilon)} \mathrm{d} 
t\int_{B(0,c\varepsilon)}\frac{\mathrm{d}  u}{\varepsilon}\lesssim
\lambda^{2mH-1-2\kappa^{\prime }}\varepsilon^{2\kappa^{\prime }}$ to drop
the indicator $\mathbf{1}_{u\geq 0}$ and can bound on the support of $%
\delta^\varepsilon(t,u)$ 
\begin{align*}
|\varphi^\lambda_s(u) (\hat{W}^\varepsilon_{su})^m-\varphi^\lambda_s(t) (%
\hat{W}_{ts})^m| &\leq |(\varphi^\lambda_s(u)-\varphi^\lambda_s(t))\cdot|%
\hat{W}^\varepsilon_{su}|^m+ |\varphi^\lambda_s(t)| \cdot \left|(\hat{W}%
^\varepsilon_{su})^m-(\hat{W}_{st})^m\right| \\
&\lesssim C_\varepsilon \mathbf{1}_{B(s,(1+2c)\lambda)}(t)\,\lambda^{-1-\kappa^{\prime }}
\varepsilon^{\kappa^{\prime }} \lambda^{mH}+C_\varepsilon \mathbf{1}_{B(s,\lambda)}(t)
\lambda^{-1} \lambda^{mH-\kappa^{\prime }} \varepsilon^{\kappa^{\prime }}\,,
\end{align*}
where $C_\varepsilon>0$ denote random constants that are uniformly bounded in $L^p$ for $p\in [1,\infty)$. This shows \eqref{eq:proof:ConvergenceModel1}. To estimate %
\eqref{eq:proof:ConvergenceModel2} we first note that due to $\mathbb{E}|(%
\hat{W})^{m-1}_{st}-(\hat{W}_{st}^\varepsilon)^{m-1}|^ 2\lesssim
|t-s|^{2(m-1)H-2\kappa^{\prime }} \varepsilon^{\delta 2\kappa^{\prime }}$ we
are only left with 
\begin{align*}
\mathbb{E}\left|\int_0^\infty \mathrm{d}  t\, \varphi^\lambda_s(t)\left(%
\mathscr{K}^\varepsilon(s,t)-K(s-t)\right) (\hat{W}^\varepsilon_{st})^{m-1}%
\right|^2 
\lesssim \int_0^\infty\mathrm{d}  t \,\varphi^\lambda_s(t) \left|\mathscr{K}%
^\varepsilon(s,t)-K(s-t)\right|^2 \, |s-t|^{2(m-1)H}\,,
\end{align*}
which is straightforward to bound with Lemma \ref%
{lem:ApproximateVolterraEstimate} if $\lambda\leq \varepsilon$. For $%
\lambda<\varepsilon$ and $t>2c\varepsilon$ with $c>0$ as in Definition \ref%
{def:Dirac} the desired bound follows from Lemma \ref%
{lem:ModelConvergenceHelpLemma}. The remaining case however contributes with 
\begin{align*}
&\int_{B(0,2c \varepsilon)} \mathrm{d}  t\,\varphi^\lambda_s(t)
|t-s|^{2(m-1)H} (\varepsilon^{2H-1}+|t-s|^{2H-1}) \\
& \lesssim \int_{B(s,\lambda^{-1} 2c\varepsilon)} \mathrm{d}  t\,
(\lambda^{2(m-1)H} \varepsilon^{2H-1}+\lambda^{2mH-1} |t|^{2mH-1}) \\
&\lesssim \lambda^{2(m-1)H-1}\varepsilon^{2H}+\lambda^{2mH-1}
(\lambda^{-1}\varepsilon)^{2mH}\leq \lambda^{2kH-\kappa^{\prime }}
\varepsilon^{\kappa^{\prime }}\,,
\end{align*}
which completes the proof.
\end{proof}

\begin{lemma}
\label{lem:ModelConvergenceHelpLemma} For $c$ as in Definition \ref%
{def:Dirac} and $t>2c\varepsilon$ and $s\in\mathbb{R} $ we have for $%
\kappa^{\prime }\in (0,H)$ 
\begin{align*}
|K(s-t)-\mathscr{K}^\varepsilon(s,t)|\lesssim |s-t|^{H-1/2-\kappa^{\prime }}
\varepsilon^{\kappa^{\prime }}\,.
\end{align*}
\end{lemma}

\begin{proof}
If $2c\varepsilon\geq |s-t|/2$ the bound easily follows from Lemma \ref%
{lem:ApproximateVolterraEstimate}. If $2c\varepsilon\geq |s-t|/2$ we can
reshape 
\begin{align*}
|K(s-t)-\mathscr{K}^\varepsilon(s,t)|=\left|\int_{-\infty}^\infty \dd u\,
\delta^{2,\varepsilon}(t,u) (\mathbf{1}_{t<s} |s-t|^{H-1/2}-\mathbf{1}%
_{s<u}|s-u|^{H-1/2}) \right|\,,
\end{align*}
where $\delta^{2,\varepsilon}(t,\cdot):=\int_{-\infty}^\infty \mathrm{d} 
x_1 \int_{-\infty}^\infty \mathrm{d}  x_2 \delta^\varepsilon(t,x_1)
\delta^\varepsilon(x_1,\cdot)$ satisfies the properties in Definition \ref%
{def:Dirac} with support in $B(t,2c\varepsilon)$. Note that for $%
2c\varepsilon\geq |s-t|/2$ either both indicator functions vanish or none so
that we only have to consider $t<s$ where we obtain with Lemma \ref%
{lem:FractionalDifferentiation} up to a constant $\int_{-\infty}^\infty
|\delta^{2,\varepsilon}(t,u)| |t-s|^{H-1/2-\kappa^{\prime }}
\varepsilon^{\kappa^{\prime }}\lesssim|t-s|^{H-1/2-\kappa^{\prime }}
\varepsilon^{\kappa^{\prime }}$.
\end{proof}

\begin{proof}[Proof of Lemma \protect\ref{lem:BesovBounds}]
We restrict ourselves to proof \eqref{eq:ReconstructionEstimate1}, the other
three inequalities follow by basically the same arguments. We fix a wavelet
basis $\phi_y=\phi(\cdot-y),\,y\in \mathbb{Z}$, $\psi^j_y=2^{j/2}\,\psi(2^j(
\cdot-y)),\,j \geq 0,\,y\in 2^{-j}\mathbb{Z}$ and use in the following the
notation $\phi_y=2^{j/2}\phi(2^j(\cdot-y)),\,j\geq 0,\,y\in 2^{-j}\mathbb{Z}$%
. Within this basis we can express the $\mathcal{B}^{\beta}_{1,\infty}$
regularity of $\varphi$ by 
\begin{align*}
\sum_{y\in \mathbb{Z}} |(\varphi,\phi_y)_{L^2}| +\sup_{j\geq 0} 2^{j\beta }
\sum_{y\in 2^{-j}\mathbb{Z}} 2^{-d j/2} |(\varphi,\psi^j_y)_{L^2}| \lesssim
\|\varphi\|_{\mathcal{B}^{\beta}_{1,\infty}}
\end{align*}

Without loss of generality we can assume that $\lambda=2^{-j_0}$ is dyadic,
so that by scaling 
\begin{align}  \label{eq:proof:BesovBoundsScaling}
\sum_{y\in 2^{-j_0} \mathbb{Z}} |(\varphi^\lambda_s,\phi_y^{j_0})_{L^2}|
+\sup_{j\geq j_0} 2^{(j-j_0)\beta } \sum_{y\in 2^{-j}\mathbb{Z}} 2^{-(j-j_0)
d/2} |(\varphi^\lambda_s,\psi^j_y)_{L^2}| \lesssim 2^{j_0 d/2} \|\varphi\|_{%
\mathcal{B}^{\beta}_{1,\infty}}\,.
\end{align}
We can now rewrite 
\begin{align}
&(\mathcal{R}F-\Pi_s F_s)(\varphi^\lambda_s)=  \notag \\
&\sum_{y\in 2^{-j_0} \mathbb{Z}} (\mathcal{R}F-\Pi_s F_s)(\phi^{j_0}_y)
\cdot (\phi^{j_0}_y,\varphi^\lambda_s)_{L^2} +\sum_{j\geq j_0} \sum_{y\in
2^{-j} \mathbb{Z}} (\mathcal{R}F-\Pi_s F_s)(\psi^j_y) \cdot
(\psi^j_y,\varphi^\lambda_s)_{L^2}  \notag \\
&=\sum_{y\in 2^{-j_0} \mathbb{Z}} (\mathcal{R}F-\Pi_y
F_y)(\phi^{j_0}_y)\,(\phi^{j_0}_y,\varphi^\lambda_s)_{L^2}+\sum_{y\in
2^{-j_0}\mathbb{Z}}
\Pi_y(F_y-\Gamma_{ys}F_s)(\phi^{j_0}_y)\,(\phi^{j_0}_y,\varphi^\lambda_s)
\label{eq:proof:BesovBounds1} \\
&+\sum_{{\scriptsize 
\begin{array}{c}
j\geq j_0, \\ 
y\in 2^{-j}\mathbb{Z}%
\end{array}%
}} (\mathcal{R}F-\Pi_y F_y)(\psi^j_y)\,(\psi^j_y,\varphi^\lambda_s)_{L^2}+
\sum_{{\scriptsize 
\begin{array}{c}
j\geq j_0, \\ 
y\in 2^{-j}\mathbb{Z}%
\end{array}%
}} \Pi_y(F_y-\Gamma_{ys}F_s)(\psi^j_y)\,(\psi^j_y,\varphi^\lambda_s)_{L^2}
\label{eq:proof:BesovBounds2}
\end{align}
Only finite terms in \eqref{eq:proof:BesovBounds1} contribute which all can
be bounded (up to a constant) by $2^{-j_0 \gamma}=\lambda^{\gamma}$.
Moreover 
\begin{align*}
\eqref{eq:proof:BesovBounds2} &\lesssim\sum_{j\geq j_0} 2^{-j \gamma} +
\sum_{j\geq j_0} \sum_{A\ni \alpha <\gamma} 2^{-j\alpha}
2^{-(\gamma-\alpha)j_0} \sum_{y \in 2^{-j}\mathbb{Z}} 2^{jd/2}
|(\varphi_s^\lambda,\psi^j_y)_{L^2}| \\
&\lesssim \sum_{j\geq j_0} 2^{-j \gamma}+ 2^{-\gamma j_0} \sum_{A\ni
\alpha<\gamma }\sum_{j\geq j_0} 2^{-(j-j_0)\alpha} 2^{-(j-j_0)\beta}
\lesssim 2^{-j_0 \gamma}=\lambda^{\gamma}
\end{align*}
where we used $\beta+\alpha>0,\,\alpha\in A$ in the last line.
\end{proof}

\begin{proof}[Proof of Lemma \protect\ref{lem:ReconstructionIdentity}]
Note first that via Taylor's formula it easy to check that for scaled Haar
wavelets $\varphi_s^\lambda$ and $\gamma\in (0,(M+1)H)$ 
\begin{align}  \label{eq:proof:ReconstructionIdentity1}
\mathbb{E} \left[\left|\int\,\varphi_{s}^{\lambda}(t)\,f(\hat{W}(t),t) 
\mathrm{d}  W(t)-\Pi_{s}F\Xi(s)(\varphi_{s}^{\lambda})\right|^{2}\right]%
^{1/2}\lesssim\lambda^{(\gamma-1/2-\kappa)}
\end{align}

uniformly for $s$ in compact sets. The same argument as in the proof of
Lemma \ref{lem:BesovBounds} then implies that %
\eqref{eq:proof:ReconstructionIdentity1} actually holds for compactly
supported smooth function $\varphi$ (or even compactly supported functions
in $\mathcal{B}^{\beta}_{1,\infty}(\mathbb{R} ^d)$). Proceeding now as in 
\cite{Hai14} we choose test functions $\eta,\psi\in
C_{c}^{\infty}$ with $\eta$ even and $\mathrm{supp\,\eta\subseteq
B(0,1),\,\int\eta(t)\,\mathrm{d}  t=1}$. We then obtain for $%
\psi^{\delta}(s)=\langle\psi,\eta_{s}^{\delta}\rangle$ 
\begin{align*}
&\mathbb{E} \left[|\mathcal{R}F\Xi(\psi^{\delta})-\int\psi^{\delta}(t)\,f(%
\hat{W}(t),t))\mathrm{d}  W(t)|^{2}\right]^{1/2} \\
&=\mathbb{E} \left[\left|\int\mathrm{d}  x\,\psi(x)\,\left(\mathcal{R}%
F\Xi(\eta^{\delta}_x)-\int\,\eta^{\delta}_x(t)\,f(\hat{W}(t),t))\mathrm{d} 
W(t)\right)\right|^{2}\right]^{1/2} \\
& \lesssim\int\mathrm{d}  x\,\psi^{2}(x)\,\delta^{\gamma-1/2-\kappa}\overset{%
\delta\rightarrow 0}{\rightarrow}0
\end{align*}
where we included a term $\Pi_x \Xi F(x)$ in the second step.  It remains to
note that 
\begin{equation*}
\int\psi^{\delta}(t)\,f(\hat{W}(t),t)) \mathrm{d}  W(t)\overset{%
\delta\rightarrow0}{\rightarrow}\int\psi(t)\,f(\hat{W}(t),t) \mathrm{d} 
W(t) 
\end{equation*}

in $L^{2}(\mathbb{P})$ and further $\mathcal{R}F\Xi(\psi^{\delta})\rightarrow%
\mathcal{R}F\Xi(\psi)$ a.s. and thus in $L^{2}(\mathbb{P})$. Putting
everything together we obtain 
\begin{equation*}
\mathbb{E} \left[|\mathcal{R}F\Xi(\psi)-\int\psi(t)\,f(\hat{W}(t),t)\mathrm{d%
}  W(t)|^{2}\right]=0 
\end{equation*}
which implies the first statement. For the second identity we proceed in the
same way but making use of Lemma \ref{lem:ApproximateIto}.
\end{proof}

\begin{lemma}
\label{lem:ApproximateIto}For $F\in L^{2}(\mathbb{P}\times\mathrm{Leb)}$ we
have 
\begin{equation*}
\mathbb{E} \left[\left|\int F(t) \mathrm{d}  W^\varepsilon(t)\right|^{2}%
\right]\lesssim\int\mathbb{E} \left[\left|F(t)\right|^{2}\right] \mathrm{d} 
t 
\end{equation*}
\end{lemma}

\begin{proof}
As a consequence of Definition~\ref{def:Dirac}, we have $\int |
  \delta^\eps(x,y)\mathrm{d}x|$ is bounded uniformly in $\eps$ and $y$.
We can, therefore, normalize $|\delta^\eps(\cdot, r)|$ to a probability density and
apply It\^o's isometry and Jensen's inequality to 
\begin{align*}
\int F(t) \mathrm{d}  W^\varepsilon(t) =
\int_{0}^{\infty}\int_{0}^{\infty}\,\delta^{\varepsilon}(t,r)\,F(t) \mathrm{d%
}  t \,\mathrm{d} 
W(r) 
.
\end{align*}
\end{proof}

\section{Large deviations proofs} \label{app:LD}

\begin{proof}[Proof of Lemma \protect\ref{lem:Pih}]
The fact that $\Pi^h$ satisfies the algebraic constraints is obvious so we
focus on the analytic ones. The Sobolev embedding $L^2 \subset C^{-1/2}$
yields that $\Pi \Xi$, $\Pi \bar{\Xi}$ satisfy the right bounds. Noting that
(by e.g. \cite[section 3.1]{SKM}) $\left\|K \ast h \right\|_{C^H} \leq C
\|h\|_{C^{-1/2}}$ gives the bound for $\Pi \mathcal{I}(\Xi)^m$. Finally, we
note that using Cauchy-Schwarz's inequality 
\begin{eqnarray*}
\left|\left\langle \Pi_t \Xi \mathcal{I}(\Xi)^m, \phi^\lambda_x \right
\rangle \right| &=& \left|\int h_1(s) \left(K\ast h_1(s)-K\ast
h_1(t)\right)^m \phi^\lambda_x(s) ds \right| \\
&\leq& \left( \sup_{|s-t|\leq \lambda} \left|K\ast h_1(s)-K\ast
h_1(t)\right|\right)^m \|h_1\|_{L^2} \|\ \|\phi^\lambda_x\|_{L^2} \\
&\lesssim& \lambda^{mH-1/2}.
\end{eqnarray*}
The inequality for $\Pi \bar{\Xi}\mathcal{I}(\Xi)^m$ follows in the same
way, and the bounds for $\Gamma$ also follow.

Continuity is $h$ is proved by similar arguments which we leave to the
reader.
\end{proof}

\begin{proof}[Proof of Theorem \protect\ref{thm:LD}]
The theorem is a special case of results in Hairer-Weber \cite{HW15} for
large deviations of Banach-valued Gaussian polynomials. Let us recall the
setting.

Let $(B, \mathcal{H},\mu)$ be an abstract Wiener space and let us call $\xi$
the associated $B$-valued Gaussian random variable, and $(e_i)$ an
orthonormal basis of $\mathcal{H}$ with $e_i \in B^\ast$. For a multi-index $%
\alpha \in \mathbb{N}^\mathbb{N}$ with only finitely many nonzero entries,
define $H_\alpha(\xi) = \Pi_{i \geq 0} H_{\alpha_i}(\langle \xi,e_i\rangle)$%
, where the $H_n$, $n\geq 0$ are the usual Hermite polynomials. For a given
Banach space $E$, the homogeneous Wiener chaos $\mathcal{H}^{(k)}(E)$ is
defined as the closure in $L^2(E,\mu)$ of the linear space generated by
elements of the form 
\begin{equation*}
H_{\alpha}(\xi) y, \;\;\; |\alpha|=k, \,y \in E.
\end{equation*}
Also define the inhomogeneous Wiener chaos $\mathcal{H}^{k}(E) =
\oplus_{i=0}^k \mathcal{H}^{(i)}(E)$. Finally for $\Psi \in \mathcal{H}{(k)}%
(E)$ and $h \in \mathcal{H}$ we define $\Psi^{hom}(h) = \int \Psi(\xi + h)
\mu(d\xi)$, and for $\Psi = \sum_{i \leq k} \Psi_i \in \mathcal{H}^k(E)$, we
let $\Psi^{hom}=(\Psi_k)^{hom}$.

Now let $E = \oplus_{\tau \in \mathcal{W}} E_{\tau}$ where $\mathcal{W}$ is
a finite set and each $E_\tau$ is a separable Banach space. Let $\Psi =
\oplus_{\tau \in \mathcal{W}} \Psi_{\tau}$ be a random variable such that
each $\Psi_{\tau}$ is in $\mathcal{H}^{K_\tau}(E_\tau)$. Letting $%
\Psi^{\delta} = \oplus_{\tau} \delta^{K_\tau} \Psi_\tau$, Theorem 3.5 in 
\cite{HW15} states that $\Psi^\delta$ satisfies a LDP with rate function
given by 
\begin{equation*}
I(\Psi) = \inf \left\{ 1/2\|h\|_{\mathcal{H}}^2, \;\;\; \Psi = \oplus_{\tau
\in \mathcal{W}} \Psi_\tau^{hom}(h) \right\}.
\end{equation*}

In our case, we apply this result with $\mathcal{W} = \left\{ \Xi \mathcal{I}%
(\Xi)^m, \bar{\Xi} \mathcal{I}(\Xi)^m, 0 \leq m \leq M \right\}$ and each $%
E_\tau$ is the closure of smooth functions $(t,s) \mapsto \Pi_t \tau(s)$
under the norms 
\begin{equation*}
\| \Pi \tau \| = \sup_{\lambda, t, \phi} \lambda^{-|\tau|}
\left|\left\langle \Pi_t \tau, \phi_t^\lambda \right\rangle \right|.
\end{equation*}

In order to obtain Theorem \ref{thm:LD}, it suffices then to identify $(\Pi
\tau)^{hom}(h)$ which is done in the following lemma.
\end{proof}

\begin{lemma}
For each $\tau \in \mathcal{W}$ and $h \in \mathcal{H}$, $(\Pi
\tau)^{hom}(h) = \Pi^h \tau$.
\end{lemma}

\begin{proof}
We prove it for $\tau = \Xi \mathcal{I}(\Xi)^m$, the other cases are
similar. Note that $\Psi \mapsto \Psi^{hom}(h)$ is continuous from $\mathcal{%
H}^{k}$ to $\mathbb{R}$ for fixed $h$ (by an application of the
Cameron-Martin formula), and so 
it is
enough to prove that 
\begin{equation}  \label{eq:Pihom}
\lim_{\varepsilon \to 0} \left(\hat{\Pi}^\varepsilon \tau\right)^{hom}(h) =
\Pi^h \tau,
\end{equation}
where $\hat{\Pi}^\varepsilon$ corresponds to the (renormalized model) with
piecewise linear approximation of $\xi$.  For any test function $\varphi$,
by definition one has 
\begin{equation*}
\left\langle \Pi^\varepsilon_t \tau, \varphi \right\rangle = - \left\langle
I^\varepsilon, \varphi^{\prime }\right\rangle, 
\end{equation*}
where 
\begin{equation*}
I^\varepsilon(s) =\int_t^s \left((K \ast \xi^\varepsilon)(u) - (K \ast
\xi^\varepsilon)(t)\right)^m \xi^\varepsilon(u) du - C_\varepsilon
R_{m}^\varepsilon,
\end{equation*}
where $R_m^\varepsilon$ is a renormalization term which is valued in the
lower-order chaos $\mathcal{H}^{m}$, so that by definition it does not play
a role in the value of $(\Pi \tau)^{hom}$. Now note that if $\Phi$ is a
Wiener polynomial whose leading order term is given by $\Pi_{i=1}^k
\left\langle \xi, g_i \right\rangle$ (where the $g_i$ are in $\mathcal{H}$)
then $\Phi^{hom}(h) = \Pi_{i=1}^k \left\langle h, g_i \right\rangle$. In our
case this means that 
\begin{equation*}
(I^\varepsilon)^{hom}(s) =\int_t^s \left((K \ast h_1^\varepsilon)(u) - (K
\ast h_1^\varepsilon)(t)\right)^m h_1^\varepsilon(u) du
\end{equation*}
where $h_1^\varepsilon = \rho^\varepsilon \ast h_1$. In other words we have $%
(\hat{\Pi}^\varepsilon \tau)^{hom} = \Pi ^{h^\varepsilon}_\tau$, and by
continuity of $h\mapsto \Pi^h$ we obtain \eqref{eq:Pihom}.
\end{proof}

\section{Proofs of Section \ref{sec:Volt}}   \label{app:RH}

The proof of Theorem \ref{thm:Volt} follows from the estimates in the lemmas
below, using the standard procedure of taking a time horizon $T$ small
enough to obtain a contraction and then iterating. Note that due to global
boundedness of $u$, $v$ the estimates are uniform in the starting point $z$,
so that one obtains global existence (unlike the typical situation in SPDE
where the theory only gives local in time existence).

By translating $u$ and $v$ we can assume w.l.o.g. that the initial condition
is $z=0$. Then the solution will take value in $\mathcal{D}%
^\gamma_{0,T}(\Gamma) := \left\{~F \in \mathcal{D}^\gamma_T(\Gamma), \;\;
F(0) = 0.\right\}$.



\begin{lemma}
For each $F$ and $\tilde{F}$ in $\mathcal{D}^\gamma_{0,T}(\mathcal{T})$ for
the respective models $(\Pi,\Gamma)$ and $(\tilde{\Pi},\tilde{\Gamma})$, and
for each $\gamma <1 $ 
and $T$ $\in$ $(0,1]$, one has 
\begin{equation*}
{\vert\kern-0.25ex\vert\kern-0.25ex\vert \mathcal{K} F; \mathcal{K} \tilde{F}
\vert\kern-0.25ex\vert\kern-0.25ex\vert}_{\mathcal{D}^\gamma_T(\Gamma), 
\mathcal{D}^\gamma_T(\tilde{\Gamma})} \lesssim T^{\eta} {\vert\kern%
-0.25ex\vert\kern-0.25ex\vert F;\tilde{F} \vert\kern-0.25ex\vert\kern%
-0.25ex\vert}_{\mathcal{D}^{\gamma+|\Xi|}_T(\Gamma), \mathcal{D}%
^{\gamma+|\Xi|}_T(\tilde{\Gamma})}
\end{equation*}
for some $\eta>0$, the proportionality constants depending only on $\gamma$
and the norms of $(\Pi,\Gamma)$ and $(\tilde{\Pi},\tilde{\Gamma})$.
\end{lemma}

\begin{proof}
($\gamma <1$ avoids the appearance of any polynomial terms, present in \cite[Sec. 5]{Hai14} but not in our case.)
Note that if $F$ belongs to $\mathcal{D}_{0,T}^\gamma$ so does $\mathcal{K} F
$. Since $K$ is a regularizing kernel of order $\beta := \frac{1}{2}+H$ in the sense
of \cite{Hai14}, it follows along the lines of \cite[Sec. 5]{Hai14} that 
\begin{equation*}
{\vert\kern-0.25ex\vert\kern-0.25ex\vert \mathcal{K} F; \mathcal{K} \tilde{F}
\vert\kern-0.25ex\vert\kern-0.25ex\vert}_{\mathcal{D}^{\bar{\gamma}%
}_T(\Gamma), \mathcal{D}^{\bar{\gamma}}_T(\tilde{\Gamma})} \lesssim {\vert%
\kern-0.25ex\vert\kern-0.25ex\vert F;\tilde{F} \vert\kern-0.25ex\vert\kern%
-0.25ex\vert}_{\mathcal{D}^{\gamma+|\Xi|}_T(\Gamma), \mathcal{D}%
^{\gamma+|\Xi|}_T(\tilde{\Gamma})}
\end{equation*}
where we pick  $\bar{\gamma} \in (\gamma, 1)$ such that $\bar{\gamma} \leq \gamma+|\Xi|
+ \beta = \gamma + H - \kappa$. On the other hand, it is clear from
the definition of ${\vert\kern-0.25ex\vert\kern-0.25ex\vert \cdot;\cdot \vert%
\kern-0.25ex\vert\kern-0.25ex\vert}$ that since $\mathcal{K} F$ and $%
\mathcal{K} \tilde{F}$ vanish at $t=0$ it holds that 
\begin{equation*}
{\vert\kern-0.25ex\vert\kern-0.25ex\vert \mathcal{K} F; \mathcal{K} \tilde{F}
\vert\kern-0.25ex\vert\kern-0.25ex\vert}_{\mathcal{D}^{{\gamma}}_T(\Gamma), 
\mathcal{D}^{{\gamma}}_T(\tilde{\Gamma})} \lesssim T^\eta {\vert\kern%
-0.25ex\vert\kern-0.25ex\vert \mathcal{K} F; \mathcal{K} \tilde{F} \vert\kern%
-0.25ex\vert\kern-0.25ex\vert}_{\mathcal{D}^{\bar{\gamma}}_T(\Gamma), 
\mathcal{D}^{\bar{\gamma}}_T(\tilde{\Gamma})}
\end{equation*}
for $\eta = \bar{\gamma} - \gamma$.
\end{proof}

\begin{lemma}
Let $G$ (resp. $\tilde{G}$) be the composition operator corresponding to $g$
(resp. $\tilde{g}$) $\in$ $C^{M+2}_b$. Then one has 
\begin{equation*}
{\vert\kern-0.25ex\vert\kern-0.25ex\vert G(F); \tilde{G}(\tilde{F}) \vert%
\kern-0.25ex\vert\kern-0.25ex\vert}_{\mathcal{D}^{{\gamma}}_T(\Gamma), 
\mathcal{D}^{{\gamma}}_T(\tilde{\Gamma})} \lesssim \|G-\tilde{G}\|_{C^{M+2}}
+{\vert\kern-0.25ex\vert\kern-0.25ex\vert F; \tilde{F} \vert\kern-0.25ex\vert%
\kern-0.25ex\vert}_{\mathcal{D}^{{\gamma}}_T(\Gamma), \mathcal{D}^{{\gamma}%
}_T(\tilde{\Gamma})} 
\end{equation*}
the proportionality constants depending only on $\gamma$ and the norms of $%
(\Pi,\Gamma)$, $(\tilde{\Pi},\tilde{\Gamma})$, $F$, $\tilde{F}$, $g$, $%
\tilde{g}$.
\end{lemma}

\begin{proof}
This follows from the estimate in \cite[Theorem 4.16]{Hai14}.
The joint continuity is not stated there but is clear from the triangle
inequality.
\end{proof}

\end{document}